\documentclass[acmsmall,screen]{acmart}
\usepackage{enumitem}
\usepackage{array}
\newtheorem{definition}{Definition}[section]
\newtheorem{lemma}{Lemma}[section]
\newtheorem{theorem}{Theorem}[section]

\usepackage{soul}

\newcommand{\blue}[1]{{\textcolor[rgb]{0,0,1}{#1}}}

\newcommand{\gray}[1]{{\textcolor[rgb]{0.5,0.5,0.5}{#1}}}

\AtBeginDocument{%
 \providecommand\BibTeX{{%
  \normalfont B\kern-0.5em{\scshape i\kern-0.25em b}\kern-0.8em\TeX}}}

\setcopyright{acmcopyright}
\copyrightyear{2018}
\acmYear{2024}
\acmDOI{XXXXXXX.XXXXXXX}

\begin{document}

\title{Half a Century of Distributed Byzantine Fault-Tolerant Consensus: Design Principles and Evolutionary Pathways}

\author{Huanyu Wu}
\email{h.wu.3@research.gla.ac.uk}
\orcid{0000-0003-2067-6086}
\affiliation{%
 \institution{James Watt School of Engineering, University of Glasgow}
 \streetaddress{G12 8QQ}
 \city{Glasgow}
 \state{Scotland}
 \country{UK}
}
\author{Chentao Yue}
\email{chentao.yue@sydney.edu.au}
\orcid{0000-0001-5877-2319}
\affiliation{%
 \institution{School of Electrical and Computer Engineering, The University of Sydney}
 \streetaddress{NSW 2006}
 \city{Sydney}
 \state{NSW}
 \country{Australia}
}

\author{Yixuan Fan}
\email{y.fan.3@research.gla.ac.uk}
\orcid{0000-0003-2865-1422}
\affiliation{%
 \institution{James Watt School of Engineering, University of Glasgow}
 \streetaddress{G12 8QQ}
 \city{Glasgow}
 \state{Scotland}
 \country{UK}
 \postcode{G12 8QQ}
}

\author{Yonghui Li}
\email{yonghui.li@sydney.edu.au}
\orcid{0000-0001-7702-1123}
\affiliation{%
 \institution{School of Electrical and Computer Engineering, The University of Sydney}
 \streetaddress{NSW 2006}
 \city{Sydney}
 \state{NSW}
 \country{Australia}
}

\author{David Flynn}
\email{David.Flynn@glasgow.ac.uk}
\orcid{0000-0002-1024-3618}
\affiliation{%
 \institution{James Watt School of Engineering, University of Glasgow}
 \streetaddress{G12 8QQ}
 \city{Glasgow}
 \state{Scotland}
 \country{UK}
}

\author{Lei Zhang}
\email{lei.zhang@glasgow.ac.uk}
\orcid{0000-0002-4767-3849}
\affiliation{%
 \institution{James Watt School of Engineering, University of Glasgow}
 \streetaddress{G12 8QQ}
 \city{Glasgow}
 \state{Scotland}
 \country{UK}
 \postcode{G12 8QQ}
}

\renewcommand{\shortauthors}{Huanyu Wu, et al.}

\begin{abstract}

The concept of distributed consensus originated in the 1970s and gained widespread attention following Leslie Lamport's influential publication on the Byzantine Generals Problem in the 1980s. Over the past five decades, distributed consensus has become an extensively researched field. Practical Byzantine Fault Tolerance (PBFT) has emerged as a prominent and widely adopted solution due to its conceptual clarity, effectiveness, and resilience to arbitrary failures. However, PBFT does not universally address all scenarios, highlighting the necessity of developing a comprehensive understanding of the history, evolution, and foundational principles of distributed consensus.
This article systematically reviews the historical evolution and foundational principles of distributed consensus, examining pivotal advancements including fault-tolerant state machine replication (SMR), consensus protocols in partially synchronous and asynchronous networks, and recent innovations in Directed Acyclic Graph (DAG)-based consensus mechanisms. We further analyse the core design rationales, essential components, and underlying primitives across various distributed fault-tolerant protocols. The relationship between BFT consensus mechanisms and their applications in environments requiring robust resilience against adversarial faults is also explored. Finally, we discuss emerging research areas and challenges, such as consensus for wireless and blockchain scenarios, highlighting potential future developments. This comprehensive overview offers valuable insights to inform the design, optimisation, and implementation of distributed consensus systems across multiple application scenarios.
\end{abstract}



\received{20 February 2007}
\received[revised]{12 March 2009}
\received[accepted]{5 June 2009}

\maketitle

\section{Introduction}

The increasing digitalisation and interconnection of critical systems, networks, and markets are transforming industries globally. Modern power grids, connected autonomous systems, financial platforms, cloud services, and communication networks, etc., all rely on fault-tolerant architectures to maintain operational reliability, high service quality, and even facilitate goals such as drive decarbonisation. However, the complexity of these adaptive systems, which involve dynamic and rapidly evolving infrastructures, introduces significant challenges, particularly in ensuring reliable and consistent operation. As these systems become more coupled and distributed, achieving fault tolerance and resilience becomes crucial to their success.
Fault tolerance is the ability of a distributed system to continue operating correctly even when some components fail or behave arbitrarily. This capability is mission-critical as industries become increasingly digital and interconnected. As systems grow distributed and dynamic, early detection and mitigation of faults are vital to prevent minor glitches from cascading into major failures. In essence, fault tolerance forms the backbone that allows distributed systems to deliver continuous, reliable service in our always-on world.
One of the most widely researched and applied concepts in this domain is Byzantine Fault Tolerance (BFT). BFT ensures that distributed systems can achieve agreement even when some components or processes fail or act maliciously, which is widely deployed in high-stakes applications like financial systems, cloud infrastructure, and critical communication networks.

The nascent stage of BFT dates back to the late 1960s, when engineers were considering using computers in mission-critical space vehicles and aircraft control systems, which comprised a central multiprocessor with duplicated processors for error detection, dedicated local processors, and multiplexed buses connecting them~\cite{hopkins1971fault}. However, a challenge arose regarding how to ensure that all the multiple processes could make consistent decisions on a consistent view, where each process is an instance that independently executes a specific task or a set of tasks. Subsequently, in 1970s, NASA's Software Implemented Fault Tolerant Project (SIFT) sought to develop resilient systems for aircraft control~\cite{hopkins1971fault}. Through the project, in 1982, Leslie Lamport's introduction of the Byzantine generals problem~\cite{lamport1982byzantine} highlighted the challenge of achieving consensus in the presence of faulty or adversarial nodes. This problem has since been a focal point of research, with numerous protocols proposed to address different failure models in distributed systems.
One of the key paradigms to emerge from this research is State Machine Replication (SMR), which ensures consistency across replicas of a distributed service by executing the same sequence of operations~\cite{lamport1979make}. Unlike the one-time consensus achieved in the Byzantine agreement, SMR offers a continuous and reliable service that remains consistent across multiple replicas over time. Over the decades, various fault-tolerant SMR solutions, such as Crash Fault-Tolerant (CFT) and Byzantine Fault-Tolerant (BFT) protocols, have been developed. Notable examples include Paxos\cite{lamport1998part}, Raft\cite{ongaro2014search}, and PBFT~\cite{castro1999practical}, which have had a profound impact on modern distributed systems.

In particular, PBFT marked a turning point, demonstrating the practicality of BFT protocols in real-world systems. Subsequent developments have sought to optimize BFT performance, especially in environments where faults are rare or minimal. These advancements include quorum-related hybrid solutions like Q/U \cite{abd2005fault}, HQ \cite{cowling2006hq}, and Zyzzyva \cite{kotla2010zyzzyva}. More recently, innovations in BFT protocols have explored the use of trusted hardware, certificate-based solutions, and parallel BFT approaches \cite{chun2007attested, cloutier2000diapm, correia2004tolerate, hotstuff2019yin}.
The rise of blockchain technologies, especially with the advent of Bitcoin in 2008~\cite{nakamoto2008bitcoin}, brought renewed attention to distributed consensus protocols. Although Bitcoin primarily employs Proof-of-Work (PoW), it is essentially built upon the principles of SMR. This has led to an explosion of research into blockchain-oriented fault-tolerant consensus, where BFT protocols like BFT-SMaRt \cite{bessani2014state} and HotStuff \cite{hotstuff2019yin} are adapted to meet the unique needs of decentralised ledgers. Asynchronous BFT solutions such as Honey Badger BFT \cite{miller2016honey} and BEAT \cite{duan2018beat}, as well as innovative DAG-based protocols like DAG-Rider \cite{keidar2021all} and Bullshark \cite{spiegelman2022bullshark}, have further expanded the potential applications and performance of BFT in blockchain and other distributed systems.


To offer a comprehensive perspective, this article presents a critical review of BFT consensus mechanisms in distributed systems, tracing its evolution from foundational principles to its advanced applications in blockchain and beyond. This review emphasizes the analysis of strategic design and verification of BFT-based mechanisms and blockchain architectures. We begin by examining the fault tolerance challenges in highly distributed and interdependent systems origins from the early years, underscoring the need for resilient consensus mechanisms and the distributed consensus primitives developed since then. Next, we analyse key developments in BFT, including state machine replication and prominent protocols such as PBFT, before expanding into more recent adaptations for both partially and fully asynchronous environments. Particular attention is given to the integration of BFT in blockchain, where novel topologies such as tree and DAG-based consensus models are emerging as promising alternatives for scalability and efficiency. Finally, we review the evolving field of wireless consensus, which is attracting increasing research attention due to the rising adoption of the Internet of Things and intelligent wirelessly connected vehicle systems.

Beyond surveying existing work, this review aims to provide deeper insights into the future design and strategic development of BFT-driven systems. We examine key research directions, including the optimisation and refinement of BFT protocols, the design of adaptive consensus mechanisms tailored for diverse deployment environments, and hybrid approaches that integrate classical BFT with external components such as cryptographic techniques and trusted hardware. Additionally, we explore the potential for cross-domain integration, highlighting emerging applications including blockchain use cases to enhance scalability, security, and performance in distributed systems.

\textbf{Rationale for Another Review on Distributed Fault-Tolerant Consensus.}
Several existing surveys have examined topics related to Byzantine Fault-Tolerant (BFT) consensus, particularly within the context of blockchain systems~\cite{cachin2017blockchain,bano2019sok,gramoli2020blockchain,natoli2019deconstructing,xiao2020survey,wang2022bft,zhang2024reaching}. In addition, traditional fault-tolerant state machine replication (SMR) has been thoroughly reviewed in works such as~\cite{schneider1990implementing,Correia2011byzantine,Distler2021byzantine,zhang2024reaching}, which discuss its respective advantages and limitations. A comparison of existing surveys is presented in Table~\ref{tab:existing}.
However, most of these surveys primarily focus on listing or classifying various consensus protocols and providing performance or complexity comparisons. They often lack a cohesive framework that captures the evolution of distributed fault-tolerant consensus and do not provide a comprehensive or accessible explanation of the underlying principles that guide the design of these protocols. Our review of recent works on BFT finds that many BFT protocols, particularly in permissioned blockchain applications, are based on PBFT. Many researchers default to PBFT as a baseline for optimisation without fully exploring or leveraging alternative BFT approaches or distributed primitives. We attribute this to a lack of comprehensive understanding among researchers, particularly newcomers and those from interdisciplinary fields, regarding the design principles and developmental pathways of distributed Byzantine fault-tolerant consensus.
As a result, they are left without the necessary guidance to understand the foundational principles and primitives, gain insights into protocol selection or improvement, or design the solutions tailored to specific environments and requirements. Furthermore, the role of timing assumptions, which has been pivotal in linking the early evolution of Byzantine fault-tolerant consensus protocols and is still an important consideration in recent research, is frequently overlooked or only superficially introduced. This omission makes it challenging for readers to grasp the evolution and rationale behind certain design choices, as well as the interrelationships and applicability of different consensus mechanisms.

\begin{table}[tbp]
\centering
\captionsetup{font=tiny}
\caption{Comparison of existing surveys}
\begin{tabular}{|p{3cm}|p{10cm}|}
\hline
\multicolumn{1}{|l|}{\textbf{Review Paper}} & \textbf{Contribution} \\ 
\hline
Schneider, 1990 \cite{schneider1990implementing} & An early review that introduces SMR and how to build fault-tolerant services using SMR. \\
\hline
Correia et al., 2011 \cite{Correia2011byzantine} & Introduces FLP impossibility and reviews the technologies to circumvent FLP. \\
\hline
Cachin et al., 2017 \cite{cachin2017blockchain} & Focuses on several particular consensus in permissioned and permissionless blockchain. \\
\hline
Bano et al., 2019 \cite{bano2019sok} & Focuses on blockchain consensus, ranging from classic BFT, PoX to hybrid consensus. \\
\hline
Natoli et al., 2019 \cite{natoli2019deconstructing} & Reviews membership selection, consensus mechanism, and structure in blockchain, and analyses goals and assumptions of the consensus. \\
\hline
Gramoli, 2020 \cite{gramoli2020blockchain} & Analyses several consensus problems tackled by blockchain, formalizes Bitcoin and Ethereum consensus, and analyses several attacks. \\
\hline

Xiao et al., 2020 \cite{xiao2020survey} & Reviews fault-tolerant distributed consensus and blockchain consensus, especially Nakamoto consensus, its variations, and PoX consensus. \\
\hline

Distler et al., 2021 \cite{Distler2021byzantine} & Focuses on BFT-SMR and its components, analyses their practical implications, and discusses the building blocks of BFT. \\
\hline
Wang et al., 2022 \cite{wang2022bft} & Analyses recent BFT protocols in blockchains, categorizes them into partially synchronous, asynchronous, and permissionless. \\
\hline
Zhang et al., 2024 \cite{zhang2024reaching} & Reviews a wide range of BFT consensus, listed by different single solutions, categorized by efficiency, robustness, and availability. \\
\hline
\end{tabular}
\label{tab:existing}
\end{table}

\begin{figure*}[tbp]
 \centering
 \includegraphics[width=6.0in,trim=220 30 100 10,clip]{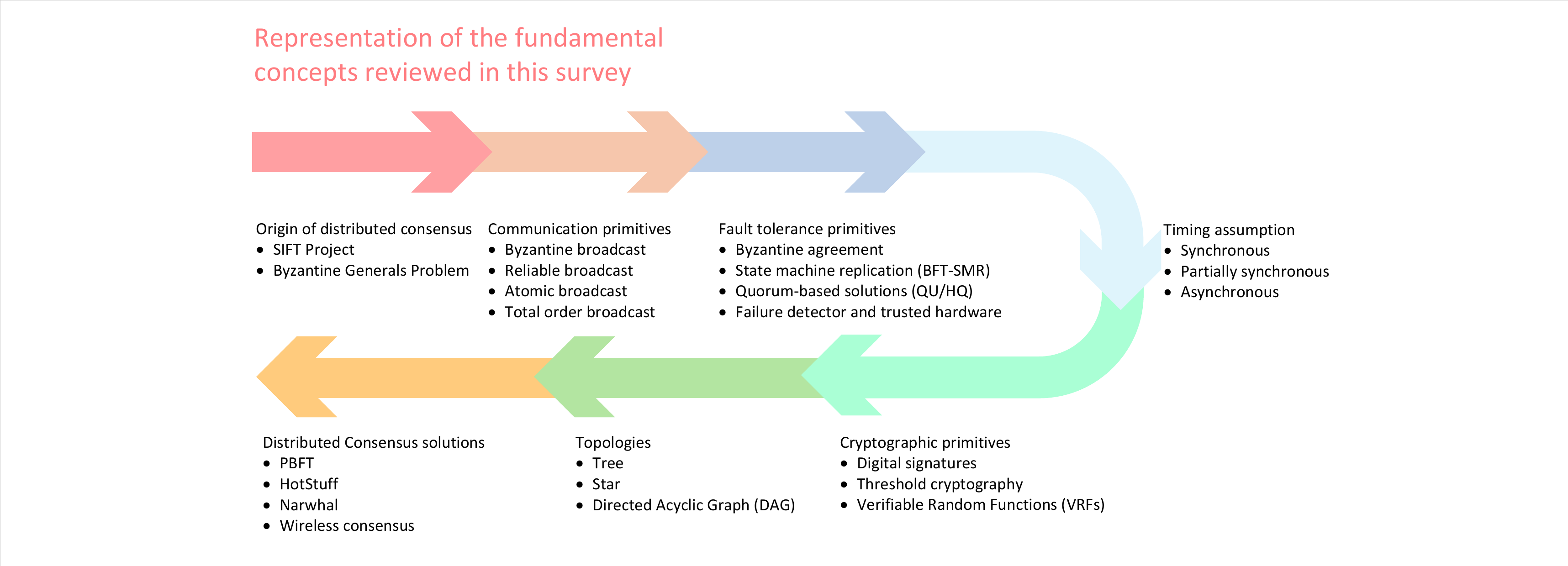}
 \caption{Fundamental concepts reviewed in this survey.}
 \Description{Fundamental concepts reviewed in this survey.}
 \label{fig:represent}
\end{figure*}

Compared with previous review papers, our article makes the following additional contributions:
\begin{itemize}
    \item Our article gives a comprehensive review of the origin and development of fault-tolerant consensus and SMR before and after the advent of blockchain, and it introduces easily-confused concepts such as Byzantine agreement and Byzantine broadcast from the perspective of the historical evolution of consensus, to build a comprehensive view of BFT problem.
    \item This article introduces the important components and primitives of fault-tolerant consensus, such as broadcast primitives and random coin built from threshold cryptography, and explains how they could be used to deal with different system and network assumptions. 
    \item We discussed the consensus designed for different scenarios, including blockchain, and analyses their design rationales, working principles and core elements. The growing novel wireless consensus, probabilistic consensus and consensus for decision making are also introduced. Furthermore, it explores future BFT research issues and potential use cases of blockchains.
\end{itemize}

The fundamental concepts reviewed in this survey is presented in Fig.~\ref{fig:represent}.
Readers can benefit from this article by gaining an understanding of the history and evolution of fault-tolerant consensus, from early solutions to current improvements, and how they can be utilised in practical blockchain systems. In addition, readers will develop insights into the different components of a consensus protocol, the principles guiding their design, and how to design a novel fault-tolerant protocol for specific deployment environments.

\begin{figure*}[tbp]
 \centering
 \includegraphics[width=6.0in,trim=30 40 20 140,clip]{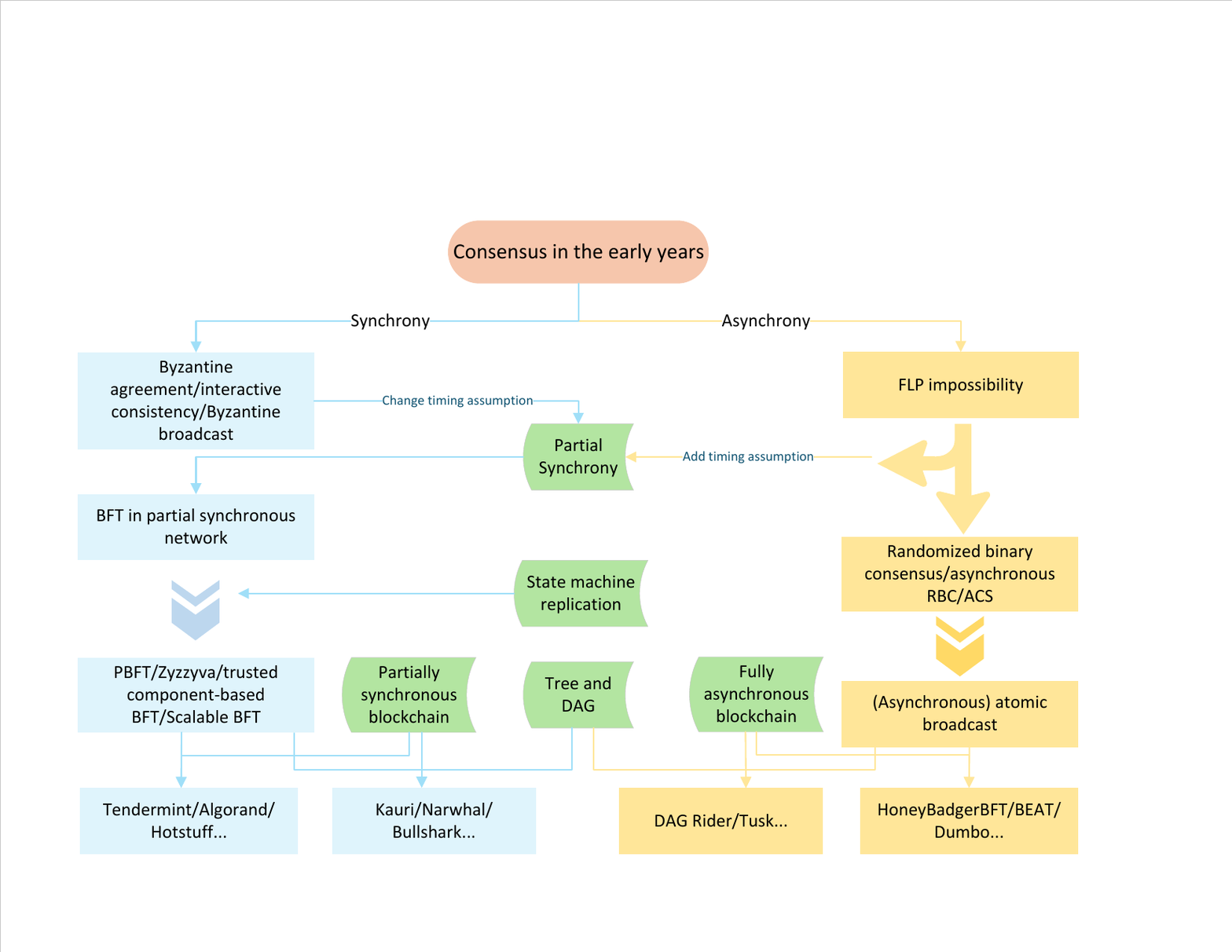}
 \caption{Development path of consensus.}
 \Description{Development path of consensus.}
 \label{fig:devpath}
\end{figure*}
The development path of distributed fault-tolerant consensus that we have identified is illustrated in Fig.~\ref{fig:devpath}. Following this path, the rest of the paper is organized as follows. Section \ref{Sect:AinSyn} begins with the origin of (synchronous) Byzantine agreement and the Byzantine Generals Problem. We then introduce the FLP impossibility result and discuss early solutions that address this issue, along with an introduction to state machine replication. In Section \ref{sect:practicalbft}, we present the milestone PBFT and its variants that operate under a partially synchronous network model, both before and after the advent of blockchain. Section \ref{sect:asynchronousbft} explores practical consensus protocols designed for an asynchronous model, which emerged before blockchain but have regained relevance due to the distinct requirements of blockchain-based systems. In Section \ref{sect:treeandDAG}, we examine recently proposed DAG-based BFT solutions and summarize different fault-tolerant consensus mechanisms. Section \ref{sect:sumBFT} summarises and compares different BFT protocols. Then, building on this review, Section \ref{sect:futureBFT} discusses open challenges and future research directions in BFT consensus, while Section \ref{sect:usecases} explores blockchain use cases in the context of Web3, underscoring the continued significance of BFT research. Finally, Section \ref{sect:conclusion} presents concluding remarks, summarizing the key insights of our study.

\section{Consensus in the Early Years}
\label{sect:consensusearlyyears}
In this section, we provide a comprehensive analysis of the origins of consensus primitives, tracing their development from the early years of distributed systems research. We link these primitives according to their historical progression and compare their relationships and properties. A summary of the reviewed primitives is presented in Table~\ref{tab:consensus_primitives}. We begin by introducing the different timing assumptions in distributed systems (Section~\ref{sect:timing assumption}), followed by the original Byzantine agreement in the ideal synchronous model (Section~\ref{Sect:AinSyn}). Next, we discuss how the FLP impossibility result renders synchronous Byzantine agreement impractical (Section~\ref{Sect:FLP}), and explore how early randomized protocols circumvent this limitation (Section~\ref{sect:earlyrandom}). Finally, we demonstrate how SMR serves as a general abstraction for solutions to the distributed consensus problem (Section~\ref{sect:SMR}).
\begin{table*}[h]
  \centering
  \renewcommand{\arraystretch}{1.3} 
  \caption{Early Distributed Consensus and Primitives}
  \label{tab:consensus_primitives}
  \resizebox{\linewidth}{!}{
  \begin{tabular}{p{3.2cm} p{2.6cm} p{2.6cm} p{5.5cm} p{3.5cm}}
    \toprule
    \textbf{Protocol (Year)} & \textbf{Fault Model} & \textbf{Scalability} & \textbf{Properties} & \textbf{Timing Assumption} \\
    \midrule
    Interactive Consistency ($\sim$ 1980s) & Byzantine, $f < n/3$ & Low ($O(n^3)$) & Agreement, Validity & Synchronous \\
    \midrule
    Byzantine Agreement / Broadcast ($\sim$ 1982) & Byzantine, $f < n/3$ & Low ($O(n^3)$) & Agreement, Validity & Synchronous \\
    
    \midrule
    Ben-Or's Randomized Consensus ($\sim$ 1983) & Byzantine, $f < n/5$; Crash, $f < n/2$ & Low (Exponential worst-case) & Agreement, Validity, Randomized Termination (probabilistic consensus) & Asynchronous \\
    \midrule
    Bracha's Reliable Broadcast ($\sim$ 1985) & Byzantine, $f < n/3$ & Moderate ($O(n^2)$) & Validity, Agreement (all correct nodes deliver the same messages) & Asynchronous \\
    \midrule
    Total Ordering Protocol ($\sim$ 1990s) & Crash / Byzantine Faults & Low-Moderate (varies by implementation) & Agreement, Validity, Termination & Synchronous / Partially Synchronous \\
    \midrule
    Atomic Broadcast ($\sim$ 1996) & Byzantine, $f < n/3$ & Low (Consensus overhead) & Total Order, Agreement, Validity, Integrity & Partially Synchronous / Synchronous \\
    \bottomrule
  \end{tabular}
  }
\end{table*}
\subsection{Timing Assumption}
\label{sect:timing assumption}
The network model in ``Byzantine generals problem'' is a synchronous network. In order to achieve synchrony, Lamport's early work~\cite{malkhi2019concurrency} focused on synchronising the clock of a group of processes. Based on the synchronous network, the problem can then be reduced to ``reaching coordination despite Byzantine failure'', which is also named Byzantine agreement.
However, Fischer, Lynch and Paterson (FLP) proved in 1985 that there
is no deterministic algorithm that can solve Byzantine agreement in asynchronous network models even only one node has crash failure, known as FLP impossibility~\cite{fischer1985impossibility}. To solve this problem in the asynchronous model, several techniques have been utilised. Failure detector~\cite{chandra1996unreliable,malkhi1997unreliable,kihlstrom2003byzantine} and ordering oracle~\cite{pedone2002solving,correia2004tolerate} was introduced to reach consensus in the presence of crash-fault and Byzantine fault in asynchronous network. These oracles and detectors allow the protocols to circumvent FLP because they encompass some degree of synchrony, e.g., synchrony to detect the crash of a single process. However, it is still impossible to achieve the level of encapsulation of the
original crash failure detector model when Byzantine failure exists~\cite{doudou2002encapsulating}. Another way to achieve consensus is to use randomised consensus algorithm which was originally presented by Ben-Or~\cite{ben1983another}. The randomised consensuses is not efficient, and its speed can be significantly boosted by the shared coin~\cite{bracha1987asynchronous} (also known as Common Coin). Nevertheless, the most commonly used method is to relax the asynchrony assumption, i.e., utilising partially synchrony~\cite{dwork1988consensus} or weakly synchrony. These network models are widely considered in many famous protocols, including practical Byzantine fault tolerance (PBFT)~\cite{castro1999practical} and its variants, i.e., PBFT-based systems require weak synchrony to guarantee liveness.

\subsection{Byzantine Agreement in Synchronous Model}
\label{Sect:AinSyn}
In the 1960s, it becomes clear that computers will be used in space vehicles and aircraft control systems, such as those in the Apollo Project. These systems are mission-critical, and processors are duplicated for error detection. To ensure the multiple machines in the system share a consistent view or decision, NASA sponsored the Software implemented Fault Tolerant (SIFT) project~\cite{Wensley1978sift} in the 1970s to build the ultra-reliable computer for critical aircraft control applications.

In SIFT for an aircraft controller, interactive consistency was needed to synchronise different processes, where each sensor supplies its local reading values of an environment of interest as the input\blue{. }
\begin{definition}
  \textbf{Interactive Consistency}: Each process $p_i$ inputs a value $v_i$, and each process outputs a vector $V$ or such values, which satisfies:
  \begin{itemize}
    \item \textbf{Agreement}: Each correct process outputs the same vector $V$ of values.
    \item \textbf{Validity}: If a correct process outputs $V$ and a process $p_i$ is correct and inputs $v_i$, then $V[i]=v_i$.
  \end{itemize}
\end{definition}




After reaching the agreement on an output vector, each process can derive its actions deterministically based on it.

Later, in 1982, Lamport wrote the paper~\cite{lamport1982byzantine} that inspired widespread interest in this problem by introducing the ``Byzantine Generals Problem''. This paper introduced an agreement problem: Imagine several divisions of the Byzantine army planning to attack an encircled city. However, to win the war, the divisions must agree on a battle plan: attack or retreat. Some generals might want to attack, while others might choose retreat; some generals are loyal, and some might be traitors, preventing the loyal generals from reaching an agreement. The generals can only communicate with one another by messenger. An algorithm was needed to address the following conditions:
\begin{itemize}
  \item CD1\blue{:} All loyal generals decide on the same action.
  \item CD2\blue{:} A small number of traitors cannot cause the loyal generals to adopt a bad plan
\end{itemize}
However, Lamport et al. found the second condition difficult to formalise as it requires a predefined ``bad plan''. Therefore, they consider how the generals reach a decision. Byzantine generals problem can thus be abstracted as follows:
A commanding general must send an order to his $n-1$ lieutenant generals such that
\begin{itemize}
  \item CD3\blue{:} All loyal lieutenants obey the same series of orders. (Agreement)
  \item CD4\blue{:} If the commanding general is loyal, then every loyal lieutenant obeys the order he sends. (Validity)
\end{itemize}
The Byzantine generals problem is equivalent to \textit{\textbf{Byzantine broadcast}}, a broadcast primitive where a sender starts with an input value $v$ and others have no input. During the broadcast, the sender conveys its input to all processes in such a way that, at the end of the broadcast, all processes output the same value.

Readers might wonder, what if in CD4 the commanding general is not loyal? In fact, the commanding general does not have to be loyal. When Byzantine broadcast is solved, to address CD1 and CD2, we only have to build $n$, (i.e., the total number of generals) Byzantine broadcast instances in parallel, one for each process. CD1 can be achieved by having all the generals use the Byzantine broadcast method to gather information, and CD2 can be achieved by a robust method, e.g., if the only decision is whether to attack or retreat, (i.e., the decision value is binary), then the final decision can be achieved based on a majority vote among all the orders gathered during the Byzantine broadcast.

Following Byzantine generals problem, we introduce the notion of Byzantine agreement:
\begin{definition}
\label{def:byzantineagreement}  
\textbf{Byzantine Agreement.} In this primitive, every process holds an input value and outputs \gray{again} a single value, and all processes satisfy
\begin{itemize}
  \item \textbf{Agreement}. Each correct process outputs the same value.
  \item \textbf{(All-Same) Validity}. If all the correct processes have the same input value $v$, then they output the same output $v$. Otherwise, all correct processes will still determine on a same value.
\end{itemize}
\end{definition}
Note that the notion of \textit{validity} encompasses many subtle definitions. The original validity condition mentioned above may not be practical when correct processes start with different input values, as this could result in no agreed output. One possible solution is to set the output value to a default value $\bot$ in such cases. Other validity terminologies include requiring that the output value was input by any process, known as Any-Input Validity~\cite{wattenhofer2019blockchain}, or ensuring that the output value is the input from a correct process, referred to as Correct-Input Validity~\cite{wattenhofer2019blockchain}. In some contexts, external-validity~\cite{cachin2001secure} is necessary, which requires that the agreed value is legal according to a global predicate known to all parties and determined by the specific higher-level application. Different subtle definitions of validity are considered based on varying assumptions and research questions.

Lamport proposed two solutions for the (synchronous) Byzantine agreement problem and its variations: an oral message solution and a signature based solution~\cite{lamport1982byzantine}. These solutions incur exponentially large communication costs, hence are also called exponential information gathering~\cite{lynch1996distributed}. After the pioneering work of Lamport, Byzantine agreement solutions with polynomial complexity were developed~\cite{dolev1982polynomial}. Synchronous Byzantine agreement protocols were rarely implemented in real systems, and this problem was considered to be primarily of academic interest until the introduction of PBFT~\cite{castro1999practical}. Therefore, we will not elaborate on these solutions in detail but only introduce the key concepts behind them.

\textbf{$\triangleright$ Remark.} It is important to note that even though the Byzantine Generals Problem does not explicitly state timing assumptions, it inherently necessitates that all generals make their decisions prior to executing their attacking or retreating. In other words, all messages must be delivered before the final decision is made, which essentially assumes a synchronous system. Working on synchronous algorithms for Byzantine agreement made Lamport realise that clock synchronisation among the processes was necessary \cite{malkhi2019concurrency}.

\textbf{\textit{Byzantine agreement}} could be built on top of \textbf{\textit{interactive consistency}}~\cite{malkhi2019concurrency}. In this approach, the processes first run interactive consistency to agree on an input value for each process, after which every process locally runs a deterministic procedure to determine the output according to the validity requirement of Byzantine agreement.
\textbf{\textit{Byzantine broadcast}} could be built on top of \textit{\textbf{Byzantine agreement}}~\cite{malkhi2019concurrency}. Here, the sender process $p_s$ first sends its value to each process, then all processes run \textbf{\textit{Byzantine agreement}} once, whereby each correct process uses the value received from $p_s$ as its input. Intuitively, \textbf{\textit{Byzantine broadcast}} could also be extended to build \textbf{\textit{interactive consistency}}, where each correct process uses \textbf{\textit{Byzantine broadcast}} once to finally construct the same vector. In summary, \textbf{\textit{interactive consistency}}, \textbf{\textit{Byzantine broadcast}}, and \textbf{\textit{Byzantine agreement}} can be implemented based on one another. All of these three notions presented here are designed for a synchronous network; note that there is no explicit termination property, as they all implicitly assume that every message arrives within a bounded time and the processes terminate naturally.

\subsection{FLP Impossibility and Asynchronous Solutions}
\label{Sect:FLP}
In 1985, Fischer, Lynch, and Paterson (FLP)~\cite{fischer1985impossibility,Fischer1982impossible} proved a famous result in the field of distributed systems, stating that in an asynchronous system, a \textbf{\textit{deterministic}} consensus protocol is impossible to achieve even in the simple case of having only one Fail-Stop (crash) process. Readers can refer to the original papers for the formal proof. To help readers understand the theory, we will give a simplified, informal example:

Assume five processes A, B, C, D, and E that need to decide whether to submit a transaction. Every process has a binary initial value, namely submit (1) or rollback (0). They will send their initial values to each other to reach an agreement. Therefore, each process needs to make a decision based on their received messages from other processes. We assume that at most one process can crash, and thus the requirement of agreement is that all of the other four correct processes must get a consistent decision (either 0 or 1). Consider a deterministic algorithm P, in which every process makes the decision based on the majority value. Now, C's initial value is 0, so C broadcasts 0 and it receives 1 from A and B, and 0 from D. Now, both 0 and 1 have two votes, and the value from E will determine the final result. Unfortunately, E crashes, or E is suffering from a network partition and the message from E is delayed forever. In this scenario, C can decide to wait forever until it receives the message from E, but in an asynchronous network, C will not be able to know if E has crashed or is simply slow. As a result, it is possible that C may never make a decision. Alternatively, C might choose a decision based on some predefined deterministic rule; w.l.o.g., we say, C decides on 1. Unfortunately, C crashes just after it makes the decision, and the decision is not sent to the others. Meanwhile, E just recovers from the network partition, and receives 1 from A and B, and 0 from C and D. Now, E's decision depends on its initial value; w.l.o.g., the initial value of E is 0, so E decides on 0, which violates the agreement requirement that all processes must decide on the same value, either 0 or 1.

FLP impossibility is regarded as a milestone in distributed system research. Several solutions have been proposed to avoid the FLP impossibility and enable practical implementations. One possibility is to add failure detectors and oracles~\cite{dwork1988consensus,chandra1996unreliable,malkhi1997unreliable,kihlstrom2003byzantine,pedone2002solving,correia2004tolerate}, but it was later realised that although this approach works fine in a crash-fault environment, it is impossible to achieve the same level of encapsulation as the original crash failure detector model when Byzantine failures exist~\cite{doudou2002encapsulating}. Moreover, these oracles generally require a trusted environment, e.g., trusted hardware. Another possibility is to loosen the network assumptions, moving from asynchrony to partial synchrony (a.k.a. almost-asynchrony, weak synchrony). Partial synchrony was implicitly considered in many famous practical CFT and BFT consensus protocols, including PBFT~\cite{castro1999practical}.
Apart from these solutions, randomisation can be used to achieve agreement in a completely asynchronous network, the idea of which was originally presented by Ben-Or~\cite{ben1983another}. 
In summary, the solutions for asynchronous network guarantee whether the network is synchronous, no honest process will be different from each other (i.e., agreement), but the asynchronous network might make the system halt forever (influence termination), and therefore, they make one of the following trade-offs regarding termination:

\begin{itemize}
\item Termination is guaranteed only if the network is partially synchronous (such as PBFT).
\item Termination is guaranteed with probability 1 (such as Ben-Or's).
\end{itemize}

\subsection{Early Randomised Protocols in Completely Asynchronous Model}
\label{sect:earlyrandom}
In the previous section~\ref{Sect:AinSyn}, all the protocols assume synchronous communication, so termination is ignored, as all processes have a global clock and are expected to follow the protocol until the final bounded step. In other words, all correct processes decide within $r$ rounds, for some previously known constant $r$\cite{bracha1987asynchronous}. However, the situation changes in the asynchronous model. As presented in \ref{Sect:FLP}, it is impossible to find a deterministic solution to address the problem in an asynchronous model, and most solutions loosen the termination condition to avoid the FLP impossibility. In this section, we focus on solutions that apply the loosened termination requirement, guaranteeing termination with probability 1. This means that a protocol may never terminate, but the probability of this occurring is 0, and the expected termination time is finite. Two approaches are widely adopted: The first approach is to add random steps into the protocol \cite{ben1983another}, while the second postulates some probabilistic behaviour about the message system~\cite{bracha1987asynchronous}. Before we proceed, we would like to introduce the notion of consensus, which is easily confused with the term ``Byzantine agreement''.

\begin{definition}
\textbf{Consensus Protocols:} A consensus protocol is a protocol used in a group of processes where each process has an initial value and can propose this value. The protocol satisfies the following properties:
\begin{itemize}
\item \textbf{Agreement:} If two (correct) processes $a$ and $b$ decide values $v_a$ and $v_b$, respectively, then $v_a=v_b$.
\item \textbf{Validity:} If a (correct) process decides a value $v$, then $v$ must be the initial value of some process.
\item \textbf{Termination:} If a (correct) process decides a value, then all correct processes eventually decide a value.
\end{itemize}
\end{definition}

Note that the definition could be slightly different for crash failures and Byzantine failures. In crash failures, all processes follow the code, but in Byzantine failures, processes could behave arbitrarily. 
The correct processes means these properties only hold for correct nodes, because a faulty node especially a Byzantine one can do anything, but some times ``correct'' is not written but known as a implicitly requirement.
Therefore, in the definition for Byzantine failures, we only require that the correct processes satisfy the properties. 

Recall that the definition \ref{def:byzantineagreement} of Byzantine agreement (in a synchronous model does not include a ``termination'' condition.
In recent research, especially in SMR and blockchain, consensus is sometimes required to be ``totally ordered'', which we will introduce later in Section \ref{sect:SMR}.

\subsubsection{Ben-Or's Randomised Consensus}\ 
\label{sect:benor}
\noindent In 1983, Ben-Or proposed the first probabilistic solution to the asynchronous agreement problem~\cite{ben1983another}. The protocol of Ben-Or uses binary values for the initial value and decision value, which are elements of {0,1}. In this work, due to the asynchronous assumption, there are no global time or clocks. Instead, this protocol advances in rounds, and in each process $p$'s local view, a round could be regarded as a local clock for $p$. During the protocol, a message is assumed to eventually be delivered. We now take a brief look at Ben-Or's original protocol with \textit{\textbf{crash failures}}: \\

\noindent\rule[0.25\baselineskip]{\textwidth}{0.5pt}
0. In the initial step, process $p$ sets value $v$ as the proposal at local round number $r_p=1$.\\
1. Phase 1:
\par(a) Broadcast <$phase=1, round=r_p, value=v_p$>.
\par(b) Wait to receive more than $N-t$ different messages <$phase=1, round=r_p, value=*$> from others (where $N$ is the total number of processes).
\par(c) If all the received messages in (b) contain the same value $v$, then $p$ sets $v_p=v$, otherwise it sets $v_p=\bot$.\\
2. Phase 2:
\par(a) Broadcast the message <$phase=2, round=r_p, value=v_p$>.
\par(b) Wait for more than $N-t$ different messages <$phase=2, round=r_p, value=*$> from others.
\par(c) Choose one of the following actions:\par
\quad (i) If all messages contain the same non-$\bot$ value $v$, then decide $v$.\par
\quad (ii) If one of the messages contains a non-$\bot$ value $v$, then accept $v$ by setting $v_p=v$. Otherwise, accept a random estimate by setting $v_p\in{0,1}$ randomly with 50\% probability.\\
3. $r_p = r_p + 1$ and \textbf{goto} Step 1.

\noindent\rule[0.25\baselineskip]{\textwidth}{1pt}

In the following paragraphs, we explain how Ben-Or's protocol work for asynchronous environment.
First, let us check the weak validity in crash-fault model.
\begin{definition}
\label{def:weakvalidity}
\textbf{Weak Validity}. If all processes are honest and start with the same initial value $v$, then the decision value must be $v$.
\end{definition}
It is clear that if all processes are non-faulty and start with the same $v$, they will all send <$phase=1, 1, v$> and will receive more than $N-f$ <$phase=1, 1, v$> messages after one round. Consequently, they will send and receive enough <$phase=2, 1, v$> messages and thus decide on $v$.

Before verifying agreement property, we introduce the notion of Quorum:
\begin{definition}
  \textbf{Quorum, and Quorum system}. Let $P=\{p_1,......,p_n\}$ be a set of processes. A quorum $Q\subseteq P$ is a subset of these processes. A quorum system $S\subset 2^{|P|}$ is a set of quorums such that every two quorums intersect, i.e., for any $Q_1, Q_2\in S$, $Q_1 \cap Q_2\neq \emptyset$.
\end{definition}

Then, we introduce the notion of quorum intersection that is used in Ben-Or's consensus.

\begin{definition}
  \textbf{N/2+1-Quorum Intersection} Let $S_1$ with $|S_1|\geq \frac{N}{2}+1$ and $S_2$ with $|S_2|\geq \frac{N}{2}+1$. Then, there exists at least a correct process in $S_1\cap S_2$.
\end{definition}
Next, we verify the agreement property, which states every process must decide on the same value. The agreement property can be proven by the following claims:
\begin{lemma}
  It is impossible for <$phase2, r_p, 1$> and <$phase2, r_p, 0$> messages to exist simultaneously due to quorum intersection. W.l.o.g., sending <$phase2, r_p, 1$> requires at least $\frac{N}{2}+1$ processes to have the same value, so <$phase2, r_p, 1$> and <$phase2, r_p,0>, 0$> is impossible to exist simultaneously, as this would require $N+2>N$ processes, which is a contradiction.
\end{lemma}

\begin{lemma}
  If some process receives $\frac{N}{2}+1$ messages of the form <$phase2, r_p, v$>, then all processes will receive at least one message of the form <$phase2, r_p, v$>. This is because at least one of the $\frac{N}{2}+1$ processes that sent <$phase2, r_p, v$> is non-faulty, which again follows from quorum intersection.
\end{lemma}

Combining both lemmas, we can conclude that for some round $r$, if a process decides $v$ in step (c)(i), then all other non-faulty processes will decide $v$ within the next round.

Finally, let us check why this protocol satisfies termination. Consider the worst case, where no process decides in round $r$ and only $\bot$ messages are received in phase 2. In this case, processes will accept a random value, so there is a chance, albeit small, that all processes accept the same random value $v$ in the next round, and the protocol terminates (by weak validity). Each round has a nonzero chance of deciding, and although the chance could be small, if the protocol runs for an infinite number of rounds, it will satisfy termination with probability 1. Note that this protocol is not efficient, but it proves that it is possible to find a solution to the consensus problem in a completely asynchronous model. The correctness of Ben-Or's protocol was proven in 1998~\cite{aguilera1998correctness}.
This kind of solution is also called a random coin protocol, as it resembles every node tossing a coin, and the protocol terminating when every node's coin lands on the same side. The performance could be improved if we set every process toss the same coin, called a \textbf{Shared Coin} (a similar notion is \textbf{Common Coin}\cite{shi2020foundations}). The concept of the Shared coin is based on the notion of shared memory\cite{chandra1996polylog}. With the development of cryptography and after the invention of threshold cryptosystems \cite{desmedt1993threshold}, a shared coin can also be built on the threshold signature scheme (TSS) or the verifiable random function (VRF)~\cite{micali1999verifiable}.

\begin{definition}
\textbf{A $(t,n)$-Threshold Signature Scheme (TSS)} is a signature scheme that generates a valid and unforgeable single digital signature only if at least $t$ out of the $n$ participants provide their approvals.
\end{definition}

\begin{definition}
  \textbf{A Verifiable Random Function (VRF)} is a function in which each participant $i$ is equipped with a secret key $sk_i$ and the corresponding public key $pk_i$, and for any input $x$, $VRF_{sk_i}(x)$ returns two values: $hash$ and $proof$. The $hash$ is a fixed-length value determined by the pair ($sk_i$, $x$) and is indistinguishable from random values to anyone who does not know $sk_i$. The $proof$ enables anyone who knows $pk_i$ to verify that the $hash$ corresponds to $x$ without needing to know $sk_i$.
\end{definition}

Ben-Or's protocol also has a Byzantine version if slight modifications are made into the quorum intersection. In the Byzantine environment, Ben-Or's protocol can tolerate $f=\frac{N}{5}$ Byzantine processes. Readers can refer to \cite{ben1983another} for the Byzantine version modification.

\subsubsection{Asynchronous Reliable Broadcast Protocol}\ 

\label{sect:asyn}

\noindent It is widely known that Byzantine consensus can tolerate up to $f=\frac{N}{3}$ Byzantine processes, where $N$ is the total total number of processes. Bracha and Toueg proved this lower bound in \cite{Bracha1983resilient}. However, in Ben-Or's protocol, whether there are asynchronous consensus protocols that can tolerate up to $t<\frac{N}{3}$ remains an open question. Bracha~\cite{bracha1987asynchronous} extended the maximum number of tolerated Byzantine processes to meet the lower bound of $f=\frac{N}{3}$. As introduced in the last paragraph, Ben-Or's Byzantine version modified the quorum intersection to tolerate Byzantine processes. Different from Ben-Or's protocol, which deals with Byzantine processes directly (by quorum intersection), Bracha uses a novel technique to reduce the effect of Byzantine processes, to limit their behaviour. This technique is composed of two parts, a broadcast primitive and a validation mechanism. By utilising the reliable broadcast primitive, a crash process or a Byzantine process can either send no messages or the same message to all correct processes. In other words, the behaviour of a Byzantine process is filtered by the reliable broadcast primitive, thereby reducing the affect of faulty processes. Then the validation mechanism forces faulty processes to send messages that could have been sent by correct processes. We now give the definition of reliable broadcast:
\begin{definition}
\label{def:reliablebroadcast}
  \textbf{Reliable Broadcast (RBC).} A protocol is a reliable broadcast protocol if:
  \begin{itemize}
  \item If process $p$ is correct, then all correct processes agree on the value of the message it broadcasts. (Validity)
  \item If $p$ is faulty, then either all correct processes agree on the same value or none of them agree on any value $p$ broadcasts. (Totality/Agreement)
  \end{itemize}
\end{definition}
We briefly review of Bracha's reliable broadcast.

\noindent\rule[0.25\baselineskip]{\textwidth}{0.5pt}
0. initial. leader (sender) with input $v$ and sends $<$v$>$ to all processes.
\par \textbf{for each} process $p$ (including the leader): echo=ture, vote=true\\
1. Phase 1: on receiving <v> from leader
\par\quad \textbf{if} echo == true: \textbf{send} <echo, v> to all processes and \textbf{set} echo = false\\
2. Phase 2: on receiving <echo, v> from $N-f$ distinct processes:
\par\quad \textbf{if} vote == true: \textbf{send} <vote, v> to all processes and \textbf{set} vote = false\\
3. Phase 3: on receiving <vote, v> from $f+1$ distinct processes:
\par\quad \textbf{if} vote == true: \textbf{send} <vote, v> to all processes and \textbf{set} vote = false\\
4. Phase 4: on receiving <vote, v> from $N-f$ distinct processes: \textbf{deliver} v.\\
\noindent\rule[0.25\baselineskip]{\textwidth}{1pt}

Next, we explain why this reliable broadcast mechanism works. First, we check the validity property. If the sender is correct, it will send $v$ to everyone, then all correct processes will send <echo, v> and every correct one will eventually receive at least $N-f$ echoes for $v$ and at most $f$ echoes for other values. Therefore, all correct processes will send <vote, v> in step 2 and will receive $n-f$ <vote, v> and at most $f$ votes for other values. Hence, all correct processes will eventually deliver $v$.

Next, we verify the totality/agreement property. We first prove by contradiction that no two correct processes will vote for conflicting values: consider two votes sent in phase 2 for $v$ and $v'$ ($v\neq v'$) from processes $a$ and $b$, respectively. Process $a$ must have seen $N-f$ echoes for $v$, and process $b$ must have observed $N-f$ echoes for $v'$. However, this is impossible due to quorum intersection, as $2\times(N-f)\geq f+1$, which means at least $f+1$ processes must have sent an echo to both $v$ and $v'$, which is a contradiction. Now, we know that correct processes only vote for the same value. Therefore, if a correct process delivers a value (i.e., agrees on a value), it must have seen $N-f$ votes, of which at least $N-f-f\geq f+1$ votes come from correct processes. Thus, every correct process will eventually deliver the same value $v$, either due to observing $N-f$ echoes for $v$ or due to seeing $f+1$ votes for $v$ from $f+1$ correct processes. Otherwise, no correct process will ever deliver/agree on any value. Readers can refer to ~\cite{bracha1987asynchronous} for the complete correctness proof.


By utilizing the reliable broadcast primitive, the power of faulty (Byzantine) nodes is restricted. Hence, Bracha's protocol improves upon Ben-Or's protocol by reducing the required number of processes from $N>5f$ to $N>3f$.

\subsection{State Machine Replication}
\label{sect:SMR}
Although the distributed primitives introduced in Sections \ref{sect:timing assumption}--\ref{sect:earlyrandom} can address the distributed consensus problem to some extent, designing a distributed system with consensus remains inherently complex. This complexity arises from the diverse properties of each primitive, which often leads to a lack of a unified design standard. To address this challenge, the concept of State Machine Replication (SMR) has been proposed as a general abstraction and framework for distributed systems. By modeling distributed systems as deterministic state machines, SMR provides a standardized approach to consensus design. This abstraction allows researchers and developers to simplify the intricate process of designing distributed consensus protocols into the replication and synchronisation of state machines, thereby promoting both theoretical understanding and practical implementation.

To begin with, we first consider the most simplified SMR model with crash failure and a synchronous network. The notion of SMR first appeared in a report~\cite{rfc677}. The authors aimed to keep different copies of a database eventually consistent while allowing different copies to update operations independently. The solution to this problem was first introduced in Lamport's paper~\cite{lamport1978time}, which is notable for defining logical time and causality in distributed systems. However, it also introduced the insight that by applying commands in the same order at all the copies in a distributed system, the copies can remain consistent with each other. In other words, each replica contains a local copy of a state machine. If initialised to the same initial state, SMR is achieved when different replicas execute the same commands in the exact same order on their local copies. The algorithm that implements a universal state machine in a distributed system by replication is called the SMR algorithm. The first solution for implementing SMR~\cite{lamport1978time} did not consider failures. Algorithms addressing failures were proposed later. The well-known the crash-fault-tolerant algorithm Paxos~\cite{lamport1998part} and the Byzantine-fault-tolerant algorithm PBFT~\cite{castro1999practical} were introduced successively.

Generally, there are two approaches to implementing an SMR: consensus protocols and total ordering protocols. These two approaches are essentially equivalent. We provide the definition of a total ordering protocol:

\begin{definition}
\textbf{Total Ordering Protocols.} A total ordering protocol is a protocol in which processes broadcast messages to one another, while the following properties are satisfied:
\begin{itemize}
\item \textbf{Agreement.} If two (correct) processes deliver the same two messages, they deliver them in the same order.
\item \textbf{Validity.} If a (correct) process delivers a message, it must be a message sent by some process.
\item \textbf{Termination.} A message delivered by a correct process is eventually delivered to all correct processes.
\end{itemize}
\end{definition}


A total ordering protocol only needs to start once to build an SMR, as it delivers consecutive messages (commands). Once these commands are delivered in order and executed by each process in the same sequence, SMR is built. In contrast, consensus protocols need to be triggered for each time slot in the sequence of commands for SMR. In short, SMR can also be built on ``repeated consensus''. In another word, the total ordering protocol is a service to deal with consecutive requests, while consensus could be used once for each request, and needs to be triggered for every message. A total ordering protocol can be built out of a consensus protocol and vice versa.

We then give the definition of a related broadcast primitive.

\begin{definition}
\label{def:atomic}
\textbf{Atomic Broadcast, a.k.a. Total Order broadcast.} Atomic Broadcast is a type of broadcast primitive \gray{that} where all correct processes receive the same set of messages in the same order while satisfying the following properties:
\begin{itemize}
\item \textbf{Validity.} If a correct process broadcasts a message, all correct processes eventually receive it.
\item \textbf{Agreement.} If a correct process receives a message, then all correct processes eventually receive it.
\item \textbf{Integrity.} A message will be received by each process only once, and only if it was broadcasted previously. 
\item \textbf{Total order.} If two processes $p1$ and $p2$ deliver two messages $m1$ and $m2$ such that one correct process receives $m1$ before $m2$, then every other correct process must receive $m1$ before $m2$.
\end{itemize}
\end{definition}

Note that Atomic Broadcast can be used to realise SMR with an even stronger guarantee by ensuring that all broadcasted messages are included. It is a sender-receiver model broadcast primitive where each process has no an initial value except the sender. Similar to the total ordering protocol, atomic broadcast can be converted to consensus, and conversely, consensus can be reduced to atomic broadcast. 


Although synchronous network assumption simplifies the SMR problem, it is not practical as the assumption is unrealistic for real computing environment. In addition, recall that the FLP impossibility result shows that solving deterministic consensus problems in an asynchronous network is impossible. Therefore, most protocols provide only agreement (safety) during asynchrony and make one of the following trade-offs for termination (liveness):
\begin{itemize}
  \item Guarantees termination only when the network is partial synchronous (the network will sometimes be synchronous anyway, e.g., PBFT).
  \item Guarantees termination with probability 1 (e.g., Ben-Or's consensus).
\end{itemize}

Here we introduce safety and liveness properties~\cite{cs452_consistency}:
\begin{itemize}
\item \label{notion:safety}\textbf{Safety:} A safety property states that something ``bad'' never happens. In a BFT context, this means behaving like a centralised implementation that executes operations atomically one at a time, i.e., the nodes agree on a valid value proposed by one of the nodes. \item \label{notion:liveness}\textbf{Liveness:} A liveness property states that something ``good'' eventually happens. In a BFT context, the nodes eventually reach agreement, i.e., the system must make progress.
\end{itemize}

\section{Practical BFT in Partial Synchronous Network}
\label{sect:practicalbft}
As introduced in Section \ref{Sect:FLP}, no deterministic consensus algorithm exists in a fully asynchronous network, and two trade-off schemes are used to address this problem. In this section, we focus on deterministic protocols in a partial synchronous network. Most practical BFT protocols, including PBFT, provide the safety property consistently, but liveness is only guaranteed when the network is synchronous (after applying the Global Stabilisation Time, GST). These protocols are generally leader-based, meaning that one of the replicas acts as a leader while others are followers. This kind of protocol usually has three sub-protocols: the normal-case agreement protocol, the view change protocol, and the garbage collection protocol (checkpoint/snapshot). A ``view'' is a configuration of replicas, i.e., who is the leader and who are the followers. We focus on the former two protocols as they are important to safety and liveness. If the leader is correct, followers run the normal-case agreement protocol, and when the leader is suspected to be faulty, followers run the view change protocol to elect a new leader. If a sufficient number of followers trigger the view change protocol, a correct leader will be elected.

In this section, we introduce the most widely adopted type of BFT (Section~\ref{sect:mostpopular}) and analyse the operation of the most prominent BFT-SMR protocol, PBFT. We also examine quorum-based BFT solutions and BFT-SMR protocols optimized with speculative execution. Furthermore, we explore various optimisations for BFT-SMR, including scalable BFT solutions (Section~\ref{sect:scalingbft}) and approaches that leverage trusted components (Section \ref{sect:bfttrustcomp}), providing a detailed explanation of their mechanisms and effectiveness.Then in Section \ref{sect:partially sync BFT for blockchain}, we examine recent prominent BFT protocols specifically designed for blockchain scenarios, analysing how their underlying design logic align with traditional BFT approaches and why they are better suited for blockchain applications.

\subsection{The Most Popular Type of BFT: Broadcast and Message-Passing-based BFT-SMR}
\label{sect:mostpopular}
As discussed in Section \ref{sect:SMR}, the original SMR framework, introduced by Lamport \cite{lamport1978time}, was designed to handle crash failures. This framework was subsequently extended to address Byzantine failures, leading to the development of BFT-SMR.
Broadcast-based message-passing BFT generally consist of several one-to-all and all-to-all phases, indicating that in these protocols, every process needs to exchange information with others through a communication channel. Message-passing is the most commonly used form of BFT protocols. This category of BFT protocols is divided into different phases, and each process gathers information through different phases. The follower processes also need to monitor the behaviour of the leader to ensure that current leader is correct.

\subsubsection{BFT-SMR-based Total Ordering Protocol: PBFT}\ 
\label{sect:pbft}

\noindent PBFT is the first and perhaps the most instructive practical BFT protocol to achieve SMR in the presence of Byzantine nodes. Only after its development did people realise BFT is not only a complex problem of academic interest but also a practically solvable issue.

We first examine the cryptography assumptions in PBFT. In this protocol, cryptographic techniques are used to prevent spoofing and replays and to detect corrupted messages. Each PBFT message contains public-key signatures, message authentication codes (MACs), and message digests generated by collision-resistant hash functions. All replicas know each other's public keys to verify signatures, meaning that a pre-set public key infrastructure (PKI) is needed. PBFT assumes that although the adversary is strong enough to coordinate faulty processes, it is bounded by computing power and cannot subvert the cryptographic techniques. This means, for example, the adversary cannot produce a valid signature of a non-faulty node or break the hash (digest) function. In early consensus protocols, there was an assumption that the receiver can distinguish the message sender, e.g., by Media Access Control (MAC) address and IP. Note this assumption is different from the cryptographic assumptions in PBFT. For example, a process in PBFT observing two messages from the same sender with different values but valid signatures can prove the sender is faulty, but identifying the sender does not mean faulty messages equivocal messages can be detected.

Other assumptions employed by PBFT include: replicas are connected by a network, where an adversary can coordinate faulty nodes, delay communication, or delay correct nodes in order to cause the most damage to the replicated service~\cite{castro1999practical}. The network model is mostly asynchronous, meaning the delay between the time $t$ when a message is sent and the time it is received by its destination does not increase indefinitely. However, the asynchronous network also has periods of synchrony to ensure liveness. The bound on the number of faulty nodes is $f\leq \frac{N-1}{3}$, where $N$ is the total number of nodes.

Before exploring the agreement protocol of PBFT, readers should be aware that PBFT is a total ordering protocol in which the safety property is ensured by guaranteeing all non-faulty replicas agree on a total order for the execution of requests despite failures. Recall that a total ordering protocol is equivalent to a consensus protocol but differs in operation when building an SMR. Now, we examine how PBFT achieves agreement on a unique order of requests within a view (when the leader is correct). Our introduction of PBFT is based on the simplified version (without checkpoint)~\cite{wattenhofer2019blockchain}.

\begin{figure}[h]
 \centering
 \includegraphics[height=1.5in]{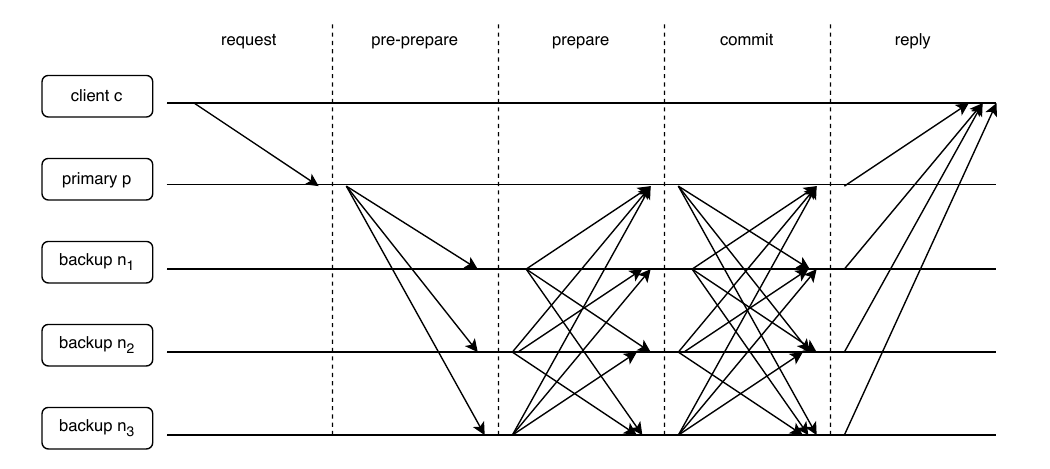}
 \caption{Agreement protocol in PBFT }
 \Description{PBFT.}
\end{figure}

\noindent\rule[0.25\baselineskip]{\textwidth}{0.5pt}
\begin{enumerate}
  \item The leader $l$ receives a request $r$ from a client, picks the next sequence number $s$, and sends \textit{pre-prepare(s,r,l)} to all followers, informing them of the intended execution of the request with the specified sequence number.
  \item For a backup follower $b$, upon receiving \textit{pre-prepare(s,r,l)}, if $l$ is the leader and $b$ has not yet accepted a \textit{pre-prepare} for the sequence $s$, it sends \textit{prepare(s,r,b)} to all other processes to confirm they agree with the leader's suggestion. Once the backup $b$ has sent \textit{prepare}, we say $b$ is pre-prepared for $(s,r)$.
  \item \label{step:2} After being pre-prepared for $(s,r)$, a backup follower $b$ waits until it collects $2f$ \textit{prepare} (including its own) and, together with the \textit{pre-prepare} \gray{one}, they form a \textit{certificate of prepare} with $|\mathrm{certificate\ of\ prepare}|=2f+1$.
  \item Once the certificate is gathered, a process $p$ (leader and follower) sends \textit{commit(s,p)} to all processes.
  \item Process $p$ waits until $2f+1$ \textit{commit} messages matching sequence $s$ have been accepted, and executes request $r$ after all previous requests with lower sequence numbers have been executed.
  \item Process $p$ sends \textit{reply $(r,p)$} to the client.
\end{enumerate}
\noindent\rule[0.25\baselineskip]{\textwidth}{1pt}

Then, we analyse how this agreement protocol works within a view $v$ when the leader is correct.

\begin{lemma}
\label{lemma:1}
  (PBFT: Agreement within a view) Within a view when the leader is correct, if a process gathers a certificate for pre-prepare $(s,r)$, where $s$ is the sequence number and $r$ is the request, then no process can gather a certificate for pre-prepare $(s,r')$ with $r \neq r'$. This ensures that PBFT achieves a unique sequence number within a view.
\end{lemma}

\begin{proof}
  Assume two (even one is enough) processes gather certificates for $(s,r)$ and $(s,r')$. Because a \textit{certificate} for prepare contains $2f+1$ messages, a correct process must sent a \textit{pre-prepare} or \textit{prepare} message for $(s,r)$ and $(s,r')$, respectively (quorum intersection). However, a correct leader only sends a single \textit{pre-prepare} for each $(s,r)$, and a correct follower only sends a single \textit{prepare} for each $(s,r)$. This creates a contradiction.
\end{proof}

One may wonder why the \textit{commit} message is necessary, since \textit{pre-prepare} and \textit{prepare} seem sufficient to tolerate Byzantine faults as proven in Lemma \ref{lemma:1}. Note that Lemma \ref{lemma:1} only guarantees agreement on the unique order within a single \textit{view} when the leader is correct, whereas \textit{commit} protects agreement \textit{across views}, i.e., when the leader is faulty and a new one is elected.

If the leader is faulty, the system needs a view change to move to the next leader's view to continue progress. In the original design, a faulty-timer is embedded in each follower to detect the failure of the leader. The fault of the leader could also include malicious behaviors, e.g., sending conflicting sequences. Now we introduce how PBFT handles the view change phase. We first describe the protocol for followers during \textit{view change}.

\noindent\rule[0.25\baselineskip]{\textwidth}{0.5pt}
\begin{enumerate}[start=7]
  \item When a follower backup $b$'s local timer expires or it detects faulty behaviour from the leader in view $v$, it enters \textit{view change phase} and stops accepting \textit{pre-prepare/prepare/commit} messages from view $v$.
  \item $b$ Gathers $P_c$, the set of all \textit{certificate of prepare} that $b$ has collected since system started. Note this could be optimised by \textit{checkpoint}, the certificates since system started could be reduced to certificate since last valid checkpoint in this way. Here we only consider the simplest mechanism to help better understanding. 
  \item $b$ sends \textit{view-change(v+1, $P_c$, b)}.
\end{enumerate}
\noindent\rule[0.25\baselineskip]{\textwidth}{1pt}

The protocol for the new leader $p_n$ in view $v+1$ is:

\noindent\rule[0.25\baselineskip]{\textwidth}{0.5pt}
\begin{enumerate}[start=10]
\item Accept $2f+1$ \textit{view-change} messages(?) (possibly including $p_n$'s) in a \textit{certificate-for-new-view} set $V$.
\item Gather $O$, which is a set of \textit{pre-prepare(s,r,$p_n$)} messages in view $v+1$ for all pairs $(s,p)$, where at least one \textit{certificate for prepare} for $(s,r)$ exists in $V$.
\item Find $s_{max}$, which is the highest sequence number for which $O$ contains a \textit{pre-prepare} message.
\item Add a message \textit{pre-prepare(s',null,$p_n$)} to $O$ for every sequence number $s'<s_{max}$ for which $O$ does not have a \textit{pre-prepare} message.
\item Send a \textit{new-view ($V, O, p_n$)} message to all nodes and start processing requests for view $v+1$ starting from sequence number $s_{max}+1$.
\end{enumerate}

\noindent\rule[0.25\baselineskip]{\textwidth}{1pt}

After the new leader is elected, other backups $b$ need to check if the newly elected leader is correct.

\noindent\rule[0.25\baselineskip]{\textwidth}{0.5pt}
\begin{enumerate}[start=15]
\item Upon receiving a \textit{new-view $(V, O, p_n)$} , a backup follower stops accepting any messages from the previous view $v$, and sets its local view to $v+1$.
\item \label{step:16}Check if the \textit{new-view} message is from the new leader, and check if $O$ is correctly constructed from $V$. If yes, respond to all the \textit{pre-prepare} messages in $O$ just as in the last view $v$ and start accepting messages from the current view $v+1$. Otherwise, if the current leader is still faulty, trigger a new \textit{view change} to $v+2$.
\end{enumerate}

\noindent\rule[0.25\baselineskip]{\textwidth}{1pt}

Now we check how PBFT guarantees agreement for a unique sequence number during the view change phase.
\begin{lemma}
\label{lemma:3.2}
(PBFT: Agreement during view change). If a view change is triggered in view $v$ and the new view is $v'>v$ in which the leader is correct, if a correct process executes a request $r$ in $v$ with sequence number $s$, then no correct process will execute any request $r'\neq r$ with sequence number $s$ in $v'$.
\end{lemma}

\begin{proof}
If any correct node executed request $r$ with sequence number $s$, there must be at least $2f+1$ processes that have sent a \textit{commit} message regarding $(s,r)$; we denote them by the set $S_1$ (this is why a commit phase is required). The correct processes in $S_1$ must have all collected a \textit{certificate-for-prepare}. When going through the view change phase, a \textit{certificate-for-new-view} must contain \textit{view-change} messages from at least $2f+1$ processes, denoted by $S_2$. Recall the notion of a quorum system and quorum intersection; $S_1$ and $S_2$ are two quorums which intersect with at least one correct process $c_p$. ($2\times(2f+1)-(3f+1)=1$). Therefore, there is at least one correct process $c_p\in S_1\cap S_2$ that has collected a \textit{certificate-for-prepare} regarding $(s,r)$ and whose \textit{view change} message is contained in $V$. Therefore, if some correct node executes request $r$ with sequence $s$, then $V$ must contain a \textit{certificate-of-prepare} regarding $(s,r)$ from $c_p$. Thus, a newly elected correct leader $p_n$ sends a \textit{new-view($v', V, O, p_n$)} message where $O$ contains a \textit{pre-prepare} message for $(s,r)$ in the new view $v'$.
\end{proof}

Correct followers will enter the new view $v'$ only when the \textit{new-view} message for $v'$ contains a valid \textit{certificate-for-new-view} and $O$ is constructed from $V$ as introduced in step (\ref{step:16}). They will then respond to the messages in $O$ before they start accepting \textit{pre-prepare} messages in $v'$, as introduced in step (\ref{step:16}). Therefore, for any sequence numbers that are included in $O$, correct followers will only send \textit{prepare} messages to respond to the \textit{pre-prepare} messages in $O$. Consequently , in the new view $v'$, correct followers can only collect a \textit{certificate-for-prepare} for $(s,r)$ that appears in $O$, while for some $(s,r')$ where $r'\neq r$, no process can collect a \textit{certificate-for-prepare} because there is a lack of \textit{pre-prepare} messages for $r'$ (recall that in step (\ref{step:2}), a \textit{certificate-for-prepare} needs $2f$ \textit{prepare} messages together with the \textit{pre-prepare} message).

\begin{theorem}
\label{theorem:3.1}
PBFT guarantees a unique sequence number.
\end{theorem}

\begin{proof}
  When PBFT is going through the normal case in which the leader is correct, by Lemma~\ref{lemma:1}, PBFT guarantees a unique sequence number. When the leader is faulty and a view change is triggered, by Lemma~\ref{lemma:3.2}, PBFT guarantees a unique sequence number through the view change phase. By combining these two lemmas, we obtain Theorem~\ref{theorem:3.1}.
\end{proof}

\textbf{$\triangleright$ Remark.} We have shown that PBFT guarantees safety. Recall the notion introduced in Section~\ref{notion:safety}, that is, as long as nothing bad happens, the correct processes always agree (never disagree) on requests that were committed with the same unique sequence number. It is important to be clear that PBFT is a total ordering protocol for SMR in which all correct processes agree on the unique sequence for different requests as a consecutive service; it is not a Byzantine agreement protocol in which every process has an initial value and all of them finally agree on the same value once. In addition, to achieve liveness, PBFT assumes that the message delays are finite and bounded , which indicates that the network of PBFT is partially synchronous . In other words , in a completely asynchronous network, PBFT might have no progress at all. However, a partially synchronous network is sufficiently practical for deployment.

\subsubsection{Quorum-based Total Ordering Solutions: Q/U and HQ}\ 

\noindent Although PBFT in \ref{sect:pbft} made the consensus solution in distributed system practical, 
\noindent Q/U~\cite{abd2005fault} (Query/Update) points out that the PBFT protocol in \ref{sect:pbft} is not fault-scalable, i.e., its performance decreases rapidly as more faults are tolerated. Therefore, Q/U adopts another mechanism to achieve replication: the operations-based interface, which can achieve SMR in a similar manner (but different from the notion of SMR, as recalled in Section~\ref{sect:SMR}). Q/U has two kinds of operations: queries and updates. The former do not modify objects, but the latter do . To reduce message and communication complexity, Q/U shifts some tasks from replicas to the clients. Unlike PBFT, the clients in Q/U not only send the requests, but are also responsible for selecting the quorum, storing replica histories returned by servers, and dealing with contention. When performing operations on an object in Q/U, clients issue requests to a quorum of servers. If a server receives and accepts a request, it invokes a method on its local object replica. The servers do not directly exchange information as in PBFT. Each time a request is sent by a client, the client needs to retrieve the object history set. This set is an array of replica histories indexed by server and represents the client's partial observation of the global system state at some time point. Quorum intersection guarantees that for any client requests, there is at least one correct node that has the newest system state. Through client-server communication, replica histories are communicated between servers via the client; if some server does not receive some update requests (it is possible because each request is only sent to a quorum by the client), it synchronises the newest system state with the help of the client. In this way, system overhead is transferred from servers to clients.

Q/U also tolerates Byzantine failures (including clients and servers). If a server crashes, some quorums might be unavailable. Consequently , the clients need to find a live quorum to collect a quorum of responses. Cryptography techniques are used to limit the power of faulty components. No server can forge a newer system state that does not exist; the forged faulty information can be detected by any correct client. Therefore, from the perspective of a correct client, as long as enough correct servers exist, it can continue to make progress. If some replicas are not up to date, correct clients will help them reach the current state. If any faulty client only issues update requests to a subset of a quorum, correct clients will repair the consistency by updating correct servers to the current state in a quorum. In addition, any faulty client that fails to issue update requests to all nodes in a quorum can be detected by techniques like lazy verification~\cite{abd2005lazy}. The total order property is guaranteed by the quorum system and cryptography, i.e., at least one correct replica has the correct history due to quorum intersection. This is used in conjunction with a cryptography-based validation mechanism to prevent forged requests and ensure the system agrees on the same sequence for every update request (query requests do not influence system consistency). If contention occurs , correct clients are responsible for the repair operation. Q/U can tolerate up to $f<\frac{N}{5}$ faulty replicas.


To summarise, Q/U relies on correct clients to make progress and is leaderless, meaning it does not contain a leader to accept and forward clients' requests. The clients are responsible for selecting quorum, storing the history, proving the validity of their history, repairing inconsistency, solving contention, etc. This design reduces the overhead when some replicas are faulty because the correct clients only need enough (a quorum) correct replicas' responses, and no all-to-all communication between replicas is needed. However, this design also leads to the inability to batch clients' requests.

HQ~\cite{cowling2006hq} (Hybrid Quorum) argues that although Q/U has better performance with increasing faulty processes than PBFT, it needs more replicas ($5f+1$) to tolerate faulty nodes than PBFT ($3f+1$). In addition, when contention exists, Q/U performs poorly because it resorts to exponential back-off to resolve contention. Therefore, HQ proposes a hybrid quorum replication protocol. In the absence of contention, HQ uses a quorum protocol where reads (a.k.a. queries in Q/U) and writes (updates in Q/U) require different communication between clients and replicas. The read/write operations in HQ are similar to those in Q/U, in which the clients send requests to a quorum, but HQ adds an additional round to make two rounds for the write operation to tolerate $f<\frac{N}{3}$ faulty nodes. Similar to Q/U, the total order property is ensured by quorum intersection together with cryptography techniques. In HQ, \textit{certificates} are used to ensure that write operations are properly ordered. Quorum intersection guarantees that at least some correct replicas have the newest state, and certificates are used to convince other replicas. When contention occurs , HQ utilises a BFT state machine replication protocol to reach agreement on a deterministic ordering of the conflicting requests. Hence, the total order property is always guaranteed as the system makes progress.

\begin{figure}[h]
 \centering
 \includegraphics[width=3.0in]{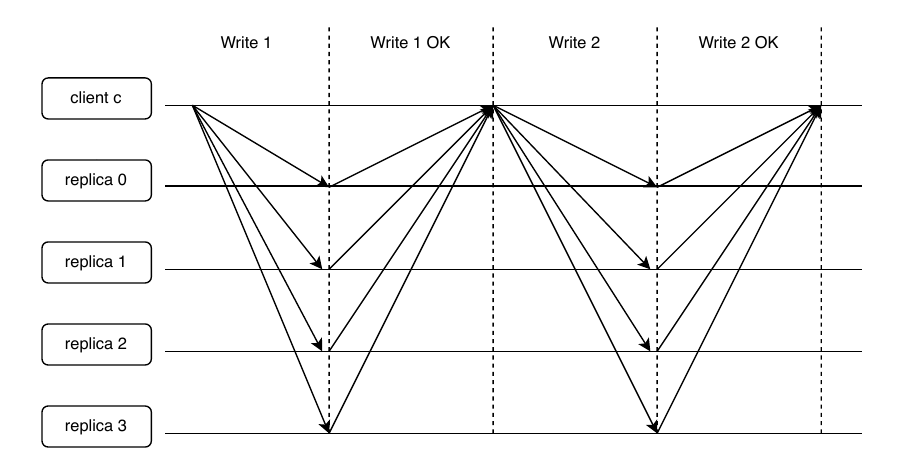}
 \caption{Quorum-based solutions}
 \Description{Quorum}
\end{figure}

\textbf{$\triangleright$ Remark.} The quorum-based solutions reduce the overhead on replicas. These protocols require clients to execute more operations, which helps avoid excessive message communication between replicas, and thus improves performance in the existence of faulty replicas. These protocols rely on correct clients to ensure system progress. The drawback, however, is also caused by their leaderless design: the clients' requests cannot be batched because each client's request is also responsible for dealing with local state updates in replicas and repairing inconsistency through client-server communications. In addition, they do not have a leader funnelling all requests to other replicas. In summary, if accepting $5f+1$ replicas, Q/U is the best choice; otherwise, if accepting $3f+1$ replicas, HQ is the best with a batch size of 1. Otherwise, PBFT is the best option to outperform Q/U and HQ due to its ability to process requests in batches ~\cite{cowling2006hq}.

\subsubsection{Leader-based Total Ordering with SMR design and Speculation: Zyzzyva}\ 

\noindent 
The development of PBFT, as discussed in Section \ref{sect:pbft}, has established a gold standard for the deployment of BFT-SMR, marking a significant milestone in the field of distributed consensus. PBFT's leader-based design, coupled with its view-change protocol, has demonstrated practical feasibility by enabling fault tolerance across all operational phases. However, in real-world applications, faults are often rare or non-existent. Consequently, some practitioners remain hesitant to adopt BFT systems, partly due to the perception that the overheads associated with Byzantine fault tolerance are prohibitively high. Zyzzyva~\cite{kotla2010zyzzyva} is therefore proposed, which is based on the SMR approach with a leader-based design, since quorum-based solutions cannot batch concurrent client requests, which limits throughput. Different from PBFT, which runs an expensive agreement protocol on the requests' deterministic final order before execution, Zyzzyva utilises \textit{speculation} and immediately executes requests. The motivation for this is that, the possibility of the existence of a failure is small in a real environment, and thus Zyzzyva proposes to make fast progress based on the speculation that every node is working properly and only exchanges messages when faults exist.

A response from a replica not only contains an application-level reply, but also includes the history on which the reply depends. This allows the client to determine when a request has completed . A request completes in one of two ways. The first way is the fast case: if the client receives $3f+1$ matching responses, it considers the request complete since all correct nodes must have responded to this request in this case. The second one is the two-phase case: if the number of received matching responses is between $2f+1$ and $3f$, the client gathers $2f+1$ matching responses and forms a \textit{commit certificate}, which includes a cryptographic proof that $2f+1$ replicas agree on a linearizable order of the request and all preceding requests. Due to quorum intersection, the commit certificate of $2f+1$ replicas ensures that no other ordering can get $2f+1$ matching responses to contradict this order. Once $2f+1$ replicas acknowledge receiving the commit certificate, the client considers the request complete. Safety is also guaranteed through clients. A faulty client may alter the commit certificate or fail to send a commit certificate. The altered certificate will be detected and ignored by correct replicas with the property of cryptography. If a faulty client fails to deliver a commit certificate, it may not learn when its request completes, and a replica whose state has diverged from its peers may not discover the situation immediately. However, if in the future a correct client issues a request, it will either complete that request or trigger a view change if it finds an inconsistency (in the history of a response). If the view change is triggered, i.e., the inconsistency causes less than $2f+1$ matching responses, the correct client will resend its request to all replicas, which then forward the request to the primary in order to ensure the request sequence guarantees the request is eventually executed. Additionally, all previous requests are delivered to every replica, and the request is completed. If the primary is faulty and a correct client receives responses indicating inconsistent ordering by the primary, it gathers a proof of misbehaviour and initiates a view change.

\begin{figure}[h]
 \centering
 \includegraphics[height=1.2in]{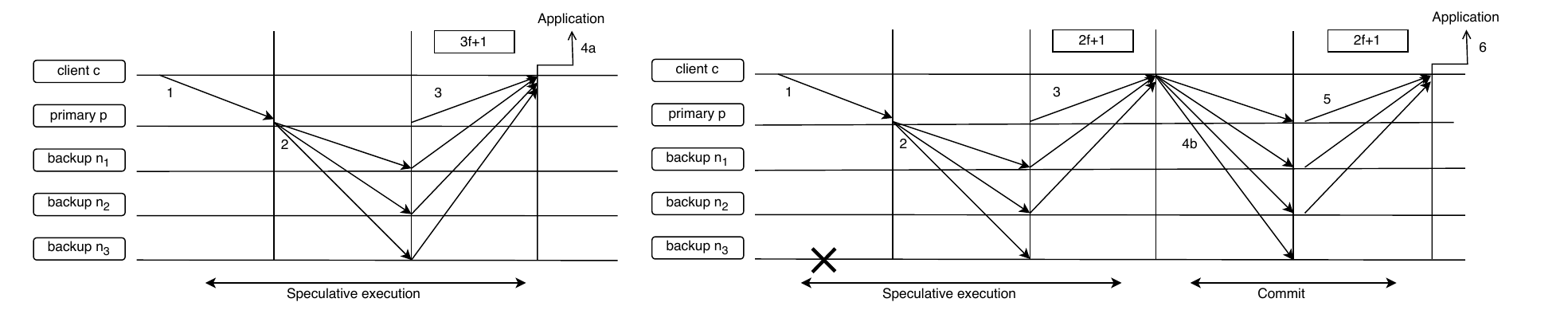}
 \caption{Zyzzyva: Fast case and two-step case}
 \Description{Zyzzyva}
\end{figure}

\textbf{$\triangleright$ Remark.} Zyzzyva follows the leader-based paradigm with an SMR design in which all replicas eventually execute the requests in the same order. It also utilises a similar idea to guarantee total ordering: it ensures a unique sequence within a view and through the view change, which guarantees a unique sequence through the system 's execution . Zyzzyva uses speculation to improve performance when the system is working properly , which is the case during most of the time. However, Zyzzyva utilises heavy cryptography in the existence of a fault to ensure a unique sequence, which causes performance in the faulty case to drop sharply.

\subsection{Make BFT Scaling: Parallel and Sharding/Partition-based BFT}
\label{sect:scalingbft}
Although PBFT have made distributed consensus solution practical, and other optimized BFT solutions have been proposed, as introduced in Section \ref{sect:mostpopular}, BFT-SMR is considered slow as replicas need to agree on a global total order of client requests with the help of the leader node. It is identified in~\cite{aublin2013rbft} that PBFT and other leader-based solutions could be attacked by the leader, i.e., the leader might be maliciously smart and degrade system throughput while not being detected by follower replicas. Aublin \cite{aublin2013rbft} proposed a novel solution called Redundant BFT (RBFT), which executes $f+1$ multiple BFT instances, where each instance has a different leader assigned to it at different machines. Every leader will order the requests, but only the requests ordered by the master leader will be executed. Multiple instances assigned to different leaders make them able to monitor the performance of the master leader. If the performance of the master leader is considered slow, a new master leader will be selected.

Parallel instances are not only used to monitor the performance of the leader, but also deployed to speed up system performance. In this kind of deployment, the motivation is to relax the requirement to execute the client requests following a deterministic sequence, i.e., these approaches usually run concurrent instances to execute non-overlapping requests~\cite{kotla2004high,kapritsos2012all}. The core idea in CBASE \cite{kotla2004high} is to relax the order requirement on state machine replication to allow concurrent execution of independent requests without compromising safety. To accomplish this, the authors introduce a paralleliser that leverages application-specific rules to identify and execute concurrent requests in parallel. By doing so, the replicated system's throughput scales with both the level of parallelism exposed by the application and the available hardware resources. In comparison, Eve~\cite{kapritsos2012all} uses a system-level approach. Their execute-verify approach uses a similar speculation improvement as Zyzzyva: replicas first speculatively execute requests concurrently, and then verify that they have agreed on the state and the output produced by a correct replica. If too many replicas diverge, Eve rolls back and re-executes the requests.


BFT has high client scalability; however, it suffers from low server scalability, i.e., BFT cannot scale to a large number of replicas. Two main streams of scalability solutions are proposed: hierarchical BFT and sharding/partition BFT. The idea of these two kinds of solutions follows the optimisation mechanism of reducing all-to-all messages, which is a common approach in distributed systems~\cite{cachin2011introduction}. Hierarchical BFT protocols usually involve several phase(s) of BFT instances at different hierarchies~\cite{amir2006scaling,wu2021me,arun2019ez,king2011breaking}. The representative example Steward~\cite{amir2006scaling} divides replicas into several groups, where each group has $n>3f+1$ replicas. One group serves as the leader group, and will be changed by a view change protocol if it is detected to be faulty. A client inside a group sends requests to the representative node in its group if it is not in the leader group. The requests will be forwarded to the leader group and will be assigned a sequence number through PBFT in the leader group. The order will be sent back to all groups and finally every replica executes the requests in the same order. The core idea of Steward is that each group is fault-tolerant and could be logically regarded as a replica, thus reducing all-to-all communication to enhance scalability. It is worth noting that the fault-tolerant capacity is reduced in this pattern, as it requires each group to have $3f+1$ replicas, which is essentially a more strict assumption than requiring the total system to have $3f+1$ replicas. EZBFT~\cite{arun2019ez} utilises a leaderless mechanism: each replica can receive requests from clients and assign a sequence number to each request, then forward the requests to other replicas. Each replica has an instance space which contains all the sequence numbers it can assign to a request. The conflict is resolved at the client side: like Q/U and HQ, a request from a client not only contains the request itself, but also includes the dependencies of this request. The dependencies are used to identify if this request does not conflict with others (e.g., a read request) and can be executed immediately. If a request might have a conflict, replicas will execute it in sequence. In ME-BFT~\cite{wu2021me}, the PBFT protocol is only deployed within a fixed number of full replicas, while other light replicas follow a full replica as a representative and rely on it to forward their requests to other full replicas. If a light replica detects that the full replica it relies on is faulty, it will change its representative.

The above design also follows the idea of reducing all-to-all communication overhead to achieve scaling, but the request data needs to be Conflict-free Replicated Data Types (CRDT)\cite{shapiro2011conflict}. In \cite{king2011breaking}, the authors design a novel protocol that replaces electing a leader with electing a random vector; combined with secret-sharing, the adversary needs to collude with more replicas to influence the system, and the possibility of this happening decreases with an increasing number of replicas. The messages sent by each replica are reduced from $O(n)$ to $O(\sqrt n)$ to directly reduce communication overhead and achieve scaling. SBFT\cite{gueta2019sbft} also considers reducing communication complexity. In order to avoid complex all-to-all communication, SBFT assigns one or several nodes as collectors to collect a threshold signature and distribute the aggregated signature to every participant, which reduces the communication complexity from linear to constant. It also uses redundant servers for the fast path , similar to Zyzzyva; it allows a fast agreement protocol when all replicas are non-faulty.

Sharding/partition based BFT divides replicas into several partitions according to different application states to achieve scaling~\cite{bezerra2014scalable,eischer2017scalable,li2017enhancing}. The core idea of these schemes is to improve performance by allowing write requests to be executed concurrently by different partitions. If write requests cause contention, i.e., different write operations access the same partition or a write request accesses multiple partitions~\cite{li2017enhancing}, it needs to be resolved. When contentions need to be resolved frequently, the performance will degrade due to heavy conflict resolution procedures. Some solutions have been proposed to mitigate the conflict problem; for example, \cite{hong2023prophet} designed a conflict-free sharding-based Byzantine-tolerant total ordering protocol. It first pre-executes cross-shard requests, then delegates a random shard to sequence the pre-executed requests for a global order based on prerequisite information. Cross-shard verification is used to verify the pre-execution result. Mir-BFT\cite{stathakopoulou2019mir} is a BFT protocol that combines parallel and partition optimisation schemes. It uses parallel leaders to order clients' requests concurrently to handle scalability. To multiplex PBFT instances into a single total order, Mir-BFT also partitions the request hash space across replicas to prevent request duplication.

\textbf{$\triangleright$ Remark.} 
Parallel and sharding/partition optimisation mechanism\blue{s} generally consider reducing all-to-all messages\cite{amir2006scaling,king2011breaking,hong2023prophet,gueta2019sbft,wu2021me} and increasing concurrency\cite{wu2021me,arun2019ez,stathakopoulou2019mir} to enhance server scalability. The overall structures follow the typical BFT-SMR paradigm with a leader and followers to execute requests in the same order to achieve the same state. Cryptography, such as threshold signatures, is used \cite{king2011breaking,gueta2019sbft} to reduce communication complexity; this works well in environments with powerful computation, such as blockchain systems, but might be a bottleneck for computation-restrained environments such as wireless~\cite{xu2022wireless,cao2022v2v}. In addition, in some parallel BFT protocols, the fault model has also been changed. For instance, in \cite{amir2006scaling}, it is assumed that each group of $n$ replicas is fault tolerant and has up to $\lfloor \frac{n-1}{3}\rfloor$ faulty replicas, which is different from (in essential a harsher assumption than) assuming the whole system has up to one-third faulty nodes.

\subsection{The Power of Trusted Component: Trusted Execution Environment-based BFT}
\label{sect:bfttrustcomp}
In the early years (c.f. Section \ref{sect:asyn}), Bracha's asynchronous reliable broadcast \cite{bracha1987asynchronous} had a disruptive impact on Byzantine fault tolerance by restricting Byzantine nodes' behaviour to a crash failure mode. However, Bracha's reliable broadcast and related solutions that restrict Byzantine nodes' ability rely on message passing, which could result in high communication overhead. The use of a trusted execution environment (TEE) shares a similar motivation: to shift tasks from replica software (i.e., communication) to special trusted hardware in order to restrict the power of a faulty node.

Recall that in Section \ref{Sect:FLP}, we introduced three mechanisms to bypass the FLP impossibility: loosening the asynchrony assumption (e.g., PBFT), utilising randomness (e.g., Ben-Or's randomised algorithm), and adding detectors or oracles. Correia et al.'s solution\cite{correia2004tolerate} employs the third mechanism; it relies on an oracle to circumvent FLP and increase resilience. Correia et al. extend the asynchronous system with an oracle called the Trusted Timely Computing Base (TTCB). The TTCB is a real-time, synchronous subsystem capable of timely behaviour and provides a small set of basic security services to achieve intrusion tolerance~\cite{correia2002design}. It could be implemented on RTAI~\cite{cloutier2000diapm}, which is protected by hardening the kernel. Based on TTCB, the authors proposed a Trusted Multicast Ordering (TMO) service . TMO is implemented in TTCB, making it secure against malicious attacks ; thus, only crash faults exist . The messages are passed through the payload network, which is assumed to be asynchronous but reliable, meaning that the messages will not be altered and will eventually be received. In addition, messages are verified by Message Authentication Codes (MACs) to protect their integrity. The TMO service assigns a unique sequence number to each request, and the messages of TMO are passed through the TTCB control channel, which is reliable and synchronous. The whole procedure is as follows: when a process calls a \textit{send} or \textit{receive} operation in its local TTCB, information about this call is broadcasted to everyone. If the processes acknowledging this operation (by sending a cryptographic hash digest of this request to the coordinator, which is also a local TTCB service) reach the threshold of $\lfloor \frac{n-1}{2}\rfloor$, the coordinator will assign a sequence number to the operation. Other nodes can obtain this sequence number through the TTCB channel. If the coordinator crashes (as previously mentioned, no malicious faults exist ), another coordinator will take over . This is possible because the other coordinator is aware of the broadcast made by the crashed coordinator through the reliable channel.

Correia et al.'s work achieves SMR for several reasons. Recall the definition of atomic broadcast in Definition \ref{def:atomic}: validity is guaranteed by the reliable payload network; agreement is achieved by the reliable payload network and the TMO service, as a sequence number is only assigned to a request when the number of correct processes reaches a threshold; integrity is guaranteed by the reliable payload channel, the TTCB control channel, and MACs; and total order is provided by the TMO. Therefore, Correia et al.'s scheme achieves atomic broadcast (a.k.a. total order broadcast), which can be used to build an SMR if all processes started with the same initial state.

Similar to TMO, which restrains the power of malicious nodes, A2M~\cite{chun2007attested} introduces Attested Append-Only Memory. This memory pool equips a host with a set of trusted, undeniable, ordered logs, which only provide \textit{append}, \textit{lookup}, \textit{end}, \textit{truncate}, and \textit{advance} operation interfaces and has no method to replace values that have already been assigned (a sequence number). In implementation, A2M could be a trusted virtual machine, trusted hardware, etc. A2M can be integrated into various BFT protocols, including PBFT and Q/U. For instance, in the A2M-integrated PBFT protocol, A2M-PBFT-EA, every message needs to be inserted into the corresponding message memory (i.e., the \textit{Pre-Prepare message memory pool}). With the help of A2M, no equivocation exists, i.e., every replica is forced to send consistent messages to others. Therefore, when utilizing the same flow as PBFT, A2M-PBFT-EA only needs $2f+1$ replicas to tolerate $f$ Byzantine faulty nodes.

TMO and A2M prove that by deploying tamper-proof components to extend the server in SMR, one can reduce the number of replicas from $3f+1$ to $2f+1$. However, TMO and A2M are both locally installed services on each computer (server), instead of being distributed deployments . To tackle this problem, MinBFT and MinZyzzyva were proposed by Veronese et al.~\cite{veronese2011efficient}. They designed a local service called the \textit{Unique Sequential Identifier Generator} (USIG). Its function is to assign a counter value to each message. USIG guarantees three properties even if the replica becomes faulty: (1) It will never assign the same identifier to two different messages; (2) It will never assign an identifier lower than a previous one; and (3) It will always attribute a sequential identifier. USIG can be built on public-key cryptography, and it needs to be deployed on an isolated, tamper-proof component that is assumed to be incorruptible.
As reviewed in Section \ref{sect:pbft}, we analysed how the \textit{prepare} and \textit{pre-prepare} phases guarantee a unique sequence within a view. In MinPBFT, the unique sequence is guaranteed by USIG, so the \textit{pre-prepare} phase is not needed. The flow is similar to PBFT: the client first sends a \textit{request} message to all replicas, then the leader obtains a sequence number from USIG and broadcasts a \textit{prepare} message to all replicas. After receiving the \textit{prepare} message, a replica verifies whether the message and sequence are valid, and broadcasts a \textit{commit} message if the verification passes. When a replica receives $f+1$ valid \textit{commit} messages, it accepts the request and sends a \textit{reply} after execution. If the leader is faulty and a \textit{view change} is triggered, replicas stop receiving messages, obtain a unique identifier from USIG for the \textit{view change} message, and then start to listen to the new leader if it is correct. Similar to PBFT, all requests since the last checkpoint will be stored in a new view certificate of the new leader to be broadcast to every replica, and requests will be executed starting from the latest one. However, unlike in PBFT, the unique sequence is not guaranteed by a quorum as proven in Lemma \ref{lemma:3.2}, but rather is assigned by the USIG. MinZyzzyva is designed in a similar way, following the flow of Zyzzyva, and is not described repeatedly.

Based on the idea of MinBFT, CheapBFT~\cite{kapitza2012cheap} was proposed as an optimisation of MinBFT. Similar to MinBFT, CheapBFT also relies on a trusted component named the Counter Assignment Service in Hardware (CASH) to manage a monotonically increasing counter, and guarantee node identity and message authentication. By utilising CASH, no equivocation exists. In the normal case when no faults exist , CheapBFT utilises an optimised flow where only $f+1$ ($f=\lfloor\frac{n-1}{2}\rfloor$) nodes act as active replicas that send and receive \textit{prepare} and \textit{commit} messages. Passive replicas only need to listen to \textit{update} messages that include the execution result and the commit certificate. When a fault is detected or the system fails to reach agreement, CheapBFT falls back to MinBFT to achieve correctness.

\textbf{$\triangleright$ Remark.} TEE-based approaches can enhance the resilience and performance of BFT protocols as well as simplify them. By utilizing trusted components, the power of malicious participants is restricted. Achieving non-equivocation is a core design goal of TEE-based BFT protocols. The power of non-equivocation is introduced in \cite{clement2012limited}; combined with the transferable authentication of messages in the network, e.g., digital signatures , it is possible to use non-equivocation to transform CFT protocols to tolerate Byzantine faults. In these TEE-based works~\cite{correia2004tolerate,chun2007attested,veronese2011efficient,kapitza2012cheap}, the cost of fault-tolerance is reduced from $3f+1$ to $2f+1$. \cite{correia2004tolerate,chun2007attested} use local trusted components, which means replicas should be physically connected. \cite{veronese2011efficient,kapitza2012cheap} use a similar flow as PBFT such that they can be deployed in a distributed manner. However, the trusted environment itself is a challenge and a drawback, as the trusted components themselves are required to participate in the protocols. In addition, since all known TEEs have vulnerabilities~\cite{cerdeira2020understanding,fei2021security}, the protocols should not rely on the trusted components to make progress.

\subsection{Partially Synchronous BFT for Blockchain}
\label{sect:partially sync BFT for blockchain}
Blockchain is essentially a form of SMR \cite{shi2020foundations}, the fast development of blockchain has triggered research interest in traditional SMR and consensus protocols, which have been researched before the advents of blockchain and is introduced in Section \ref{sect:consensusearlyyears} and the first half of Section \ref{sect:practicalbft}. However, blockchain has different requirements and challenges than traditional SMR and consensus. One simple difference is that in a blockchain system, computational power is much cheaper than in general distributed systems, so more complex cryptographic techniques can be used. Another difference is that blockchain requires ``fairness'' and chain quality~\cite{garay2015bitcoin}, which necessitates regular leader rotation, while in SMR there is no need to change a leader if it is working properly. In addition, blockchain has a higher demand for scaling because the number of participants is generally larger than in traditional SMR.

\subsubsection{Tendermint: BFT in Blockchain arena with Voting Power}\ 

\noindent Tendermint~\cite{buchman2018tendermint} is a PBFT-like consensus mechanism designed for the blockchain context. Similar to PBFT, it assumes a partially synchronous network model in which the communication delay is bounded . However, Tendermint identifies several different requirements for blockchain compared to standard SMR. First, in the consensus mechanisms like PBFT, the value for a decision that every correct process considers acceptable is not just the value sent by someone, but a set of $2f+1$ messages with the same value ID . Therefore, each message can be regarded as an ``endorsement'' or ``vote'' for a value. In blockchain, each participant has a different amount of voting power. This makes the scenario different from traditional SMR, where voting power in SMR is completely identical, with each message having an identical weight. It can also be challenging to scale, as the total number of messages being sent in traditional SMR depends on the total number of processes. To address this issue, in Tendermint, the $2f+1$ quorum intersection is replaced by $2f+1$ voting power, with the total system voting power being divided into $N=3f+1$. Another difference is that in PBFT, the requests from clients will be executed, but in Tendermint, any transaction in the proposed block needs application-level validity verification. In Tendermint, one or more proposer(s) will propose a block until a majority (2N/3) of voting power approves it, and the block will be solidified through several rounds of communication involving voting, corresponding to several rounds of message passing in PBFT.

\subsubsection{Algorand: Byzantine Agreement with VRF and Common Coin}\ 

\noindent Unlike most BFT blockchains that generally operate in a permissioned network, Algorand~\cite{gilad2017algorand} proposes to design a BFT protocol that can be used for the deployment of a permissionless blockchain. They argue that a permissionless blockchain requires ``fairness'', which indicates that a participant with a higher balance amount should be more likely to be elected. Therefore, the authors invented a cryptographic sortition scheme using a VRF to secretly elect some participants based on priority (which is positively correlated with their balance amount) as the committee to validate the block sent by the proposer (which has the highest priority). Either only one of the proposed blocks or an empty block should be approved. Because the committee members change each round, and the value that needs to be approved is not a sequence number or an order, but a block, the authors propose using Byzantine Agreement (Definition~\ref{def:byzantineagreement}) to reach agreement on the block. The procedure of the proposed Byzantine agreement protocol is voting-based. Once a proposed block receives a voting threshold , it is approved. Recall that Byzantine Agreement is a synchronous protocol, as introduced in \ref{Sect:AinSyn}, and its ``validity'' property requires that all correct nodes have the same initial value (a.k.a. all-same validity).

The author clearly states that the best-case scenario happens when the network is under strong synchrony and only a same block was sent to everyone by the proposer, and in this scenario Byzantine agreement protocol directly reaches $final$ consensus. If a worse-case scenario happens, either in an asynchronous network or the proposer is malicious, binary agreement is used to decide if the decision is one of the proposed blocks or an empty block. Intuitively, the decision of a binary consensus is simply from $\{0,1\}$. 1 represents one selected block, and 0 is an empty block.
The selection scheme follows: when the proposer proposes different blocks and sends them to others, but one of them receives enough votes will be chosen, or simply one of the proposed blocks is selected if none of them receives enough votes. Recall that in section~\ref{sect:benor}, Ben-Or's randomised binary consensus can reach consensus with the probability of 1, although it could be slow, but a shared coin (common coin) can significantly boost the consensus. In Algorand, the authors applies a randomised binary consensus with the deployment of common coin. The common coin is calculated by the ``$hash$'' produced by VRF which has the least priority. If the committee that yields $hash$ is honest, all correct participants will observe the same coin and terminate fast. The probability for each loop that the committee is honest excels $\frac{2}{3}$ due to the utilisation of VRF. In addition, even if the this committee is malicious, binary consensus will terminate with the probability of 1 anyway with more loops to be executed.

The authors state that the best-case scenario happens when the network is under strong synchrony and only the same block is sent to everyone by the proposer; in this scenario, the Byzantine agreement protocol directly reaches a $final$ consensus. If a worst -case scenario occurs , either in an asynchronous network or when the proposer is malicious, binary agreement is used to decide if the decision is one of the proposed blocks or an empty block. Intuitively, the decision of a binary consensus is simply from ${0,1}$, where 1 represents one selected block, and 0 is an empty block. The selection scheme is as follows: when the proposer proposes different blocks and sends them to others, the block that receives enough votes will be chosen; otherwise, one of the proposed blocks is selected if none of them receives enough votes. Recall that in Section ~\ref{sect:benor}, Ben-Or's randomised binary consensus can reach consensus with a probability of 1, although it may be slow, but a shared coin (or a common coin) can significantly boost the consensus. In Algorand, the authors apply a randomised binary consensus with the deployment of a common coin. The common coin is calculated by the ``hash'' produced by the VRF output . If the committee that yields the $hash$ is honest, all correct participants will observe the same coin and terminate fast. The probability for each loop that the committee is honest exceeds $\frac{2}{3}$ due to the use of the VRF. In addition, even if this committee is malicious, binary consensus will terminate with a probability of 1 regardless, but more loops will need to be executed.

\subsubsection{HotStuff: BFT with Quorum Certificate and Frequent View Change}\ 
\label{sect:hotstuff}

\noindent 
In blockchain, there is sufficient calculation power available for cryptographic mechanisms, and the leader node needs to be frequently rotated to ensure chain quality~\cite{garay2015bitcoin}. Therefore, HotStuff~\cite{hotstuff2019yin} was proposed to utilize a cryptographic threshold signature scheme to reduce authenticator complexity, which refers to the number of authenticators a replica must reach to make a consensus decision. Similar to PBFT, HotStuff is a leader-based SMR protocol which includes several phases, namely \textit{prepare}, \textit{pre-commit}, and \textit{commit}. The core idea of HotStuff is also \textit{quorum intersection}, but it uses a \textit{Quorum Certificate (QC)} instead of a quorum of messages as in PBFT. A QC is essentially a certificate that is signed by a sufficient number of (i.e., $n-f$) replicas to demonstrate their agreement on a \textit{view}; each view corresponds to a unique QC. Another difference is that \textit{view change} is triggered after reaching a consensus decision each time , to address the demand for fairness in blockchain.

HotStuff is similar to PBFT in that it does not require synchrony for safety but does for liveness.
Each replica stores a tree of pending commands as a local data structure. The commands are linked by \textit{parent links}. As a tree structure, the command tree can fork, just like a fork in Bitcoin. A \textit{branch} led by a given tree node is the path from this node all the way back to the tree root following \textit{parent links}. Only one branch will be chosen during each round of the decision process. The main procedure is as follows:
\begin{enumerate}

\item Prepare Phase: For each \textit{new view}, a leader will be elected by a certain scheme known to everyone, and each replica sends a \textit{new view} message that carries the highest \textit{prepareQC} from the last view, which records the last branch that received $n-f$ votes.
\par (a) The leader will find the one with the highest view, namely \textit{highQC}, and choose a branch that extends from the tail of the last QC, i.e., choosing a new leaf node from the command tree. The selected leaf, together with the \textit{highQC}, will be embedded into a \textit{prepare} message as the current proposal \textit{curProposal} and sent to the replicas.
\par (b) Once a replica receives the \textit{prepare} message, it performs a \textit{safeNode} predicate to verify if the \textit{prepare} message is acceptable. The \textit{safeNode} predicate has two rules. One is to check if the proposed leaf node directly extends from a previously decided node; if so, the replica is synchronised with the leader and no gap exists. The other rule is to check if its current view of \textit{lockedQC} is smaller than that of the \textit{prepare} message; if so, the replica might be stuck at a previous view and should try to synchronise with the leader. If \textit{safeNode} predicate returns true, the replica accepts the prepare message.
\item Pre-Commit Phase: Once the leader receives $n-f$ \textit{prepare} votes for \textit{curProposal}, it combines them into a \textit{prepareQC} for the current view. Then, the leader broadcasts the \textit{prepareQC} in a \textit{pre-commit} message. Replicas will respond to the leader with a vote for \textit{pre-commit} along with their signature for the proposal.

\item Commit Phase: Similar to the Pre-Commit Phase, when the leader receives $n-f$ \textit{pre-commit} votes, it combines them into a \textit{precommitQC} and broadcasts them in the \textit{commit} messages. Then, replicas respond with \textit{commit} votes. At this point, a replica becomes locked on this \textit{precommitQC} and sets \textit{lockedQC} to \textit{precommitQC}.
\item Commit Phase: Once the leader receives $n-f$ \textit{commit} votes, it combines them into a \textit{commitQC} and sends it with the \textit{decide} messages to all replicas. After receiving \textit{decide} messages, a replica considers the proposal in the \textit{commitQC} a committed decision and executes the commands in the committed branch. Then, every replica starts the next view.
\end{enumerate}

\begin{figure}[h]
 \centering
 \setlength{\belowcaptionskip}{-0.2em}
 \includegraphics[width=4.0in]{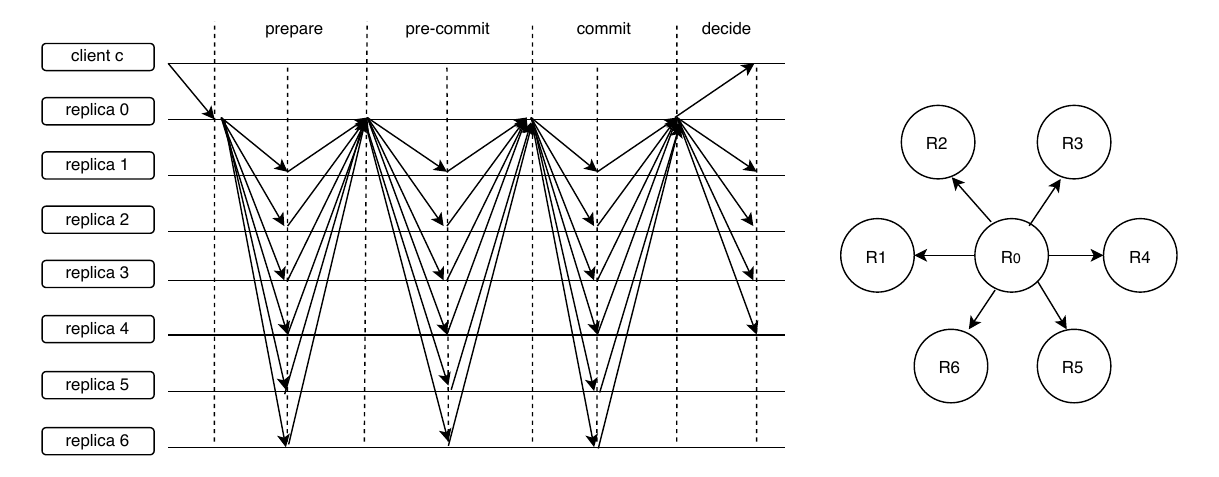}
 \caption{Communication pattern and topology of HotStuff}
 \Description{HotStuff}
 \label{fig:hotstuff}
\end{figure}

The communication pattern with seven processes and the topology are illustrated in Fig.~\ref{fig:hotstuff}. The topology is a star-topology where the leader disseminates and aggregates information from others. In the whole procedure, three kinds of QC exist: \textit{prepareQC}, \textit{lockedQC} (equal to \textit{precommitQC}), and \textit{commitQC}. Readers might be curious about the second one: why do we need \textit{lockedQC}? This question is similar to the question of ``why is the \textit{commit} message needed '' in PBFT. Remember, the answer for PBFT is that the \textit{commit} message is for ensuring the safety of a unique sequence during a ``cross view'' transition. Similarly, \textit{lockedQC} is for safety during a ``cross view'' transition in the command tree. As introduced, the command tree forks, and only one branch will be selected. Therefore, to achieve safety, HotStuff needs to ensure:
\par \textit{If $w$ and $b$ are conflicting nodes, i.e., $w$ and $b$ are on two different branches where none of the branches is an extension of the other, they cannot both be  committed, each by a correct replica.}

As shown in Fig.\ref{fig:HSTree}, if $w$ and $b$ are at the same view, for example, both at $v_1$ in Fig.\ref{fig:HSTree}, it is obvious that they will not be committed simultaneously, as each commit needs $n-f=2f+1$ votes from correct nodes (comparable to the unique sequence within a view in PBFT). If $w$ and $b$ are at different views, as illustrated in the figure ($v_1$ and $v_2$, respectively), we will show how \textit{lockedQC} protects safety.

We assume $w$ is at $v_1$ and $b$ is at $v_2$. $s$ is a node that is the lowest node higher than $w$ and conflicts with $w$. $s$ comes with a valid \textit{prepareQC} which is at $v_s$. $qc_1$ and $qc_2$ are valid \textit{commitQC}s, and $qc_s$ is a valid \textit{prepareQC}. We already know that a \textit{commitQC} is composed of $n-f$ \textit{lockedQC}s, and a \textit{prepareQC} needs to pass \textit{safeNode} checks by $n-f$ replicas. If $w$ and $b$ are both committed, $w$ and $s$ are also both committed. Therefore, the intersection of \textit{commitQC} and \textit{prepareQC} is $n-f$ \textit{lockedQC}s and $n-f$ \textit{safeNode} predicates, $2\times (n-f)-(3f+1)=f+1$, which has at least $f+1-f=1$ correct replica $r$ that updates \textit{lockedQC} to \textit{precommitQC} (i.e., $qc_1$) at view $v_1$. As for \textit{safeNode}, as $w$ and $s$ are conflicting, the first rule in \textit{safeNode} cannot be matched. For the second rule, it is impossible for the correct replica $r$ to have a \textit{lockedQC} view smaller than $qc_1$ because it has already updated its \textit{lockedQC} to \textit{precommitQC}, which is exactly $qc_1$; hence, $n-f$ \textit{safeNode} predicates can never be gathered, which is a contradiction. Therefore, safety during a cross view'' transition is achieved with the help of \textit{lockedQC}, just like with the \textit{commit}'' messages in PBFT.

\begin{figure}[h]

 \centering
 \includegraphics[width=3.0in]{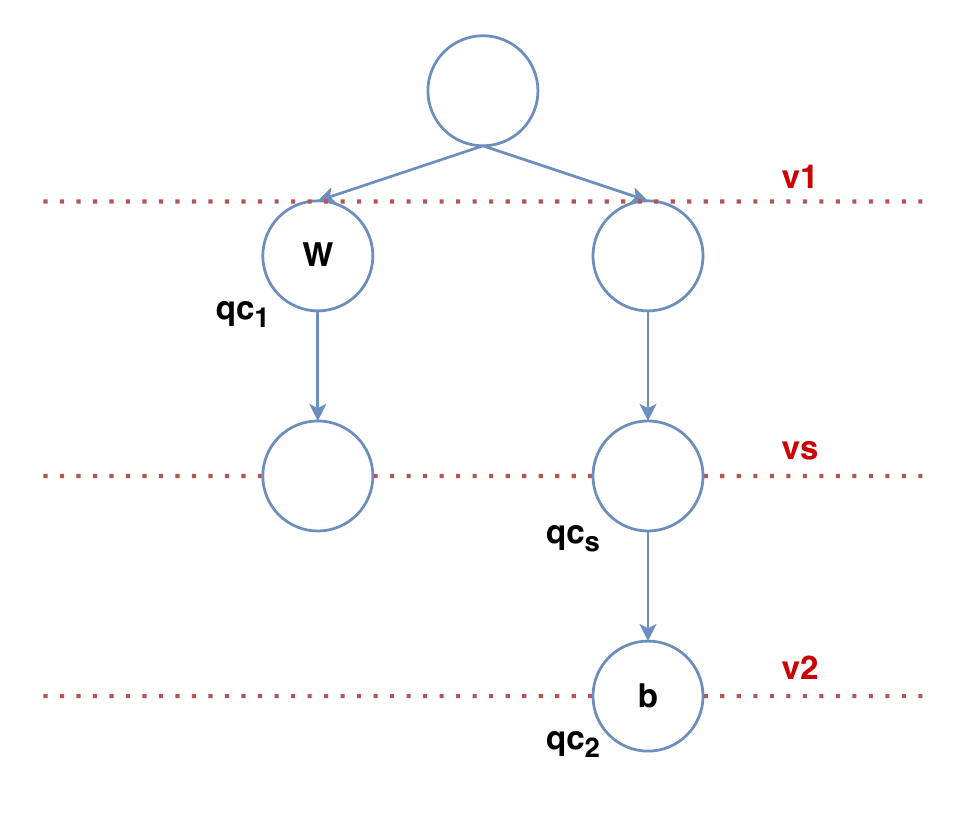}
 \caption{Safety for conflicting nodes}
 \Description{HSTreeConf.}
 \label{fig:HSTree}
\end{figure}

\textbf{$\triangleright$ Remark.} Partially synchronous BFT protocols have been improved for blockchain scenarios. Although blockchain is essentially an SMR system, it has different requirements from traditional SMR. The proposed blockchain-oriented BFT schemes are inspired by traditional agreement and consensus protocols. Tendermint applies quorum intersection but uses voting power instead of the message count. Algorand combines Byzantine agreement and  binary consensus with a common coin, which originate from the Byzantine Generals Problem and Ben-Or's randomised consensus. HotStuff uses a Quorum Certificate, a cryptographic scheme to reduce messages, and  is inspired by PBFT's quorum intersection principle to guarantee safety within a view and across different views.

\section{Practical BFT in Fully Asynchronous Network}
\label{sect:asynchronousbft}
The FLP impossibility theorem~\cite{fischer1985impossibility} states that in an asynchronous environment, even if only one failure exists, a consensus protocol might never be able to achieve correctness. Intuitively speaking, partially synchronous and synchronous protocols might achieve zero throughput against an adversarial asynchronous network scheduler~\cite{wang2022bft}. This is why although partially synchronous protocols like PBFT and HotStuff do not need synchrony for safety, they require partial synchrony and GST for liveness. Thanks to the foundational work of Ben-Or~\cite{ben1983another}, as discussed in Section \ref{sect:earlyrandom}, we know that randomized asynchronous protocols can solve the consensus problem by employing a probabilistic approach to guarantee termination with a probability of 1. However, these protocols are often inefficient due to their exponential worst-case complexity. To address this limitation, techniques such as the common coin method have been proposed to accelerate termination. The common coin can be collectively generated using threshold cryptography~\cite{desmedt1993threshold} and Verifiable Random Functions (VRFs)~\cite{micali1999verifiable}, as explored in Section~\ref{sect:earlyrandom}. When combined with these cryptographic primitives, asynchronous protocols have the potential to achieve practical efficiency.
Asynchronous consensus is considered to be more robust than partially synchronous consensus in the presence of timing and denial-of-service (DoS) attacks, making them appropriate solutions for mission-critical distributed applications and blockchain systems. Before digging into asynchronous consensus protocols, we first introduce the atomic broadcast primitive and how to build it by different ways (Section \ref{sect:MVBA}-\ref{sect:ACS}), and then analyse how asynchronous BFT is designed and deployed in blockchain scenarios (Section \ref{sect:asyncblockchain}).

\subsection{Building Atomic Broadcast on Binary Agreement and Multi-Valued Byzantine Agreement}
\label{sect:MVBA}
In this section, we give a brief overview of how to build an atomic broadcast protocol on top of vector consensus and multi-valued consensus, which can eventually be reduced to binary consensus. Note that atomic broadcast is equal to SMR, which we have introduced in Section ~\ref{sect:SMR}. 

\begin{definition}
\textbf{(Asynchronous) Binary Agreement (ABA)} performs consensus on a binary value $v\in${0,1}. It satisfies the following properties:
\begin{itemize}
\item \textbf{Validity}: If all correct processes propose the same value $v$, then any correct process decides on $v$.
\item \textbf{Agreement}: No two correct processes decide on different values.
\item \textbf{Termination}: Every correct process eventually decides.
\end{itemize}
\end{definition}

Note that Ben-Or's randomised consensus is a typical binary agreement, meaning binary agreement can work in a fully asynchronous network. Note that all of the following ABA schemes utilise the common coin technique to boost termination.
\begin{definition}
\label{def:MVBA}
\textbf{Multi-Valued Consensus (a.k.a. Multi-Valued Byzantine Agreement, MVBA)} performs consensus on a value $v\in\mathcal{V}$ with arbitrary length, where $v$ can be a value proposed by some process in a domain $\mathcal{V}$ or a default value $\perp\notin\mathcal{V}$ in case correct processes fail to propose the same value. The following properties are satisfied:
\begin{itemize}
\item \textbf{All-Same Validity}: If all correct processes propose the same value $v$, then any correct process decides on $v$. 
\item \textbf{Validity1}: If a correct process decides on $v$, then $v$ was proposed by some process or $v=\perp$.
\item \textbf{Validity2}: If no correct process ever proposed $v$, then no correct process will decide on $v$.
\item \textbf{Agreement}: No two correct processes decide on different values.
\item \textbf{Termination}: Every correct process eventually decides.
\end{itemize}
\end{definition}

Correia et al.~\cite{correia2006consensus} showed that an asynchronous multi-valued consensus could be built on top of binary consensus. The main procedure of the protocol for each process $p_i$ is:
\begin{enumerate}
\item \textbf{\textit{(Asynchronously) reliable broadcast}} its initial value and wait until $2f+1$ initial values have been delivered. The received values will be stored in a vector $V_i$. Note that \textbf{\textit{reliable broadcast}} guarantees that two processes will not receive different values from the same sender process, but the vectors might be different for each process because the first $2f+1$ initial values do not have to be the same.
\item If in the vector $V_i$ there are at least $f+1$ instances of the same $v$ value in the $2f+1$ initial values received previously , then $p_i$ sets $w=v$ and broadcasts $w$ together with $V_i$. Otherwise, $p_i$ selects $\perp$ and also broadcasts it with $V_i$.
\item If $p_i$ does not receive two different values $w\neq w'$ broadcasted in the last step, and it receives at least $f+1$ messages with the same $w$, then it proposes 1 for the next \textbf{\textit{binary consensus}}; otherwise, it proposes 0.
\item If the \textbf{\textit{binary consensus}} decides on 0, then the multi-valued consensus returns $\perp$. Otherwise, $p_i$ \textbf{\textit{waits}} for ($f+1$) messages with $w$ in case it has not already received them .
\end{enumerate}

After the multi-valued consensus is built on top of binary consensus, we consider building vector consensus on top of multi-valued consensus. The goal of vector consensus is to reach agreement on a vector containing a subset of proposed values. In a Byzantine system, vector consensus is useful only if a majority of its values are proposed by correct processes. Therefore, the decided vector must have at least $2f+1$ values. Vector consensus is similar to \textbf{\textit{interactive consistency}} in a synchronous environment, which also reaches agreement on a vector. The difference is that interactive consistency reaches agreement on a vector with the values proposed by all correct processes, while vector consensus only guarantees that majority of the values are proposed by correct processes. The reason is that in an asynchronous network, it is impossible to ensure that a vector has the proposals of all correct processes due to the possibility of arbitrary delays. Remember that FLP states that one can never distinguish whether a message is lost or delayed in an asynchronous network.
\begin{definition}
\textbf{\textit{Vector Consensus}} performs consensus on a vector containing a subset of proposed values, which satisfies the following properties:
\begin{itemize}
\item \textbf{Vector Validity}: For every correct process that decides on a vector $vec$ of size $n$,
\par - $\forall p_i$, if $p_i$ is correct, then either $vec[i]$ is the value proposed by $p_i$ or $\perp$.
\par - At least $f+1$ elements in $vec$ are proposed by correct processes.
\item \textbf{Agreement}: No two correct processes decide on different vectors.
\item \textbf{Termination}: Every correct process eventually decides.
\end{itemize}
\end{definition}

Correia et al.~\cite{correia2006consensus} implemented the \textbf{\textit{vector consensus}} on top of \textbf{\textit{multi-valued consensus}} using the following steps:
\begin{itemize}
\item Every process $p_i$ sets $r_i=0$ and \textbf{\textit{reliable broadcasts}} an initial message $<in\_msg_i, v_i>$ with initial value $v_i$.
\item {$p_i$} waits until at least $2f+1+r_i$ $in\_msg$ messages have been delivered. If $<in\_msg_j, v_j>$ is delivered, $p_i$ sets $vec_i[j]=v_j$, otherwise it sets $vec_i[j]=\perp$.
\item $p_i$ calls \textbf{\textit{multi-valued consensus}} with $vec_i$ as the input and repeats this step until it returns a value other than $\perp$; if $\perp$ is returned, $p_i$ increments $r_i$ by 1 and goes back to the second step to wait for more $in\_msg$ messages.
\end{itemize}
Note that when \textbf{\textit{vector consensus}} repeats the second step, it does not re-start the second step but rather waits until the required number of messages has cumulatively been received since the beginning. As $in\_msg$ messages are \textbf{\textit{reliable broadcasted}}, all correct processes will eventually receive the same $in\_msg$ messages and build an identical $vec$. When enough processes propose the same $vec$ to \textbf{\textit{multi-valued consensus}}, $vec$ will be decided, and then \textbf{\textit{vector consensus}} immediately decides. Once \textbf{\textit{vector consensus}} is implemented, \textbf{\textit{atomic broadcast}} can be built on top of it. Recall the definition of \textbf{\textit{atomic broadcast}} (definition~\ref{def:atomic}) and \textbf{\textit{reliable broadcast}} (definition~\ref{def:reliablebroadcast}), and \textbf{\textit{atomic broadcast}} can be seen as \textbf{\textit{reliable broadcast}} plus the \textit{integrity} and \textit{total order} properties. We can utilize a Secure Hash Digest function~\cite{fips2012180} to achieve \textit{integrity}, thanks to the properties of a secure hash function. Then, we only need to guarantee the order of all delivered messages to implement \textbf{\textit{atomic broadcast}}.

Correia et al.~\cite{correia2006consensus} give their solution for building \textbf{\textit{atomic broadcast}} on top of \textbf{\textit{vector consensus}}. The protocol is as follows:
\begin{enumerate}
\item First, use a Hash function to guarantee that if a malicious process tries to call reliable broadcast twice with the same message, reliable broadcast delivers the message only once. Once a message is delivered to a process by reliable broadcast, the process adds it to a message pool \textit{R\_delivered}, which is initially $\emptyset$. Note that each message in an atomic broadcast is assigned a unique sequence number.
\item When \textit{R\_delivered} $\neq\emptyset$, $p_i$ sets a vector $H_i$=Hashes of the messages in \textit{R\_delivered}. Then, the $H$ vector from each process is sent to the vector consensus protocol. The vector consensus protocol decides on a vector $X_i$ after receiving at least $2f+1$ $H$ vectors from different processes, where $X_i$ is a vector containing different $H$ vectors as elements.
\item After receiving $X_i$, every process $p_i$ waits for all messages $M$ whose $Hash(M)$ is in at least $f+1$ cells in $X_i$, where a cell is an element of an $H$ vector. Then, $p_i$ adds these messages to a message pool \textit{$A\_deliver_i$}.
\item $p_i$ atomically delivers \textit{$A\_deliver_i$} messages in a deterministic order for each message. Then, $p_i$ removes \textit{$A\_deliver_i$} messages from \textit{$R\_deliver_i$}.
\end{enumerate}

How does this atomic broadcast work? W.l.o.g., we consider one message $M$ in \textit{$R\_deliver_i$} to be atomically broadcasted. Because reliable broadcast is used, it is guaranteed that in some execution, the hash $Hash(M)$ will be put in $H$ vectors by all correct processes, and the vector consensus protocol will decide on a vector $X_i$ that includes at least $f+1$ entries with $Hash(M)$, since $X_i$ has at least $2f+1$ elements and there are at most $f$ malicious processes. Therefore, $Hash(M)$ will be added to \textit{$A\_deliver_i$}, and $M$ will be delivered by every correct process. To ensure the total order property, note that before each atomic broadcast decides, an execution of the vector consensus protocol is done, and each vector consensus instance is identified by a unique sequence number assigned to each atomic broadcast. Each vector consensus instance decides on a value only when a sufficient number of processes (at least $2f+1$) propose the same vector; it is impossible for a different vector to be decided because there cannot be $2\times (2f+1)$ processes in total . Because each vector consensus instance terminates in the same order, each atomic broadcast instance will decide in the same order. Combined with the deterministic order of messages within one atomic broadcast instance, the total order of each message is guaranteed.

\textbf{$\triangleright$ Remark.} Atomic broadcast could also be directly built on top of MVBA, but the properties of vector consensus would be realised nonetheless . We have shown that atomic broadcast can be built on top of vector consensus and multi-valued consensus, which can in turn be implemented on top of binary consensus. Readers can refer to ~\cite{correia2006consensus} for the complete protocol and proofs. Although ~\cite{correia2006consensus} might not be practical due to its complexity and latency, it demonstrates that asynchronous consensus primitives can be reduced from one level to another, which can finally be reduced to the level of asynchronous binary consensus. Another question might be, what if the output of MVBA is from dishonest ones? The solution could be adding external predicate/validity check to guarantee only valid values could be selected as the output \cite{cachin2001secure}. This scheme is also adopted in \cite{guo2020dumbo}, which will be introduced later in this section.

\subsection{Building Atomic Broadcast on Binary Agreement and Asynchronous Common Subset}
\label{sect:ACS}
The asynchronous common subset (ACS) is arguably one of the most practical frameworks for asynchronous BFT. ACS was first proposed by Ben-Or et al.\cite{ben1994asynchronous}, and recently made practical by HoneyBadgerBFT\cite{miller2016honey} and BEAT~\cite{duan2018beat}.
\begin{definition}
\textbf{Asynchronous Common Subset (ACS)} assumes a setting with a dynamic predicate $Q$ that assigns a binary value to each process, and ensures that a correct process $i$ will (eventually) be assigned the value $Q(i)=1$. Every process is guaranteed to agree on a subset of at least $2f+1$ processes for whom $Q(i)=1$. ACS satisfies:
\begin{itemize}
\item \textbf{Validity}: If a correct process $p_i$ outputs a set $SubSet_i$, then $|SubSet_i|\geq 2f+1$ and $SubSet_i$ contains the input of at least $f+1$ correct processes.
\item \textbf{Agreement}: If a correct process outputs $SubSet$, then every node outputs $SubSet$.
\item \textbf{Totality}: If $2f+1$ correct processes receive an input, then all correct processes produce an output.
\end{itemize}
\end{definition}
The original ACS protocol by Ben-Or~\cite{ben1994asynchronous} is quite simple. We now give a brief description.

\noindent\rule[0.25\baselineskip]{\textwidth}{0.5pt}\\
(0) Execute an ABA for each process to determine whether it will be in the agreed set.\\
(1) For process $p_i$, participate in $ABA_j$ with input 1 for each $p_j$ for which $p_i$ knows $Q(j)=1$.\\
(2) Upon completing $2f+1$ ABA protocols with output 1, enter input 0 to all ABA protocols for which $p_i$ has not yet entered a value.\\
(3) Upon completing all $n$ ABA protocols, output $SubSet_i$ which is the set of all indices $j$ for which $ABA_j$ had output 1.\\
\noindent\rule[0.25\baselineskip]{\textwidth}{1pt}\\
Ben-Or's ACS was designed for multiparty computations. Miller et al. adapted it for use in asynchronous BFT protocols in HoneyBadgerBFT~\cite{miller2016honey}. HoneyBadgerBFT's ACS uses the same $n$-ABA parallel execution as Ben-Or's ACS, but makes an adjustment to the predicate $Q$. HoneyBadgerBFT's ACS adds a RBC phase before ACS. For each process $i$, it first inputs value $v_i$ to $RBC_i$. Upon delivery of $v_j$ from $RBC_j$, if input has not yet been provided to $ABA_j$, then it provides input 1 to $ABA_j$. After that, the procedure is identical to Ben-Or's ACS. In other words, the predicate is implemented by $RBC$. After the $SubSet_i$ is output, i.e., all instances of ABA have completed, $p_i$ waits for $v_j$ from $RBC_j$ such that $j\in Subset_i$.

\begin{figure}[ht]
 \centering
 \setlength{\abovecaptionskip}{-1em}
 \includegraphics[height=2.0in]{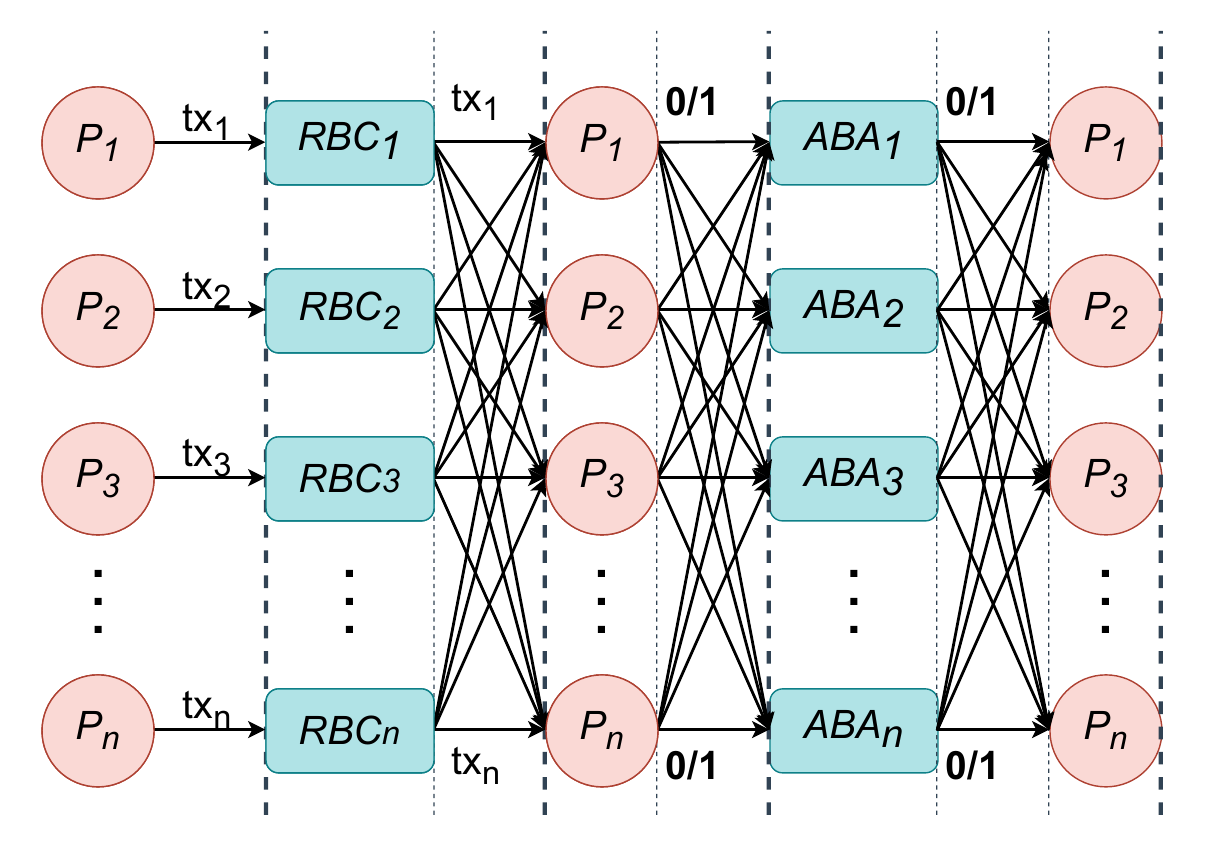}
 \caption{ACS of HoneyBadgerBFT}
 \Description{ACS}
 \label{fig:ACS}
\end{figure}

HoneyBadgerBFT's ACS solution executes $n$ RBC and $n$ ABA in parallel. In the RBC phase, each process broadcasts its value; in the ABA phase, the $i$-th $ABA$ is used to agree on whether $p_i$'s value has been delivered in the RBC phase. Once the $SubSet$ is output by ACS, every process will be aware of whose value is already determined to be accepted by BA, and will wait for these values to be delivered by RBC. The RBC guarantees these values will eventually arrive, so every process can build the same vector $vec$ that contains all these values. Therefore, an atomic broadcast can be achieved by having every process input its value to RBC and delivering values in $vec$ in a deterministic sequence.

\textbf{$\triangleright$ Remark.} ACS and MVBA are two frameworks for implementing asynchronous atomic broadcast. Both ACS and MVBA can be built on ABA. Note that the atomic broadcast built on MVBA could be used for a scenario in which, during a specific time period, only a small subset of processes (e.g., one or two) are broadcasting, because there is no assumption on how many processes trigger RBC. However, for ACS in HoneyBadgerBFT, there is a requirement that $n$ processes propose values in parallel, and each correct process must propose something during the RBC phase, as illustrated in Fig.\ref{fig:ACS} (otherwise, for example, if only a small number of processes trigger RBC, step (2) in ACS that requires $2f+1$ ABA with output 1 will never be satisfied). Intuitively, HoneyBadgerBFT's ACS is suitable for the ``batch'' settings, particularly in blockchain scenarios. In addition, the complexity of these two protocols is different, and MVBA is considered less efficient than ACS in blockchain scenarios\cite{miller2016honey}. Since we focus on how these BFT protocols work , we do not delve into the details in this paper.

\subsection{Asynchronous BFT for Blockchain}
\label{sect:asyncblockchain}
The increasing popularity and interest in blockchain research have reignited focus on traditional SMR and consensus protocols, as discussed in Section~\ref{sect:partially sync BFT for blockchain}. Additionally, recent research trends in asynchronous BFT protocols have also gained attention, as it has become increasingly clear that partially synchronous solutions do not suit all scenarios.
The motivation for implementing asynchronous BFT protocols in blockchains instead of partially asynchronous protocols is that conventional wisdom, like PBFT and its variations, rely critically on network timing assumptions, and only guarantee liveness after GST, which is considered ill-suited for the blockchain scenario as blockchain systems face a much stronger adversary~\cite{miller2016honey}.

\subsubsection{HoneyBadgerBFT: An ACS-based Practical Asynchronous BFT for Blockchain}\
\label{sect:HoneyBadger}

\noindent As an optimisation for cryptocurrency and blockchain scenarios, HoneyBadgerBFT considers network bandwidth to be the scarce resource but computation to be relatively ample. Therefore, HoneyBadgerBFT could take advantage of cryptography that is considered expensive in classical fault-tolerant settings.

The authors first argue that an adversary could thwart PBFT and its variations. In PBFT, the designated leader is responsible for proposing the next batch of transactions at any given time, and if progress isn't made, either due to a faulty leader or a stalled network, the nodes attempt to elect a new leader. Because PBFT relies on a partially asynchronous network for liveness, one can construct an adversarial network scheduler that violates this assumption, resulting in PBFT making no progress at all. For example, when a single node has crashed, the network scheduler delays messages from all the newly elected correct leaders to prevent progress until the crashed node is the next one to be elected as the leader, then the scheduler immediately heals the network and delivers messages very rapidly among correct nodes. However, as the next leader has crashed, no progress is made in this case . In addition, any partially synchronous protocols that rely on timeouts have the problem that they are very slow when recovering from network partitions, because the delay assumption and timeout implementation are bound by a polynomial function of time. The core BFT protocol HoneyBadgerBFT uses is ACS, which we have introduced before, so we do not repeat the core idea here. Next, we take a look at how HoneyBadgerBFT adapts to the blockchain scenario.

The theoretical feasibility of ACS has been demonstrated in~\cite{ben1994asynchronous,cachin2001secure}, but the utilisation of ACS could cause a censorship problem in blockchains. As described, the ACS used by HoneyBadgerBFT has an RBC phase and an ABA phase. An adversary might find a transaction $tx_i$ (which is delivered during RBC, but the attacker receives it earlier than others) that is disadvantageous to them and attempt to exclude it and whoever proposed it. The adversary can then quickly control $f$ faulty nodes to input 0 to $ABA_i$ for which $i$ is the sender of this transaction, in an attempt to stop it from being agreed upon. Next, the adversary delays $RBC_i$ that corresponds to $tx_i$, resulting in correct nodes receiving other $2f+1$ (including those from the adversary) transactions first. As a result, $tx_i$ is censored as the index of its sender will not be included in $SubSet$, and all correct nodes will ignore $tx_i$.

The authors first improve the efficiency by ensuring nodes propose mostly disjoint sets of transactions. As mentioned before, in HoneyBadgerBFT's ACS, every correct process should propose something due to the $n$-parallel design and there is no leader to handle the transactions. Therefore, every process receives all transactions and stores them locally, and randomly chooses a sample such that each transaction is proposed by only one node on average. The authors then propose to use ($f+1,n$)-threshold encryption such that the network nodes must work together to decrypt transactions , i.e., only after $f+1$ (at least one honest) nodes compute and reveal decryption shares for a ciphertext, can the plaintext can be recovered. Prior to this, the adversary knows nothing about the encrypted transactions. Finally , the authors argue that selecting a relatively large batch size can improve efficiency, because one ABA instance is needed for every batch proposed by a process, regardless of the size.

\subsubsection{BEAT: Adaptations and Improvements on HoneyBadgerBFT}\
\label{sect:beat}

\noindent HoneyBadgerBFT provides a novel atomic broadcast protocol and the basic adaptation of using it in a cryptocurrency-like blockchain. BEAT~\cite{duan2018beat} implements a series of adaptations that adapt to different blockchain scenarios, e.g., append-only ledger (BFT storage) and smart contract (general SMR). Although these two types of blockchains share the same security requirements, their storage requirements are different. The former focuses on decentralised storage, and might use redundancy-reducing mechanisms such as allowing servers to keep only fragments~\cite{fanN2021dr,li2021lightweight,wu2021me,song2021supply}, while in contrary, the smart contract form has a more strict requirement on storage as each participant is expected to keep a full redundant copy of all contract states to support contracts executing ~\cite{wood2014ethereum,natanelov2022blockchain,zhang2018smart}.

To meet the requirements of these two major types of application scenarios, BEAT proposes five protocol instances, BEAT0-BEAT4. BEAT0 utilizes a more secure threshold encryption, and a threshold coin-flipping instead of threshold signature as used in HoneyBadgerBFT, and adopts flexible and efficient erasure-coding support for reliable broadcast. The erasure-coding reduces the bandwidth requirement as the sender only has to send a fragment to any other process instead of the whole batch. BEAT1 replaces the erasure-coded broadcast used in BEAT0 with another erasure-coded broadcast protocol to reduce latency when the batch size is small. BEAT2 opportunistically offloads encryption operations of the threshold encryption to the clients to reduce latency, and achieves causal order~\cite{duan2017secure} that ensures a ``first come, first served'' manner, which could be useful for financial payments. BEAT0-BEAT2 are designed for general SMR-based blockchains. Besides, BEAT3 and BEAT4 are suitable for BFT storage. Since BFT storage usually only requires each process to store only a fragment, BEAT3 uses a fingerprinted cross-checksum-based~\cite{hendricks2007verifying} bandwidth-efficient asynchronous verifiable information dispersal protocol to broadcast messages, which significantly reduces bandwidth consumption. BEAT4 is optimised for scenarios where clients frequently read only a fraction of the stored transactions, e.g., a portion of a video. It further extends fingerprinted cross-checksum techniques together with a novel erasure-coded asynchronous verifiable information dispersal protocol to reduce access overhead.

\subsubsection{Dumbo: Reduction on ABA and Carefully-Used MVBA}\
\label{sect:dumbo}

\noindent It is found in Dumbo~\cite{guo2020dumbo} that the $n$ ABA instances terminate slowly when $n$ gets larger and the network is unstable, and the slowest ABA instance determines the running time of HoneyBadgerBFT. Dumbo proposes an improved protocol instance, Dumbo1, which only needs to run $~\kappa$ ABA instances, where $\kappa$ is a security parameter independent of $n$. $\kappa$ is an adjustable parameter; if $n$ is relatively large, the probability that none of $\kappa=\kappa_0$ nodes is honest is at most $(1/3)^{\kappa_0}$ which is safe to neglect, so $\kappa_0$ could be selected; otherwise, $\kappa=f+1$ is used . Committee election is implemented using pseudo-randomness.

The ACS of Dumbo1 is revised from that of HoneyBadgerBFT, as illustrated in Fig. \ref{fig:Dumbo1}. After the $n$ RBC phase, every process executes committee election and $\kappa$ processes are elected as the committee $C$, among which at least one process is honest. If an honest process $c_j$ in the committee receives $2f+1$ values from the RBC phase, it initiates an index-RBC. Each index-RBC is used by $c_j$ to broadcast $S_j$, the set of $2f+1$ RBC instance $c_j$ has already received values from . Finally , $\kappa$ ABA instances are still needed because there might be two or more honest nodes in the committee, resulting in any two honest nodes $\notin C$ receiving $S_i\neq S_j$, which are both sent from honest committee members (due to asynchrony and different latency). The $\kappa$ ABA instances are then used to reach agreement on $S$ that contains $2f+1$ values. For a process $p_i$ $\notin C$, once it receives the indices in $S_j$ from $c_j$ that indicate which $2f+1$ values in the RBC phase should be agreed upon , it waits until all these values are received. Note that $p_i$ might have received $S_i\neq S_j$; in this case, it will input 1 to only one $ABA_i$ once the $2f+1$ values in $S_i$ are all received. Then, once $p_i$ has received 1 from any $ABA_j$, it inputs 0 to all other $\kappa-1$ ABA instances. As a result, all correct processes will agree on the same $2f+1$ values and wait for them to be delivered. Because in $C$ at least one process is honest, and this honest process must have received all input values corresponding to the index set $S_i$. Thus, following the properties of RBC, all other honest nodes will also eventually receive those values.

\begin{figure}[h]
 \centering
 \includegraphics[width=4.3in]{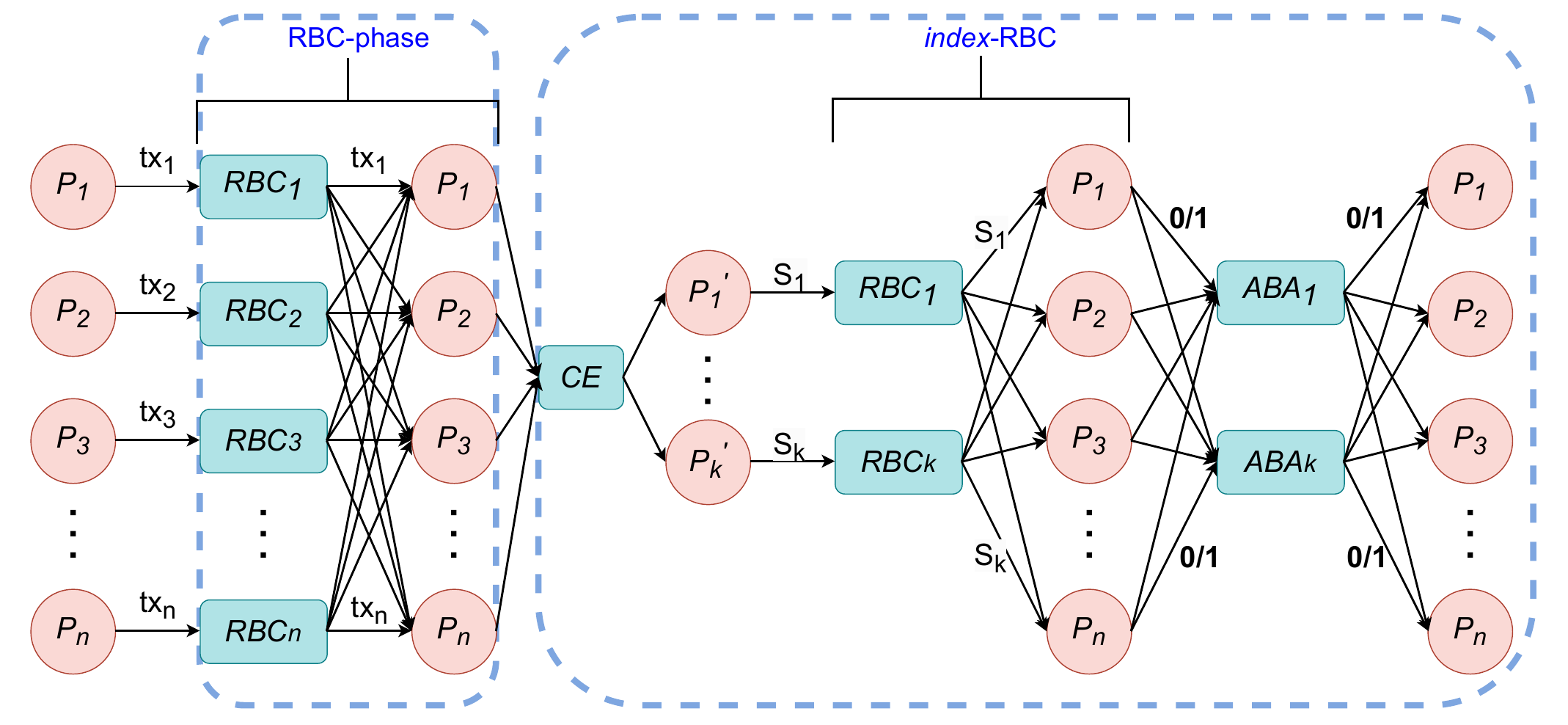}
 \caption{ACS of Dumbo1}
 \Description{Dumbo1}
 \label{fig:Dumbo1}
\end{figure}

In addition, Dumbo also gives another instance, Dumbo2. As we introduced in \ref{sect:MVBA}, it has been demonstrated that asynchronous BFT could be built on MVBA, although it is considered complex and impractical in~\cite{miller2016honey}. The authors of Dumbo examine MVBA again and point out that if the message size of the input value becomes small, the situation changes. Therefore, the authors propose Dumbo2 that utilises MVBA in another way: to agree on a set of indices rather than the proposed values themselves. To guarantee that the output of MVBA is valid (the output might be a fake set of indices for which some processes with those indices are faulty and have never broadcast), Dumbo2 uses a variation of RBC called provable RBC, which is realised by combining RBC and threshold signatures. This guarantees that once a message is delivered by provable RBC, it is indeed the case that at least $f+1$ processes have received this message. Once this set of indices is decided by MVBA, it is guaranteed that at least one honest process has received all values corresponding to these indices. The following idea is similar to that of Dumbo1.

\textbf{$\triangleright$ Remark.} MVBA and ACS are two techniques for realising asynchronous BFT. Both of these two mechanisms could be reduced to ABA, which is a randomised consensus protocol. Generally, a common coin is used to facilitate the termination of ABA. Threshold cryptography is widely used in all these protocols, including for generating the common coin and building provable RBC. The MVBA protocols were originally used in asynchronous BFT solutions~\cite{cachin2001secure,cachin2002secure,correia2006consensus} to agree on the proposed values directly, which was found impractical for blockchain scenarios \cite{miller2016honey}. Then, a number of works\cite{miller2016honey,duan2018beat, guo2020dumbo, liu2020epic,liu2021mib} utilise ACS, which was initially used for multiparty computation~\cite{ben1994asynchronous}, to realise atomic broadcast in blockchain scenarios. In addition, inspired by the notion of ACS, \cite{guo2020dumbo} uses MVBA to agree on the indices of RBC rather than the proposed values themselves to improve performance. The core idea of ACS-based asynchronous BFT is simple: input values to RBC first, and then use ABA to agree on which values should be accepted and wait until these values are eventually delivered by RBC. More optimisation works have been done on ACS and asynchronous BFT~\cite{gao2022dumbo,duan2023practical}; due to content limitation and since our focus is the core ideas of these protocols, we do not introduce them in this paper.

\section{Tree and DAG-based BFT}
\label{sect:treeandDAG}
Conventional BFT protocols like PBFT are considered to have scaling problems, as well as recent advanced BFT protocols for blockchain like HotStuff~\cite{li2021scalable,stathakopoulou2022state}. Researchers have made efforts to improve scalability by applying sharding and partition techniques, as we introduced in Section \ref{sect:scalingbft}. These solutions change the deployment of participants, and achieve scaling at the cost of reduced resilience. A t ree is the most intuitive data structure that could improve scaling, and it has been utilised in many cryptocurrency blockchains\cite{kogias2016enhancing,kogias2018omniledger} to address the bottleneck at the leader by organising processes in a tree topology. In addition, the more parallel-looking DAG structure is used in blockchain~\cite{WU2022102720} to organise transactions and ledgers , allowing participants to propose different blocks simultaneously. Since we focus on BFT protocols, we introduce the tree and DAG based solutions in BFT in detail.

\subsection{Tree and DAG-based Message Dissemination and Memory Pool Abstraction}

Tree and DAG are considered more efficient in nature than a chain due to their parallel structure for organising transactions and blocks~\cite{WU2022102720}. Researchers are also inspired by these structures and have attempted to utilise them to improve efficiency in BFT. Recently, some proposals have optimised  data dissemination schemes~\cite{neiheiser2021kauri} and memory pool~\cite{danezis2022narwhal} as extensions of existing BFT protocols.

\subsubsection{Kauri: Tree-based Message Dissemination and Aggregation}\

\noindent Kauri~\cite{neiheiser2021kauri} is typically designed as an extension of HotStuff~\cite{hotstuff2019yin}. Recall the communication model of HotStuff (see Section \ref{sect:hotstuff}), which has four phases. In each phase, the leader performs dissemination and aggregation procedures during each phase, as shown in Fig.\ref{fig:hotstuff}. Kauri assumes that most commonly the number of consecutive faults is small and the network is partially synchronous. It proposes a tree-based dissemination and aggregation scheme to reduce bandwidth consumption on the leader's side, where the leader is at the tree root. Essentially , Kauri is a speculation-based optimisation mechanism. Because if a process is faulty, all its children could fail to receive the message from the leader. Thus, every process will decide on its own value through a broadcast process anyway, either the value from its parent or a default value $\perp$ (if it receives no value from its parent before a timeout, which is used to guarantee liveness). The tree-based aggregation process is similar: a parent will collect all signatures from its children, form them into an aggregate signature, and send it back to its parent recursively until reaching the tree root. The topology and communication pattern are illustrated in Fig~\ref{fig:kauri}. If a correct leader fails to collect enough signatures, Kauri tries to reconfigure the system a certain number of times (determined by a specific mechanism based on an evolving graph; we do not introduce this detail here) to find a correct tree configuration, and will return to the original HotStuff star topology if no correct configuration can be found.

\begin{figure}[h]
 \centering
  \setlength{\abovecaptionskip}{-0.2cm}  %
 \setlength{\belowcaptionskip}{-0.2em} 
 \includegraphics[width=5.5in]{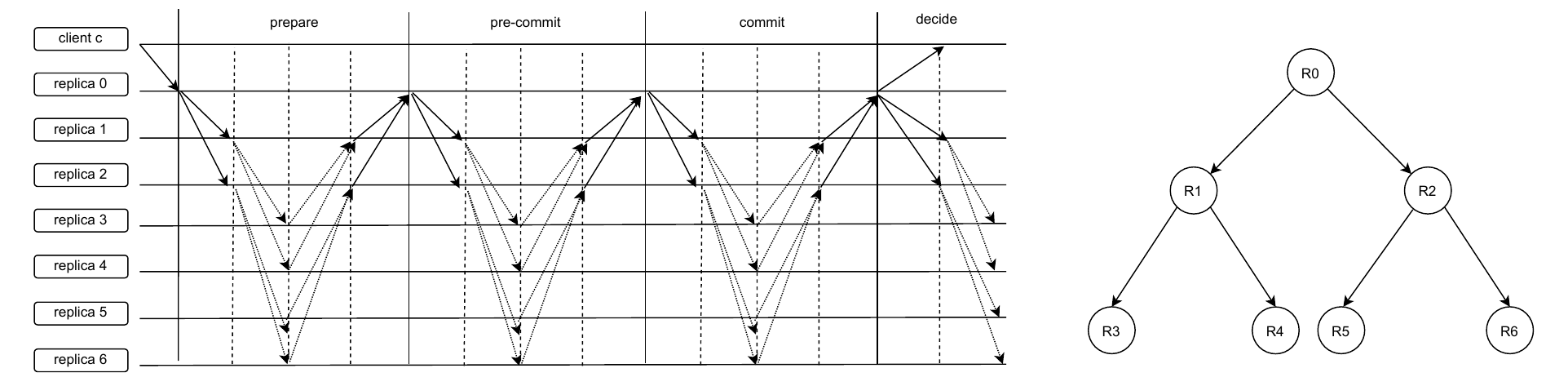}
 \caption{Tree communication pattern and topology}
 \Description{Kauri}
 \label{fig:kauri}
\end{figure}

\subsubsection{Narwhal: DAG-based Memory Pool}\
\label{sect:narwhal}

\noindent The authors in~\cite{danezis2022narwhal} pointed out that some consensus protocols in which the leader is expected to broadcast and collect information lead to uneven resource consumption. This phenomenon is also discussed in Kauri, as we introduced above. In addition, it is found in~\cite{danezis2022narwhal} that consensus protocols group a lot of functions into a monolithic protocol, such as Bitcoin and LibraBFT~\cite{baudet2019state}. In these protocols, the transactions will be shared among participants and then a subset of them are periodically re-shared and committed as part of the consensus protocol, which causes redundancy. Narwhal is proposed as a better memory pool that separates transaction dissemination (message delivery) from ordering in the consensus protocol. This is theoretically feasible; recall that we pointed out in Section ~\ref{sect:MVBA} that atomic broadcast (could be converted to consensus) could be realised by the reliable broadcast protocol plus the total ordering property. Similarly, Narwhal separates the monolithic consensus protocol into reliable transaction dissemination and sequencing, and the sequencing could be realised by any consensus protocol such as HotStuff or LibraBFT. The consensus protocol used only needs to be performed on a very small amount of metadata instead of the complete transactions/block to increase performance.

We first check how Narwhal realises reliable transaction dissemination. Initially, it assumes clients send transactions to different validators (processes). Ideally, a transaction only needs to be sent to one validator to save bandwidth. Each validator will accumulate the transactions to form a block. Then, validators reliably broadcast~\cite{bracha1987asynchronous} each block they create to ensure the integrity and availability of the block. Each block corresponds to a round $r$ (column in the DAG) and contains certificates for at least $2f+1$ blocks of round $r-1$. These certificates are used to confirm that the validator has received these blocks and prove they will be available. In addition, the authors propose that every block includes certificates of past blocks, from all validators. By doing this, a certificate refers to a block plus its full causal history (i.e., all blocks directly and indirectly connected to it); therefore, when a block is committed, all blocks in its causal history are committed simultaneously. An example of Narwhal is shown in Fig.~\ref{fig:Narwhal}. Even if the network is unstable and asynchronous, the DAG-based memory pool continues to grow. Because reliable broadcast is used, every validator eventually observes the same DAG and the same certificates.
The Narwhal memory pool could be integrated with various BFT protocols. Using HotStuff as an example, in the original HotStuff (see Section~\ref{sect:hotstuff} a leader proposes a proposal that is certified by other validators. To be integrated with Narwhal, a leader proposes to commit one (or several) certificates corresponding to the block(s) created by Narwhal. Due to Narwhal's design , once a block is committed, all blocks in its causal history are committed as well . In Fig.~\ref{fig:Narwhal}, once C1 is committed, all blocks in its causal history (blue vertices) are committed, and the same applies to the yellow and green vertices. Hence, a partial order is built: green vertices (C1)$\leftarrow$ red vertices (C2)$\leftarrow$ blue vertices (C3). To build a total order, one can deploy any deterministic algorithm on the partial order, e.g., by sorting by hash.

\begin{figure}[h]
 \centering
 \includegraphics[width=3.8in]{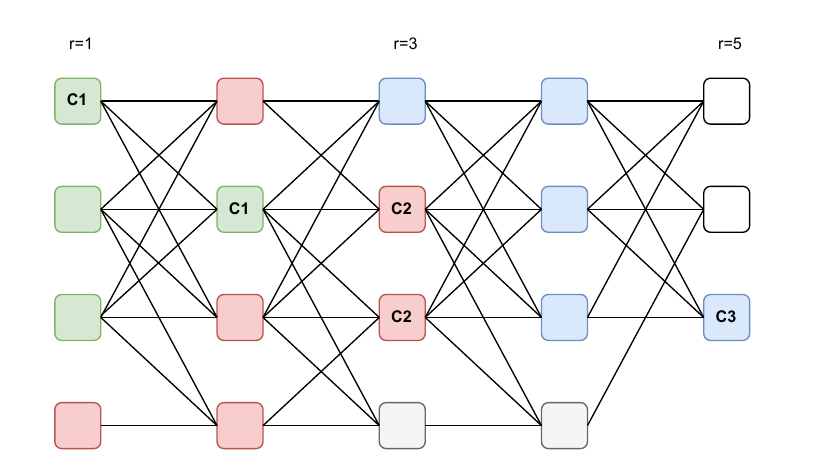}
 \caption{Example of Narwhal memory pool with HotStuff, where c2 in red indicates the vertex committed in second round and all its former connected vertices after (not committed) by c1, and so forth.}
 \Description{Narwhal}
 \label{fig:Narwhal}
\end{figure}

\subsection{DAG-based BFT}

\noindent Narwhal builds a DAG-based memory pool that abstracts the communication and could be integrated with any other partially synchronous BFT protocols. Nevertheless, asynchronous BFT protocols are also proposed that could be directly deployed on the DAG abstraction. These protocols share a similar core idea: utilize the DAG abstraction with reliable broadcast to build the communication history, then use threshold cryptography to build a common coin to select an anchor block every several rounds (i.e., in each column in the DAG) to establish a partial order, and finally use a deterministic algorithm to build the total order.

\subsubsection{DAG-Rider: DAG-based Atomic Broadcast}\ 

\noindent Similar to the DAG-based memory pool, DAG-Rider~\cite{keidar2021all} also builds a DAG to abstract communication history. The DAG of DAG-Rider is similar in structure , in which each vertex also refers to $2f+1$ vertices in the last round (called strong edges in DAG-Rider). The difference is, these $2f+1$ edges have been included in a vertex when this vertex is broadcasted, rather than being built after an availability certificate is created . In addition, with the goal of implementing an atomic broadcast protocol (see definition~\ref{def:atomic}), the validity property should be met . Therefore, for each vertex $v$ in round $r$, in addition to the $2f+1$ strong edges referring to the vertices in round $r-1$, it also refers to at most $f$ vertices in round $r'<r-1$ such that otherwise there would be no path from $v$ to them. The weak edges are used to guarantee validity, as every message that has been broadcasted must be delivered and needs to be included in the total order.

DAG-Rider assumes that each process broadcasts an infinite number of blocks (vertices). Once a process $p_i$ invokes atomic broadcast for a vertex $v$, it reliably broadcasts $v$. The round and source, a.k.a. sender, of the vertex are included in the broadcasted message so that the receiver process knows its location in the DAG. Once $p_i$ receives a vertex, it adds it to its buffer. Then, it continuously goes through its buffer to see if there is any vertex $v$ for which the DAG contains all the vertices that $v$ has a strong or weak edge to. Once this requirement is met , $v$ can be added to the DAG. Once $p_i$ has at least $2f+1$ vertices in the current round, it advances to the next round by creating and reliably broadcasting a new vertex $v'$, which is the new atomically broadcasted one ($p_i$ has an infinite number of vertices to be broadcasted).

Once the local DAG is built, each process needs to interpret it. The DAG is divided into waves, each of which consists of four consecutive rounds. This is because after three rounds of all-to-all sending and collecting accumulated sets of values, all correct processes have at least $2f+1$ common values, by the \textit{common-core} abstraction~\cite{canetti1996studies}. Therefore, at least $2f+1$ vertices $\in V$ in the last round of a wave have a strong path to at least $2f+1$ vertices $\in U$ in the first round of the same wave. The $2f+1$ quorum intersection guarantees that if an anchor vertex in (the first round of) wave $w$ is committed by some process $p_i$, then for every process $p_j$ and for every anchor vertex of a wave $w'>w$, a strong path exists between these two anchor vertices. In each wave $w$, an anchor vertex will be selected by a globally known common coin, and it is committed by a process $p_i$ if in $p_i$'s local DAG there are at least $2f+1$ vertices in the last round of $w$.
The common-core property guarantees that once an anchor vertex is selected, for any process $p_i$, it has $\frac{2f+1}{3f+1}\approx \frac{2}{3}$ probability that it has this vertex in its local DAG. It should be noticed that some processes might not have this anchor vertex in their local DAG, so they simply advance to the next wave. However, it should be ensured that all correct processes commit the same anchor vertices. In other words, if some processes committed an anchor vertex, those who do not have it in their local DAG must commit it later. To ensure this, when $p_i$ commits an anchor vertex $v$ in some wave $w$, and there is a strong path from $v$ to $v'$ such that $v'$ is an uncommitted anchor vertex in a wave $w'<w$, then $p_i$ commits $v'$ in $w$ as well. It is guaranteed by quorum intersection that if any process ever commits $v'$ there must be a strong path, otherwise $v'$ is not committed by anyone. The anchor vertices committed in the same wave are ordered by their round wave numbers, so the anchor vertices of earlier waves are ordered before the later ones. An example of ordering is illustrated in Fig.\ref{fig:DAG-Rider}. Fig.\ref{fig:DAG-Rider} shows the local DAG of $p_1$. The highlighted vertices $v_2$ and $v_3$ are the anchor vertices of waves 2 and 3, respectively. $v_2$ is not committed in wave 2 since there are fewer than $2f+1$ vertices in round 8 with a strong path to $v_2$, but it is met in wave 3 since there are $2f+1$ vertices in round 12 with a strong path to $v_3$. Since there is a strong path from $v_3$ to $v_2$ (highlighted), $p_i$ commits $v_2$ before $v_3$ in wave 3.

\begin{figure}[h]
 \centering
 \setlength{\abovecaptionskip}{-0.2cm}  %
 \setlength{\belowcaptionskip}{-0.2em} 
 \includegraphics[width=4.8in]{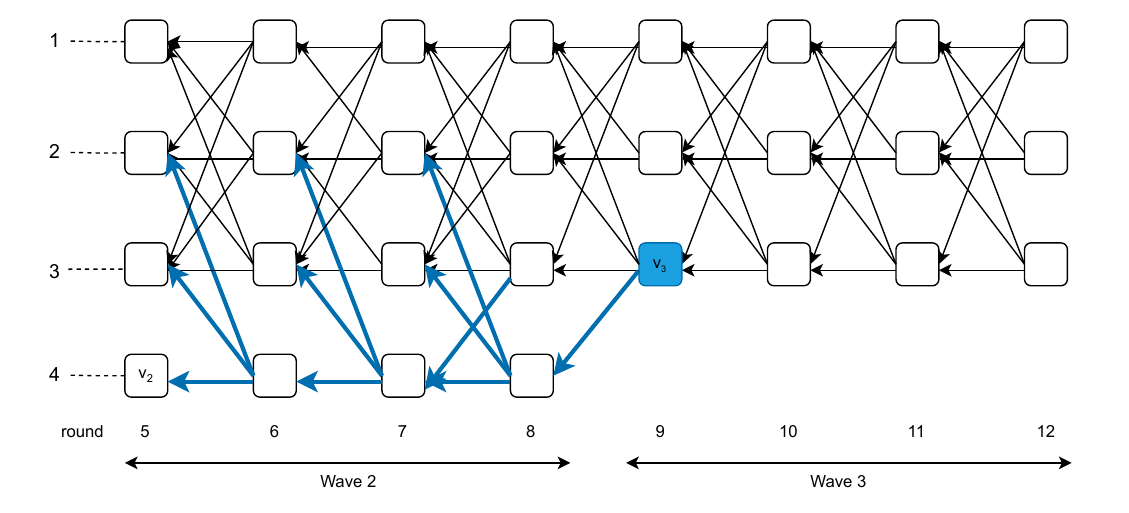}
 \caption{Illustration of DAG-Rider}
 \Description{DAG-Rider}
 \label{fig:DAG-Rider}
\end{figure}

It's important to note that although the \textit{common-core} property guarantees the commit probability, it could be attacked by an adversary who can fully control the network. Therefore, the common coin should be unpredictable. This motivation is similar to the unpredictable property in HoneyBadgerBFT and Dumbo. In DAG-Rider, the authors also use threshold signatures, and the elected anchor vertex in a wave $w$ can only be revealed after the processes complete $w$. Since reliable broadcast is used, every process will eventually observe the same DAG. After the anchor vertices and all vertices in its causal history are committed, the remaining step is to add some deterministic ordering scheme to order the causal history. Finally , every process only needs to atomically deliver the causal history vertices one by one, which is guaranteed to be the same for everyone. It's important to note that atomic broadcast guarantees that every message, if broadcasted, will be delivered. Therefore, if one utilises DAG-Rider for blockchain, an external predicate is needed to check the validity of blocks or transactions.

\subsubsection{Tusk: The Asynchronous Consensus on Narwhal}\

We introduced in Section ~\ref{sect:narwhal} that the Narwhal DAG memory pool could be integrated with various partially synchronous BFT protocols to propose several anchor blocks as well as their causal history to be committed. However, these partially synchronous BFT protocols make no progress when the network is asynchronous, such as under DDoS attacks. Therefore, the authors propose an asynchronous consensus protocol called Tusk on the DAG. The core idea is similar to that of DAG-Rider: the vertices are broadcasted by reliable broadcast to build the DAG (in Tusk this step has been done by Narwhal), then use threshold signatures to generate a common coin that elects an anchor vertex every several rounds, and processes try to commit the anchor together with all vertices in its causal history to build a partial order, then apply any deterministic scheme to achieve total ordering. However, Tusk is a total ordering protocol on Narwhal rather than an atomic broadcast protocol, which means it does not guarantee that once a vertex is broadcasted, it is ensured to be delivered and ordered. Therefore, in Tusk no weak path exists. If a vertex arrives very late and no certificate of availability is signed, it could be an orphan vertex and thus ignored.

In Tusk, a process (a.k.a. validator in Narwhal) interprets every three rounds of the DAG in Narwhal as a consensus instance. The links in the first two rounds are interpreted as all-to-all message exchange and the third round produces a common coin to elect a unique vertex from the first round to be the leader. To reduce latency, the third round is combined with the first round of the next consensus instance. The goal of this interpretation process is the same as DAG-Rider in that, with a constant probability, it aims to safely commit the anchor of each instance. Once an anchor is committed, its entire causal history (a sub-DAG) in the DAG is also committed and could be totally ordered by any deterministic ordering scheme such as topological sorting~\cite{kahn1962topological}. The commit rule is simple: a process commits an anchor vertex $v$ of an instance $i$ if its local DAG includes at least $f+1$ nodes in the second round of $i$ with links to $v$. Similar to DAG-Rider, some processes might advance to the next round without observing the elected anchor vertex. To solve this problem and guarantee that every process commits the same anchors, once an anchor vertex $v$ is committed, Tusk also continuously checks if there are any previous anchor vertices in $v$'s causal history, and commits previous anchor vertices if they exist . The commit rule is illustrated in Fig \ref{fig:Tusk}.
The commit rule requires $f+1$, which is different from DAG-Rider; the probability of each anchor meeting the commit rule is $\frac{1}{3}$.

\begin{figure}[tbp]
 \centering
 \setlength{\abovecaptionskip}{0.1em}  %
 \setlength{\belowcaptionskip}{-0.2em} 
 \includegraphics[width=4in]{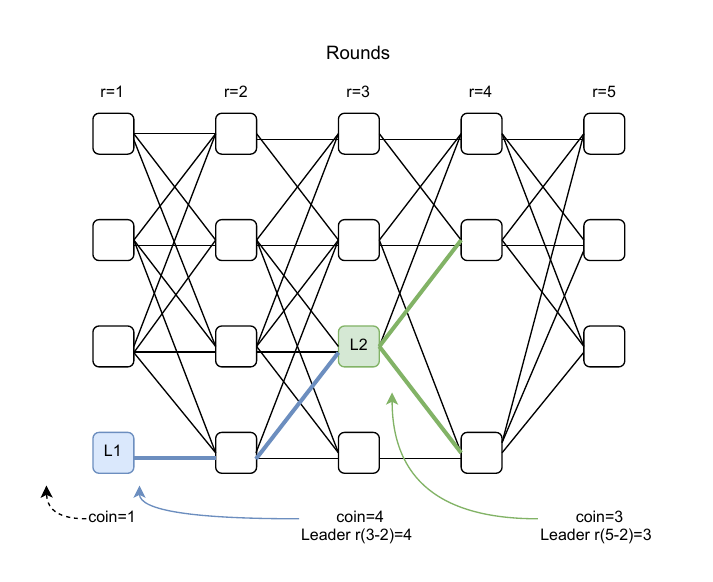}
 \caption{Illustration of commit rule in Tusk}
 \Description{Tusk}
 \label{fig:Tusk}
\end{figure}

\textbf{$\triangleright$ Remark.} All DAG-based BFT protocols are based on a similar idea: to divide consensus or atomic broadcast into two separate protocols: a reliable broadcast protocol and a total ordering protocol. This scheme separates message dissemination from ordering to improve performance and fit better in asynchronous networks. With this approach, each process only needs to check its local DAG to make progress and no further communication is needed. Recently, a DAG-based asynchronous BFT protocol with optimisation for synchronous networks has been proposed~\cite{spiegelman2022bullshark}, which maintains the same core ideas as other asynchronous DAG-BFT protocols but adds timeouts to boost performance in synchronous networks.
The DAG-based abstraction of building a DAG to represent communication and interpreting it locally has been considered before~\cite{dolev1993early,chockler1998adaptive}, but was not realised in the Byzantine setting . Tree~\cite{sapirshtein2016optimal} and DAG-based~\cite{WU2022102720,baird2016hashgraph,churyumov2016byteball,sompolinsky2018phantom} communication abstractions and system structures are also widely used in DLT, which is an extended notion of blockchain where the structure is no longer required to be a ``chain''.

\section{Summaries and Comparisons of Different BFT Protocols}
\label{sect:sumBFT}
\begin{table}[tbp]
 \centering
 \caption{Comparison of core protocols under different timing assumption.}
 \scalebox{0.85}{

  \begin{tabular}{p{13.5em}p{8em}p{7em}p{4.5em}p{9em}}
  \toprule
  Protocol & Timing & Termination & \textit{Resilience} & Components \\
  \midrule
  Byzantine agreement (Sign)\cite{lamport1982byzantine} & Synchronous & Deterministic & \textit{N>2f} & Message, signature \\
  Byzantine agreement (Oral)\cite{lamport1982byzantine} & Synchronous & Deterministic & \textit{N>3f} & Message \\
  Ben-Or's binary consensus\cite{ben1983another} & Asynchronous & Eventually & \textit{N>2f/N>5f} & Message, random coin \\
  Bracha's RBC\cite{bracha1987asynchronous} & Asynchronous & Eventually & \textit{N>3f} & Broadcast message \\
  Vector consensus\cite{correia2006consensus} & Asynchronous & Eventually & \textit{N>3f} & RBC, ABA, MVBA \\
  ACS\cite{ben1994asynchronous}  & Asynchronous & Eventually & \textit{N>3f} & RBC, multiple ABA \\
  PBFT\cite{castro1999practical} & Partial synchronous & Deterministic & \textit{N>3f} & Message, signature \\
  \bottomrule
  \end{tabular}%

  }
 \label{tab:comparisontiming}%
\end{table}%

The comparison of core protocols of synchronous, asynchronous and partially synchronous systems is shown in Table.~\ref{tab:comparisontiming}.

Originating from Lamport's Byzantine generals problem, the distributed fault tolerance research could be classified into three main categories: synchronous, asynchronous and partially synchronous, according to timing assumptions. Synchronous algorithms were the earliest researched. However, due to the requirement of synchrony, synchronous algorithms are generally deployed with clock synchronisation mechanisms. For example, in aviation control systems, a time-triggered bus might be used~\cite{rushby2001bus}, for which clock synchronisation is a fundamental requirement. Based on the synchronous bus, applications that perform safety-critical functions must generally be replicated for fault tolerance. For example, several computers need to perform the same computation on same data, and any disagreement could be viewed as a fault; then comparisons can be used to detect faults. A vital requirement for this approach is that the replicated computers must work on the same data. Therefore, the data distribution needs to be identical for every computer. Byzantine agreement (or Byzantine broadcast) can then be used to achieve this requirement. In addition, due to synchrony, the consensus problem which requires consistency and total order could be solved by executing Byzantine agreement one by one, because everyone knows the order. The improvements and use cases of synchronous Byzantine agreement protocols are shown in Table.~\ref{tab:sync}.


\begin{table}[tbp]
\centering
\caption{Synchronous protocols}
\scalebox{0.55}{
\begin{tabular}{lllll}
\toprule
Protocol & \textit{Resilience} & Components & Functions & Additional Requirements \\
\midrule
Byzantine agreement (Sign)\cite{lamport1982byzantine} & \textit{N>2f} & Message, signature & Guarantees same value for honest participants & None \\
Byzantine agreement (Oral)\cite{lamport1982byzantine} & \textit{N>3f} & Message & Guarantees same value for honest participants & None \\
SAFEbus/TTP\cite{hoyme1992safebus,kopetz1993ttp} & \textit{N>3f} & Message & Builds interactive consistency using Byzantine agreement & Clock synchronisation \\
ScalableBA\cite{king2011breaking} & \textit{N>3f with prob. 1} & Message, secret sharing & Significantly reduces communication complexity & Private channels \\
\bottomrule
\end{tabular}%
}
\label{tab:sync}%
\end{table}%

However, due to the strong requirement of synchrony, synchronous Byzantine protocols have limited utilisation ; they are typically implemented with physically proximate participants, such as in SAFEbus~\cite{hoyme1992safebus} and TTP~\cite{kopetz1993ttp}. However, when the timing assumption changes from synchronous to partially synchronous, the consensus problems could be transformed into an agreement problem that ensures every honest participant has the same value of request plus an ordering problem that ensures every honest participant executes the requests in the same order.
In the representative partially synchronous consensus protocol PBFT, each follower needs $2f$ \textit{prepare} messages that match the \textit{pre-prepare} message before \textit{commit}ting. In the worst case, a follower might receive $f$ \textit{prepare} messages from Byzantine nodes first, and in this scenario , it needs to gather all $3f$ \textit{prepare} messages. This also shows why PBFT needs partial synchrony for liveness: if the last $f$ messages are delayed for a long time before GST, during this period, the follower will be unable to distinguish whether the mismatch is caused by the leader sending equivocal messages or $f$ Byzantine nodes sending fake messages. Then, after GST, followers will detect if the mismatch is caused by the leader and propose a \textit{view change}. The agreement protocol is integrated within the PBFT consensus: the $2f$ \textit{prepare} messages plus 1 \textit{pre-prepare} message guarantee that no honest participant will choose a different value after the \textit{prepare} phase. After the agreement protocol, the \textit{commit} phase ensures that the order of requests is identical for all honest followers.
The quorum size in PBFT is $2f+1$ out of $3f+1$ overall nodes; this is because the leader might be Byzantine, and $2f+1$ ensures that at least $f+1$ messages are from honest participants , hence at least 1 honest follower exists in any quorum. Because the honest follower does not collude with the Byzantine leader, the leader will never be able to separate the system and convince different quorums to commit two different values (Lemma. \ref{lemma:3.2}). However, if the leader cannot send equivocal messages to followers, the quorum size could possibly be reduced in certain scenarios . If we can guarantee that all participants cannot send equivocal messages to others, e.g., by using TEE~\cite{chun2007attested}, the quorum size could be reduced to $f+1$ among $2f+1$ nodes. BFT consensus protocols in the partially synchronous model are listed in Table.~\ref{tab:psync}.

\begin{table}[htbp]
\centering
\caption{Partial synchronous protocols}
\scalebox{0.6}{
\begin{tabular}{lllp{16.785em}l}
\toprule
Protocol & \textit{Resilience} & Components & \multicolumn{1}{l}{Functions} & Additional Requirements \\
\midrule
PBFT\cite{castro1999practical} & \textit{N>3f} & Message, signature & Guarantees unique order for SMR & GST for liveness \\
Zyzzyva\cite{kotla2010zyzzyva} & \textit{N>3f} & Message, signature, crypto proof & Guarantees unique order for SMR & Speculation \\
Q/U\cite{abd2005fault}  & \textit{N>5f} & Client, quorum, crypto validation & Achieves similar SMR functionality using quorums & Shifts tasks to client \\
HQ\cite{cowling2006hq}  & \textit{N>3f} & Client, quorum, certificates & Improves contention and resilience of Q/U & Certificates for solving contention \\
RBFT\cite{aublin2013rbft} & \textit{N>3f} & f+1 BFT instances & Makes BFT scalable & N>3f in each instance \\
CBASE\cite{kotla2004high} & \textit{N>3f} & BFT with relaxed order & Allows concurrent execution of independent requests & A paralleliser for contention safety \\
Eve\cite{kapritsos2012all}  & \textit{N>3f} & Parallel requests with speculation & Speculatively executes requests concurrently & Re-executes when executions diverge \\
EZBFT\cite{arun2019ez} & \textit{N>3f} & Enhanced replicas and multiple BFT & Multiple instances to boost throughput & Clients need to resolve contention \\
ME-BFT\cite{wu2021me} & \textit{N>3f} & BFT, full and light replicas & Reduces all-to-all communication to scale by structure & N>3f in full replicas, CRDT \\
SBFT\cite{gueta2019sbft} & \textit{N>3f} & BFT, threshold/aggregated signature & Reduces all-to-all communication to scale by signature & Collector to collect signatures \\
Mir-BFT & \textit{N>3f} & BFT with parallel partition & Multiple leaders to achieve scalability & Partitions the request hash space \\
CNV\cite{correia2004tolerate}  & \textit{N>2f} & BFT, TMO & Uses TMO to restrict Byzantine failures & TTCB for basic security service \\
A2M\cite{chun2007attested}  & \textit{N>2f} & BFT, attested append only memory & Uses append only memory to avoid equivocation & Realises the append only memory \\
MinBFT\cite{veronese2011efficient} & \textit{N>2f} & BFT, USIG & USIG assigns the sequence & Public key crypto for USIG \\
CheapBFT\cite{kapitza2012cheap} & \textit{N>2f} & MinBFT, CASH & CASH manages counters and can authenticate messages & Falls back to MinBFT after failure \\
Tendermint\cite{buchman2018tendermint} & \textit{N>3f} & PBFT, power-based quorum intersection & Replaces 2f+1 messages with 2f+1 voting power & None \\
Algorand\cite{gilad2017algorand} & \textit{N>3f} & Byzantine agreement, VRF, binary consensus & Byzantine agreement for synchrony, binary consensus for asynchrony & Common coin for binary consensus \\
HotStuff\cite{hotstuff2019yin} & \textit{N>3f} & Threshold signature, quorum certificate & Replaces message quorums with quorum certificates & Sufficient computational power \\
Kauri\cite{neiheiser2021kauri} & \textit{N>3f} & Tree-based dissemination and aggregation & Extension of HotStuff, replaces the communication model & Tree topology \\
\bottomrule
\end{tabular}%
}
\label{tab:psync}%
\end{table}%

\begin{table}[htbp]
\centering
\caption{Asynchronous protocols}
\scalebox{0.6}{
\begin{tabular}{llp{12.915em}p{17.785em}p{13.785em}}
\toprule
Protocol & \textit{Resilience} & \multicolumn{1}{l}{Components} & \multicolumn{1}{l}{Functions} & Additional Requirements \\
\midrule
Ben-Or's\cite{ben1983another} & \textit{N>2f/N>5f} & Message, random coin & Terminates with probability 1 & Binary value \\
Bracha's RBC\cite{bracha1987asynchronous} & \textit{N>3f} & Broadcast message & Ensures that the sender cannot equivocate & None \\
Vector consensus\cite{correia2006consensus} & \textit{N>3f} & RBC, ABA, MVBA & RBC to restrain Byzantine failure, ABA for asynchrony, MVBA for agreement, deterministic order for atomic broadcast & Common coin for binary consensus \\
ACS\cite{ben1994asynchronous}  & \textit{N>3f} & RBC, multiple ABA & RBC to restrain Byzantine failure, multiple ABA for agreement, deterministic order for atomic broadcast & Suitable for batch settings \\
HoneyBadgerBFT\cite{miller2016honey} & \textit{N>3f} & ACS, threshold encryption & ACS for agreement, threshold encryption to resist censorship & Suitable for blockchains \\
BEAT\cite{duan2018beat} & \textit{N>3f} & Asynchronous atomic broadcast & Implements asynchronous broadcast for different scenarios & Fingerprinted cross-checksum \\
Dumbo\cite{guo2020dumbo} & \textit{N>3f} & Provable-RBC, ABA, MVBA, revised ACS, threshold signature & Crypto election to reduce ABA instances; uses MVBA to agree on indices rather than on the requests & \multicolumn{1}{p{15em}}{Suitable for blockchains/batch settings} \\
DAG-Rider\cite{keidar2021all} & \textit{N>3f} & RBC, DAG, threshold signature, common coin & DAG for communication abstraction, common-core for committing anchor blocks and their causal history, deterministic order for atomic broadcast & Each process broadcasts an infinite number of blocks \\
Tusk\cite{danezis2022narwhal} & \textit{N>3f} & RBC, DAG, VRF, common coin & An ordering protocol for the DAG mempool & Unlimited memory \\
\bottomrule
\end{tabular}%
}
\label{tab:addlabel}%
\end{table}%

In the asynchronous settings, due to the restrictions of the FLP impossibility, which states that no ``deterministic" consensus protocol exists, consensus protocols relax the termination requirement i.e., the protocol eventually terminates with probability 1 . In all the asynchronous consensus protocols we introduced , RBC is used to guarantee that any sender cannot send equivocal messages (otherwise the protocols might be too complex). Therefore, beyond the RBC protocol, every node is not a traditional Byzantine node; its malicious behaviour is restrained, and it can only either broadcast same messages or not broadcast at all. Three research paths are being conducted: the first one is to utilise ABA to construct MVBA, then implement vector consensus from MVBA. Once the vector is confirmed, atomic broadcast (equivalent to BFT-SMR) could be built by delivering values in the vector in a deterministic order. Another path is similar; it implements ACS using ABA, then implements atomic broadcast. The difference is that ACS is more suitable for batching, i.e., $n$ participants exchange messages simultaneously. The third pathway does not utilise ABA for randomisation to evade the FLP impossibility result. Instead, it first broadcasts messages by RBC to ensure that everyone will eventually receive the same messages, then uses a DAG to build the message history and applies randomness by utilising a VRF to select anchor vertices with a probability. Finally , a deterministic order could be built on the anchor vertices to form atomic broadcast.
The quorum size in asynchronous consensus is $2f+1$ among $3f+1$ nodes. However, in contrast to partially synchronous consensus where GST exists, in asynchronous consensus a node might never have the chance to receive $3f+1$ messages, and thus cannot rely on it for liveness. Therefore, once a node observes $2f+1$ messages (or that instances have terminated), it must move on to the next step as it can never determine whether the remaining $f$ messages will be delivered or not, because in an asynchronous network, one can never distinguish if a message is delayed or lost. BFT consensus protocols in the partially asynchronous assumptions are summarised in Table.~\ref{tab:psync}.


\section{BFT in Wireless Environment}
\label{sect:BFT in wireless environment}
So far, nearly all the reviewed solutions imexplicitly assumed a wired connection, which could be seen from their message assumptions: for nearly all BFT (even CFT) solutions in different network environment, such as early years' Ben-or consensus \cite{ben1983another}, Paxos \cite{lamport2005generalized} and the most popular PBFT \cite{castro1999practical}, they all need a reliable channel, i.e., they assume a message could be delayed for arbitrary time through asynchronous network condition, but will eventually be delivered, and the safety property of these consensus exactly rely on this assumption. However, in wireless scenario, the network challenge is not only limited to the timing problem in asynchronous scenarios, but also face unique issues including limited resources, unreliable signals, unstable communication,  dynamic topology and various communication protocols \cite{xu2024wireless}, which makes the wireless BFT problem more complex, and existing BFT protocols might face significant challenges. In this section, we will analyse the challenges of deploying wireless BFT (Section \ref{sect:challenges and impact on wireless BFT}), and introduce recent works on analysing BFT resource (Section \ref{sect:wireless BFT resource}) and analysing the probabilistic BFT consensus result (Section \ref{sect:probabilistic wireless BFT}). Finally, we present current BFT solutions in wireless applications, including wireless blockchains (Section \ref{sect:BFT in wireless applications and wireless blockchain}).

\subsection{Challenges and Impact on Wireless BFT}
\label{sect:challenges and impact on wireless BFT}
\begin{itemize}
    \item \textbf{Unreliable Communications:} Wireless communication are susceptible to physical obstructions, frequent mobility and communication ranges, leading to frequent communication failure. The limited frequency resource increases the probability of data collisions and errors, further affecting communication reliability.
    \item \textbf{Resource and Energy Limits:} Wireless devices such as mobile devices and IoT sensors, are limited by resource and energy constrains, further restricting their capabilities and communication durations. Wireless devices cannot consistently participate in consistent message passing and data exchanges, leading to significant challenges of BFT implementation.
    \item \textbf{Dynamic Topology and Network:} Wireless devices may move position frequently, causing network topology changes. The dynamic of network environment raise challenges in maintaining stable and synchronous network, increasing consensus complexity.
    \item \textbf{Impact on BFT Primitives:} The unreliable communications introduce unique challenges that significantly impact the performance and reliability of the BFT primitive reliable broadcast (Definition \ref{def:reliablebroadcast}). The unstable network and communication introduce inconsistencies in message delivery, hindering  the effectiveness of traditional reliable broadcast primitive. Additionally, if an ideal reliable broadcast primitive cannot be built, i.e., unreliable links always exist, building atomic broadcast (a.k.a. total order broadcast, Definition \ref{def:atomic}) primitive is impossible \cite{xu2024wireless}.
\end{itemize}

\subsection{Wireless BFT Resource Usage and Optimisation}
\label{sect:wireless BFT resource}

While it is well established that BFT is a resource-sensitive solution, its demands can be even greater in wireless communication environments. To address this challenge, researchers \cite{xu2024wireless} have suggested optimisations such as hierarchical or multi-layered consensus and Gossip-based protocols, which have been analysed from a resource perspective in \cite{zhang2021much}. 
Zhang et al. \cite{zhang2021much} investigate the relationship between PBFT performance and communication resource requirements, highlighting that PBFT necessitates a total of $N + 2N^2$ communications, where $N$ is the total number of nodes. They further demonstrate that adopting a multi-layered approach can reduce communication complexity to $O(N^{3/4})$, aligning with the optimisations suggested in \cite{xu2024wireless}. 
Additionally, their study reveals that the spectrum requirement for PBFT is $2N+1$ and that higher transmission power is essential to maintain secure consensus, particularly as the number of faulty nodes increases. While greater communication resources (e.g., spectrum and transmission power) can improve throughput and reduce latency, PBFT’s scalability in wireless networks remains limited due to its quadratic growth in communication complexity, making it inherently less efficient for large-scale deployments.

\subsection{Probabilistic Wireless BFT}
\label{sect:probabilistic wireless BFT}
In this subsection, we introduce two kinds of works, one is analysing the reliability through probabilistic analysis, the other is bypassing the impossibility result on atomic broadcast with practical probability by a novel DAG-based data dissemination mechanism.

\subsubsection{Reliability Analysis}\ 

\noindent Reliability calculation is crucial in wireless scenarios because the extra uncertainty from wireless channels directly affects the performance of distributed consensus systems and ensuring a high consensus success rate is a primary challenge. Xu et al. \cite{xu2022wireless} proposed a probabilistic reliability analysis method for both wireless PBFT and Raft in their research. The core of this method is to model the reliability of nodes and links in wireless communication environments as random variables. This method decomposes the consensus process into multiple stages, and uses a chain probability model to derive a closed-form expression for the overall consensus success rate, thereby quantitatively characterizing the impact of a single node or link failure on the overall consensus reliability of the system. Another work by Li et al. \cite{li2023raft} also used a probabilistic model to analyze the reliability of the RAFT consensus protocol in a wireless environment. Especially, \cite{li2023raft} assumed more comprehensive failure conditions by employing a probability density function based computation method to capture the dynamics of node and link failures. In addition, this work also proposed performance indicators, Reliability Gain and Tolerance Gain, which provide theoretical guidance for fast estimation of the reliability of the consensus systems.

\subsubsection{Bypassing the Impossibility Result}\ 

\noindent Recognising the impossibility of achieving atomic broadcast in unreliable wireless communication channels, researchers suggested utilising probabilistic broadcast such as gossip-based protocols \cite{zhang2021much,xu2024wireless} and virtual communication protocols \cite{xu2024wireless} that combine multiple mechanisms such as probabilistic dissemination  mechanism (e.g., gossip) and hierarchical consensus.
Wu et al. \cite{wu2025when} designed a DAG-based message dissemination protocol with two-layered DAG-based hierarchical consensus. The protocol requires each message to carry part of the historic communication, and use the DAG to build the complete communication history, which achieves reliable broadcast with practical probability. Further, they build a two-dimension DAG structure, which orders the messages disseminated through the probabilistic reliable broadcast protocol to form an atomic broadcast protocol with practical possibility.

\subsection{BFT in Wireless Applications}
\label{sect:BFT in wireless applications and wireless blockchain}
Although facing significant challenges, wireless BFT is still developed and deployed in wireless applications. Feng et al. \cite{feng2023wireless} proposed a BFT consensus optimised for wireless vehicle to vehicle networks for autonomous driving. They pointed out the sensor reading in vehicles might be inaccurate, even incorrect, and the wireless channel is unreliable. To achieve consistent decision making, they utilise PBFT consensus as the core mechanism to achieve consistency among the whole network. For the impossibility result for total order broadcast, the authors add a gossip-based message dissemination phase to synchronise failure nodes to enhance the success rate. The multi-layered protocol designed by Wu et al. \cite{wu2025when} is specifically tailored to vehicular network at application layer, which support the decision making among vehicles, as well as the blockchain systems which need total ordering. Liu et al. \cite{liu2025partially} designed a wireless communication protocol at MAC layer to support reliable communications for seamless adaptation from wired to ad hoc wireless networks, further support deploying traditional wired consensus such as PBFT and HotStuff on wireless network. Cao et al. \cite{cao2024hierarchical} designed a hierarchical consensus framework that specific tailored to negotiation scheme for multi-platoon cooperative control. Consensus in this framework is divided into intra-platoon and inter-platoon, to achieve better scalability and lower delay. Additionally, they use digital twin to model vehicle position and predict vehicular network topology.
Xie et al. \cite{xie2024aircon} designed a over-the-air computation-based \cite{nazer2007computation} consensus protocol that enables users to transmit their hash symbols to the base station simultaneously using the same wireless channel, and use a hash consistency verification scheme in the physical layer, wherein each user determines whether its hash value is consistent or not with the aggregated hash symbols without decoding the signal. The hash symbols is therefore acting as a failure detector (Section \ref{Sect:FLP}) to reduce wireless resource usage and degrade the consensus complexity.

\section{Future BFT Research Issues}
\label{sect:futureBFT}
The BFT problem has been a long-standing research field since Lamport~\cite{lamport1978time} proposed it in the 1970s. A series of notions and solutions have been proposed since then , including SMR, reliable broadcast, and randomised consensus, which have built the cornerstone for distributed services. BFT has radiated energy again with the advance of blockchain, which, although it has more strict requirements, is essentially an SMR problem and could be built on BFT-SMR with optimisations and justifications. It also turns out that with the fast development of the Internet of Everything~\cite{kong2022edge}, networks are becoming more collaborative to support high-level applications, which generates demand for consensus protocols to provide a consistent and trustworthy cooperative platform. In this section, we introduce future research topics in BFT, including optimisations of the protocols themselves, adaptation to deployment on various applications, and improvements for future communication environments.



\subsection{Fundamental Research and Optimisations on the BFT Primitives}
\subsubsection{Fault Tolerance without Total Order}\ 

\noindent Recently proposed BFT solutions~\cite{miller2016honey,hotstuff2019yin, duan2018beat,guo2020dumbo} are total ordering protocols for SMR. Although SMR is the gold standard for implementing ideal functionality, it incurs higher overhead as in some cases, the cost of totally ordering all transactions and requests is unnecessary. It was found by Lamport~\cite{lamport2005generalized} that sometimes only a weaker problem than total ordering consensus needs to be solved. This finding has led to fundamental research like EPaxos~\cite{moraru2013there}. In addition, this finding also appears in cryptocurrency blockchain settings~\cite{lewenberg2015inclusive}. Further, it is also considered in general blockchain settings~\cite{karlsson2018vegvisir, wu2021me} that Conflict-free Replicated Data Types~\cite{shapiro2011conflict} have weaker requirements on ordering. Releasing the total ordering requirements could decrease system overhead. Thus, protocols based on weak ordering requirements, as well as methods to achieve weak ordering in current total ordering protocols deserve further research efforts.

\subsubsection{Scalability and Performance}\ 

\noindent Scalability of BFT is bounded by complexity, including communication complexity and running time complexity. Sharding technology and cryptography are used to reduce communication complexity (e.g.,~\cite{amir2006scaling} and \cite{hotstuff2019yin}, respectively), but sharding will alter the deployment structure and reduce the fault tolerance ability, while cryptography requires more powerful computation . The balance between computation and complexity should be researched.

In the early work on asynchronous MVBA, the protocols had a communication complexity of $O(n^2)$ for agreement on each message ($O(n^3)$ if there are $n$ peers). Recently, asynchronous BFT protocols utilize ACS to reach asynchronous agreement on a common core. HoneyBadgerBFT uses $n$-parallel ACS to reduce the communication complexity to $O(n^2)$ ($O(n)$ for each message). Dumbo reduces the number of ABA instances to optimise the time complexity and proposes to carefully use MVBA to further reduce time complexity. However, an open question remains: due to the $n$-parallel ABA instances, the time complexity of ACS is $O(log_n)$~\cite{guo2020dumbo} (in Dumbo-ACS, it is reduced to $O(log_{\kappa})$, $\kappa$ is a constant). However, we expect that the time complexity could be reduced to $O(1)$ in future work.

\subsubsection{Configuration and Reconfiguration}\ 

\noindent In nearly all of the recently proposed works~\cite{hotstuff2019yin, miller2016honey, duan2018beat,gao2022dumbo,guo2020dumbo, danezis2022narwhal, keidar2021all, spiegelman2022bullshark} that utilise cryptography, it is assumed that a reliable third party exists to initially distribute keys and configure the Public Key Infrastructure (PKI). However, this open problem needs to be solved in real implementations, as in some blockchain scenarios there are no trusted third parties to initialise the platform. In addition, the participants might change during the system's progress even after initial configuration. Therefore, future research should consider methods for initial configuration without trusted third parties and reconfiguration mechanisms during the system's progress, both of which should not significantly impact system performance.


\subsection{Deployment in Different Types of Blockchain Applications} It is known that most of the permissioned blockchains are based on BFT-SMR~\cite{wang2022bft}. This is because although BFT-SMR protocols provide a higher throughput, they generally less scalable, which generally only support up to several hundred of nodes. It is less decentralised due to the utilisation of PKI which is a centralised service to some degree. Nevertheless, there are still efforts to implement BFT in permissionless blockchains. These works integrate BFT with Proof-of-Stake (PoS) and PoW to resist Sybil attack~\cite{douceur2002sybil}. To launch Sybil attack, an attacker creates multiple identities in the blockchain to cause security issues. Sybil-resistance is unique in permissionless blockchains, and PoS, PoW are two widely used techniques to realise it. By integrating PoS/PoW together with BFT, a blockchain can enjoy both Sybil-resistance and high throughput. However, in permissionless blockchains, one can join or leave the system freely, which again leads to the configuration and reconfiguration research issue. Except for the permissioned/permissionless classification, blockchain can also be divided into different types according to different domains according to their use cases, such as financial, medical, governmental blockchains, etc. Different types of blockchain have different property requirements. For example, in financial applications such as cryptocurrencies~\cite{nakamoto2008bitcoin}, safety, decentralisation and fairness are strong demands but throughput could be balanced. In governmental~\cite{KASSEN2022101862} and medical blockchains~\cite{mamun2022blockchain,wu2021MB}, they usually take throughput and safety into account while decentralisation is a supplement, while in Internet of Things (IoT)~\cite{ayub2022internet,huo2022comprehensive,fernandez2018review} applications, an increasing number of heterogeneous devices are connected and their ability of communication is better than that of complex cryptography algorithms, which makes rethink of the schemes such as HotStuff that applies cryptography but reduces communication. In addition, IoT systems might be sensitive to latency, raising higher requirements on reducing time complexity. Different requirements lead to to a common belief that no one-size-fits-all BFT protocol exists. Research efforts are necessary on optimisations and adaptations on current protocols and even novel protocols that could satisfy different requirements.

\subsection{Deployment in Complex Application Structures}

The original protocols of BFT assume a flat system structure in which every participant has completely the same status. This includes the protocol design and the underlying peer-to-peer communication. This assumption holds for distributed service models and cryptocurrencies. However, in some other real applications, such as Internet of Things and Industrial Internet of Things~\cite{huo2022comprehensive,fernandez2018review}, healthcare~\cite{wu2021MB}, and autonomous driving~\cite{cao2022v2v,feng2023wireless}, the participants, for example, vehicles and Road Side Units (RSU), have different resources including computation power, bandwidth and may have different levels, i.e., leader-member relations that make the structure hierarchical. Reputation~\cite{zhang2023prestigebft} and sharding~\cite{kwak2020design} techniques have been used in partial BFT protocols to mitigate these issues, but there is still a lack of adaptations for asynchronous and DAG-based consensus in these contexts. In addition, research efforts are needed on how to carefully deploy consensus protocols that adapt to the system structure and yield better performance .

\subsection{BFT in Next Generation Communications}

With the rapid development of communication technology and the gradual maturity of 5G technology, researchers have various prospects for the development direction of the next generation communications. Among them, wireless distributed consensus based on the idea of distributed consensus algorithms has been proposed in some research \cite{zhang2021much}. Most of this idea is aimed at the design of communication system architectures for the Internet of Things (IoT) systems. Most current IoT systems adopt a centralised architecture, which requires a central controller to exchange data with other nodes. However, the increasing number of IoT nodes brings more challenges to the reliability and stability of the central node. In addition to the centralised architecture, implementing a distributed consensus network in wireless communication systems may be an alternative or supplementary solution \cite{zongyao2022design}. In a wireless distributed consensus system, the transmission and interaction of information no longer completely depends on the scheduling of the central node but instead realises reliable execution request transmission according to the consensus algorithm. Since the design of BFT algorithms leave a fault-tolerant space for Byzantine nodes, node failures and link failures in the wireless communication system can be recovered within a certain limit \cite{xu2022wireless}. Nevertheless, adapting and refining the original consensus algorithms to suit the specific requirements of a given communication system is paramount. As noted in \cite{xu2024wireless}, wireless BFT presents unique challenges, as the inherent unreliability of wireless communication may render traditional BFT primitives infeasible. This necessitates dedicated research efforts to develop novel primitives and to optimise existing wireless BFT across different wireless network layers to ensure performance and efficiency.

\subsection{BFT in Machine Learning and Artificial Intelligence}\ 

\noindent Machine Learning (ML) and Artificial Intelligence (AI) have emerged as rapidly evolving research areas with substantial recent advancements. The intersection of consensus mechanisms and ML/AI has begun to attract scholarly attention. For instance, Tessler et al. \cite{tessler2024ai} demonstrate that AI-assisted mediators can support humans in achieving consistency during democratic deliberation. In the context of Byzantine Fault Tolerance (BFT), Wu et al. \cite{wu2024bftbrain} propose BFTBrain, an adaptive framework in which ML techniques dynamically optimise network and hardware configurations to enable real-time switching of BFT protocols. These studies illustrate how ML/AI can enhance the efficiency and adaptability of consensus processes, either by accelerating the path to consistency or improving system performance.

Conversely, consensus mechanisms may also play a role in supporting AI-driven decision-making. Wu et al. \cite{wu2025when} explore the implications of intelligent agents participating in fault-tolerant consensus. Traditional BFT solutions rely on predefined rules and often assume uniform initial conditions and deterministic state transitions. However, in AI/ML-driven environments, agent behaviours may vary significantly due to differing levels of intelligence and diverse decision spaces, which do not necessarily conform to a deterministic state machine model. This divergence poses new challenges for the applicability and design of BFT in AI/ML scenarios.

In summary, two promising research directions emerge at the intersection of BFT and AI/ML: (1) leveraging AI/ML to enhance BFT protocols, such as through intelligent error detection and system reconfiguration; and (2) rethinking BFT design to accommodate the dynamic and non-deterministic nature of intelligent agents in AI/ML-driven systems.

\section{BFT in Blockchain and Blockchain Use Cases}
\label{sect:usecases}
Blockchains, also referred to as DLT, originated with Bitcoin~\cite{nakamoto2008bitcoin} and have since been explored and adopted across various domains. Their emergence has reignited significant research interest in BFT. Notably, many recent BFT solutions, such as HotStuff and Honey Badger BFT, have been specifically inspired by and designed for blockchain applications.
As the most popular killer application in blockchain, Bitcoin has driven the trend of research on cryptocurrency. With the development of cryptocurrency, in 2013, Ethereum~\cite{wood2014ethereum} was proposed, which provides a platform that supports any extendable application based on consensus and cryptocurrency. At this stage, researchers also abstracted blockchain from cryptocurrency as a decentralised, trustless platform or component that could be integrated with various information systems, including financial, healthcare, cloud computing, Internet of Things , supply chain, etc. With the power of blockchain, the trust, security, and privacy aspects of these information systems could be improved. In addition, the decentralised idea has gradually become a novel paradigm for brand new design models based on decentralisation and the crypto ecosystem. Building on the underlying blockchain infrastructure, Web3 , the next generation of the internet, is emerging.

In the preceding sections, we systematically reviewed BFT technology, which originated four decades before blockchain and has been extensively studied for nearly half a century, with recent research trends focusing on enhancing blockchain systems. In this section, we shift our view to post-blockchain use cases, including blockchain integrated with Web2 and general scenarios, and Web3 visions which use blockchain as the underlying protocol layer. Our goal in this section is not to provide a comprehensive review of as many blockchain systems as possible, nor to advocate any form of blockchain. Instead, we aim to present use cases and demonstrate how blockchain could be used in the real world and how blockchain may evolve in the future. The exploration of real-world applications and futuristic visions serves to illuminate the vital importance of sustained research on enduring BFT problems. By offering a panoramic overview of the distributed and decentralised system models before and after blockchain, we aspire to stimulate increased research that spans beneath (supporting blockchain) and atop (constructing upon blockchain) the blockchain paradigm.

\subsection{Web2 Applications and General Scenarios}

\subsubsection{Financial Applications}

The financial application Bitcoin and its variants are exactly how blockchains became famous as they offer a decentralised, trustless platform and environment for people to exchange value. Financial applications are still a major focus in blockchain research. Financial use cases include asset management, insurance, and payment. For example, RippleNet~\cite{ripple} is exploring financial trading and management on blockchain as well as new business models. Another example in the insurance industry is OpenIDL~\cite{openidl}, which is built on the IBM Hyperledger Fabric blockchain platform~\cite{hyperledger} that supports different consensus mechanisms. OpenIDL was initiated in 2020 by the American Association of Insurance Services (AAIS) and aims to solve data sharing, privacy, standardisation, and auditability challenges in the insurance industry. These two use cases show that financial applications can be built on both cryptocurrency and general blockchain platforms.

\subsubsection{Healthcare and Medical Services}

With the need for building secure and reliable data storage systems , research has been conducted to explore blockchain applications in healthcare and medicine . The intuitive idea is to use blockchain as a tamper-resistant platform to store data so that all historical data and operations are immutable, thus improving trust and safety. Numerous efforts have been made in this domain to build health records, medical data sharing and analysis platforms, medical supply chain records, etc~\cite{agbo2019blockchain}. In addition, recent research has also considered the future Internet of Medical Things (IoMT) \cite{QU2022942}. Moreover, blockchain in healthcare and medicine has also attracted government agencies. For example, the Centre for Surveillance Epidemiology and Laboratory Services of the Centres for Disease Control (CDC) built a proof of concept to track the opioid crisis. The U.S. Department of Health and Human Services (HHS) is also discussing the potential of blockchain in public and private healthcare\cite{hhs2021}.

\subsubsection{Supply Chain Management}
The blockchain naturally provides a tamper-proof, immutable, and transparent history, which is a perfect match for supply chains. The U.S. Food and Drug Administration (FDA) launched pilot programs to utilize blockchain as a solution to assist drug supply chain stakeholders with developing electronic, interoperable capabilities that will identify and trace certain prescription drugs as they are distributed within the United States~\cite{fda1,fda2}. Blockchain-empowered supply chains are also being researched in industry. For example, a fish provenance and quality tracking framework was proposed in \cite{wang2022blockchain}. The World Wildlife Fund (WWF) also launched a supply chain project, OpenSC, which is a global digital platform based on blockchain that tracks the journey of food along the supply chain.

\subsubsection{Cloud Computing}

Fusing blockchain with cloud computing has the potential to enhance security and privacy aspects. Blockchain has been utilised in enhancing traditional cloud access control. The major drawback of traditional cloud access control schemes is that they highly rely on centralised settings. Blockchain-based access control~\cite{sukhodolskiy2018blockchain,nguyen2019blockchain} has a few benefits, including traceability and immutable governance provided by the blockchain structure, and the involvement of all stake holders by the requirement of consensus. Blockchain is also considered in cloud resource allocation due to its economical nature; the tokenised incentive-based resource allocation could increase the resource sharing rate and reduce the energy cost~\cite{gai2020blockchain}. Some approaches have embedded blockchain to provide a secure and reliable environment for data sharing and collection in deep learning algorithms~\cite{liu2019blockchain}.

\subsubsection{Edge Computing and Internet of Things}

As an extension of cloud computing, edge computing~\cite{WU2022102720,yang2019integrated} aims to offload heavy centralised tasks from the cloud to decentralised edge servers. With the development of 5G, the Internet of Things, and the Internet of Vehicles (IoV), edge computing has gradually played an important role in these system models as the infrastructure to provide low latency computation and data processing, making it a key technology for realising various visions for the next-generation Internet. Having the same distributed architecture, edge computing is a perfect use case for blockchain applications. In addition to access control and resource allocation, which have already been considered in cloud computing, blockchain-assisted decentralised data storage for IoT devices~\cite{WU2022102720}, secure multi-party computation (MPC)\cite{guan2021blockchain}, and self-sovereign identity (SSI)\cite{basha2022fog} are also being researched.

\textbf{$\triangleright$ Remark.} Some of the financial applications that require a completely decentralised permissionless network use cryptocurrency consensus protocols such as PoW and PoS, while most of the general applications in a more controllable network consider using permissioned blockchains, which tend to deploy BFT consensus for much higher efficiency. Current BFT consensus protocols still have a long way to go to satisfy the increasing demands on throughput, security, and latency requirements in various blockchain applications.

\subsection{Web3 Visions}

\subsubsection{Layer1 and Layer2 Blockchain}

Layer1 and Layer2 blockchains are the fundamental layers for the web3 ecosystem. Layer1 blockchains, a.k.a. smart contract platforms, are the settlement layer that provides decentralisation, security, and decentralised programmability for the network through consensus, scripts, smart contracts, virtual machines, and other protocols. For instance, Bitcoin~\cite{nakamoto2008bitcoin} and Ethereum~\cite{wood2014ethereum} are layer1 blockchains. Layer2 blockchains are based on layer1, which could be regarded as off-chain protocols~\cite{decker2018eltoo}. Generally, Layer2 utilizes technologies including side chains and state channels to improve throughput without altering the basic layer1 blockchain rules. For example, Lightning Network~\cite{poon2015bitcoin} is a layer2 implementation on Bitcoin, and Polygon~\cite{polygon} is a layer2 solution on Ethereum. In addition to layer2 on permissionless blockchains, layer2 could also be built on permissioned blockchains~\cite{GANGWAL2023103539}. The different layers of the web3 ecosystem are illustrated in Fig.~\ref{fig:web3}.

\begin{figure}[ht]
 \centering
 \setlength{\abovecaptionskip}{0.1em}  %
 \setlength{\belowcaptionskip}{-0.3em} 
 \includegraphics[width=0.75\textwidth, trim=20 300 250 75, clip]{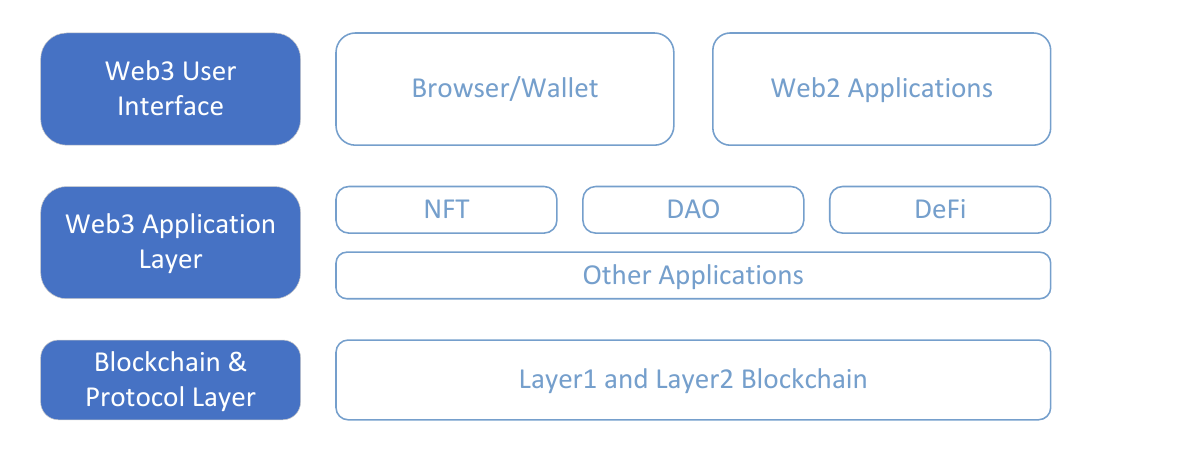}
 \caption{Illustration of Web3 Ecosystem}
 \Description{Web3 Ecosystem}
 \label{fig:web3}
\end{figure}

\subsubsection{NFT}
Non-fungible tokens (NFTs) are digital assets built on blockchain technology. Unlike cryptocurrencies such as Bitcoin, NFTs represent not only monetary value but also unique digital assets, meaning that each NFT is distinct from all others \cite{wang2021non}. The uniqueness and ownership confirmation functions realised by blockchain technology have enabled NFTs to be widely used in Web3. For example, since NFTs can be used to prove the ownership of digital content, artists can convert their works into NFTs, ensuring their uniqueness and authenticity. Users can obtain ownership of content by purchasing NFTs, and creators can break the shackles of traditional monopoly platforms and directly benefit from their works \cite{nadini2021mapping}. Another example is that NFTs can endow items, land, characters, and buildings in virtual reality environments with true value \cite{nadini2021mapping}. The virtual assets connected by NFTs can even be circulated across platforms, providing users with more practical autonomy.

\subsubsection{DAO}
Decentralised autonomous organisations (DAOs) have their origins rooted in Ethereum. They are internet-native communities/organisations that are executed through smart contracts and recorded by blockchain, and featured as open and transparent \cite{santana2022blockchain}. In contrast to traditional Internet organisations, being supported and influenced by blockchain technology, DAOs emphasise a flat personnel structure and a decentralised decision-making mechanism. They decentralise decision-making rights to all members as much as possible, write execution logic into smart contracts, and record all the critical activities on the chain \cite{fan2023insight}. The decentralisation of infrastructure and the decentralisation of management spirit in DAOs are complementary to each other. Blockchain technology provides a reliable infrastructure for decentralised autonomous organisations, enabling them to operate autonomously and make decisions without trusted intermediaries.

\subsubsection{DeFi}
Decentralised Finance (DeFi) is a newly emerged financial system built on blockchain technology. In contrast to traditional financial services, DeFi services are usually built on decentralised smart contract platforms for executing financial agreements and transactions, providing financial services that do not rely on traditional financial institutions and intermediaries. There are various financial products in the DeFi ecosystem, such as lending agreements, liquidity pools, decentralised exchanges, stablecoins, etc. \cite{werner2022sok}. These products can be realised through smart contracts, enabling users to participate in a more open and diversified financial market. The openness and programmability brought by blockchain to DeFi enables anyone to build and deploy their own financial applications, promoting innovation and competition \cite{chohan2021decentralized}. This provides users with more choices and drives the evolution of financial services.

\textbf{$\triangleright$ Remark.} Web2 and Web3 applications can be built upon both permissioned and permissionless blockchains. Web3 is considered to potentially change the current Internet service model, and Web3 applications usually assume an ideal underlying blockchain protocol layer and focus on the application layer. However, current blockchain consensus protocols still encounter challenges, and ongoing research into consensus protocols, including Byzantine Fault Tolerance (BFT) and hybrid approaches (e.g., BFT-based PoW), continues to be a focal point of interest.

\section{Conclusion}
\label{sect:conclusion}
This article provides a comprehensive review of BFT protocols for the long-term popular Byzantine fault-tolerance problems over half a century, from the original ones to state-of-the-art protocols that were designed for general applications before blockchain, such as distributed databases, to present a comprehensive view of the develop pathways of BFT research. In addition, we also review post-blockchain BFT protocols that are specifically designed for blockchain and cryptocurrency scenarios, and analyse the design principles of BFT in different scenarios with varied requirements. We introduce important concepts, properties, and technologies in BFT protocols, and elaborate on how they achieve these properties by employing specific designs and technologies in general applications and blockchain. After reviewing the principles and developments of BFT, we present future research instructions. We conclude by summarising blockchain applications, encompassing both real-world Web2 and general-purpose scenarios, as well as the rapidly evolving Web3 paradigms built on blockchain technology. These developments represent a potential next stage in BFT research, highlighting the transformative power of blockchain and reinforcing the critical role of continued BFT advancement in supporting future distributed systems.

This article presents a comprehensive survey of Byzantine Fault Tolerance (BFT) protocols, tracing their evolution over the past five decades. We begin by reviewing foundational and pre-blockchain-era BFT protocols developed for general-purpose distributed systems, such as databases, to provide a clear view of the developmental trajectories and underlying principles of BFT research. We then examine post-blockchain BFT protocols designed for various environments and requirements, especially the novel protocols specifically developed for blockchain and cryptocurrency applications, highlighting the varied design goals and performance requirements across different contexts. Key concepts, properties, and mechanisms of BFT protocols are introduced and analysed, with a focus on how these are realized through distinct architectural and algorithmic choices in both traditional and blockchain-based systems. Further, we outline promising directions for future research, including in the context of blockchain-enabled systems. 
Finally, we conclude by summarising blockchain applications across both real-world Web2 and general-purpose scenarios, as well as the rapidly evolving Web3 paradigms built upon blockchain technology. These emerging developments represent a potential next wave of breakthroughs and high-impact applications in BFT research, underscoring the transformative potential of blockchain and reaffirming the essential role of advancing BFT protocols in enabling the next generation of resilient and trustworthy distributed systems.

\bibliographystyle{ACM-Reference-Format}
\bibliography{ref}

\appendix

\end{document}